\documentclass[12pt]{article}
\usepackage[utf8]{inputenc}
\usepackage[T1]{fontenc} 
\usepackage[english]{babel} 
\usepackage{csquotes}
\usepackage{url}
\usepackage{geometry}
\usepackage{mathtools}
\usepackage{amsfonts}
\usepackage{amssymb}
\usepackage{amsthm}
\usepackage{mathrsfs}
\usepackage{esint}

\usepackage{xcolor}
\definecolor{myurlcolor}{rgb}{0,0,0.4}
\definecolor{mycitecolor}{rgb}{0,0.5,0}
\definecolor{myrefcolor}{rgb}{0.5,0,0}
\usepackage{graphicx}
\usepackage{tikz}
\usepackage{tikz-cd}

\usepackage{etaremune}

\usepackage[normalem]{ulem}

\usepackage{braket}

\usepackage{comment}

\usepackage{makeidx}

\usepackage{enumitem} 

\usepackage{caption}

\usepackage{lmodern}

\usepackage{microtype}

\usepackage{hyperref}
\hypersetup{colorlinks,linkcolor=myrefcolor,citecolor=mycitecolor,urlcolor=myurlcolor}
\usepackage[capitalize]{cleveref}
\usepackage[
backend=biber,
style=numeric-comp,
sorting=ynt, %year name title
sortcites=true, %sort inline citations according to sorting= irrespectively of how are written
giveninits=true, %abbreviated initial name
url=false,
backref=true,
maxbibnames=99 %no et al.
]{biblatex}
\DefineBibliographyStrings{english}{
  backrefpage = {$\downarrow$\,p.},
  backrefpages = {$\downarrow$\,pp.}
}
\AtEveryBibitem{\clearfield{note}}
\AtEveryBibitem{\clearfield{isbn}}
\AtEveryBibitem{\clearfield{issn}}
\AtEveryBibitem{\clearfield{abstract}}
\AtEveryBibitem{\clearfield{keywords}}
\AtEveryBibitem{%
	\clearfield{urlyear}%
	\clearfield{urlmonth}%
	\clearfield{urlday}%
	\clearfield{urldate}%
}

\newtheorem{remark}{Remark}

\newtheorem{theorem}{Theorem}
\newtheorem{corollary}{Corollary}
\newtheorem{proposition}{Proposition}
\newtheorem{definition}{Definition}
\newtheorem{example}{Example}
\newtheorem{lemma}{Lemma}

\newtheorem*{proof*}{Proof}

\title{Metric tensors and two-forms in information geometry from the GNS construction}

\author{M. Castrill\'on L\'opez$^{3,4}$\href{https://orcid.org/0000-0001-7673-9870}{\includegraphics[scale=0.7]{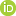}}, F. M. Ciaglia$^{1,5}$ \href{https://orcid.org/0000-0002-8987-1181}{\includegraphics[scale=0.7]{ORCID.png}}, \\ L. Gonz\'alez-Bravo$^{3,6}$ \href{https://orcid.org/0000-0002-4382-7978}{\includegraphics[scale=0.5]{ORCID.png}}, 
A. Ibort$^{1,2,7}$ \href{https://orcid.org/0000-0002-0580-5858}{\includegraphics[scale=0.7]{ORCID.png}}, }

\begin{document}

\maketitle 

\noindent
{\footnotesize $^{1}$  Universidad Carlos III de Madrid, ROR: \href{https://ror.org/03ths8210}{03ths8210}, Departamento de Matem\'aticas, Avenida de la Universidad 30 (edificio Sabatini), 28911 Legan\'es (Madrid), Espa\~na.}  \\
{{\footnotesize $^{2}$ Instituto de Ciencias Matem\'{a}ticas ICMAT (CSIC-UAM-UC3M-UCM), ROR: \href{https://ror.org/05e9bn444}{05e9bn444}, Campus de Cantoblanco UAM, Calle Nicol\'as Cabrera 13-15, 28049 Madrid, Espa\~na.} \\}
{\footnotesize $^{3}$ Universidad Complutense de Madrid, ROR: \href{https://ror.org/02p0gd045}{02p0gd045}. Departamento de \'Algebra, Geometr\'{\i}a y Topolog\'{\i}a, Facultad de Ciencias Matem\'aticas, Pl. de las Ciencias 3, Moncloa-Aravaca, 28040 Madrid, Espa\~na.} 

\bigskip
\noindent
{\scriptsize 
$^{4}$\texttt{mcastri[at]mat.ucm.es}
$^{5}$\texttt{fciaglia[at]math.uc3m.es}   $^{6}$\texttt{lauraego[at]ucm.es} $^{7}$\texttt{albertoi[at]math.uc3m.es}}

\begin{abstract}
We develop a GNS-based construction of geometric tensors on smooth parametric statistical models over $C^*$-algebras. 
Since the state space of a $C^*$-algebra is generally not a smooth manifold, the construction does not rely on pulling back tensors from an ambient state manifold. 
Instead, the GNS Hilbert spaces and their duals are organized into non-locally-trivial Hilbert fibrations over the state space. 
For models satisfying a compatibility condition expressing derivatives of expectation values as continuous functionals on the realified GNS fibers, each tangent vector admits a canonical dual GNS representative. 
Pulling back the dual GNS Hermitian product along the corresponding canonical lift produces a Hermitian tensor $K$ on the complexified tangent bundle of the model, whose real and imaginary parts define, under suitable regularity assumptions,
a smooth weak Riemannian metric tensor $G$ and a smooth two-form $\Omega$.
In finite-dimensional parameter manifolds the metric is, of course, strong.

The construction recovers the Fisher--Rao metric in the commutative dominated case, the Fubini--Study geometry for pure states up to the normalization and sign convention imposed by the dual GNS pairing, and the SLD metric for faithful quantum states. 
In finite-dimensional faithful models, the two-form $\Omega$ is proportional, up to convention, to the expected commutator of the SLD representatives, equivalently to the mean Uhlmann curvature. 
We show through faithful qubits and displaced thermal states that $\Omega$ need not be closed. 
For bundle-regular models, the associated fiberwise symplectic form on the real dual GNS bundle admits connection-dependent closed extensions to the total space, while closedness of $\Omega$ on the parameter manifold is controlled by the covariant exterior derivative of the canonical real dual GNS lift.
\end{abstract}

\tableofcontents

%%%%%%%%
%%%%%%%%

\section{Introduction}\label{sec:intro}

Statistical models play a central role whenever one wishes to describe, compare, and infer from incomplete information.
%%%%%%%%%%%%%%%%%%%%%%%%%%%%%
In practice, observations are finite, noisy, and often indirect, and one works with structured families of admissible states depending on a finite number of parameters.
%%%%%%%%%%%%%%%%%%%%%%%%%%%%%%%%%%
Such models make it possible to perform estimation, prediction, uncertainty quantification, hypothesis testing, and experimental design.
In physics they encode families of accessible states of a system; in statistics they organize the passage from data to inference; in information theory they provide a setting in which one can quantify how distinguishable nearby states are.
For this reason, understanding the geometry naturally carried by statistical models is not merely a formal matter, but one that is closely tied to the operational content of inference and estimation.

A particularly fruitful way to study statistical models is through differential geometry.
Once a smooth manifold whose points parametrize the relevant states is selected, classical and quantum information geometry provide systematic ways to endow the parameter manifold with a Riemannian metric \cite{AN2000,BZ2006,BC1994,P1996,C1981a}, obtaining a local notion of statistical distance, typically related to the distinguishability of nearby points on the model.
At the same time, the geometric structures arising in these two settings seem strikingly different, and understanding both their differences and their possible points of contact is one of the central themes behind this work.

The classical case is governed by the geometry of the Fisher--Rao metric tensor \cite{F1922,M1936,C1981a}.
%%%%%%%%%%%%%%%%%%%%%%%%%%%%%%%%%%%%%%%%%
On smooth manifolds of strictly positive probability measures on a discrete and finite outcome space, the Fisher--Rao metric tensor is not merely one possible choice among many.
%%%%%%%%%%%%%%%%%%%%%%%%%%%%%%%%
In his seminal work, \v{C}encov introduced a categorical framework in which he showed that this metric is uniquely characterized by its invariance under congruent embeddings induced by Markov morphisms \cite{C1981a}.
%%%%%%%%%%%%%%%%%%%%%%%%%%%%%%%%%%%%%%%
This celebrated theorem shows that, for finite-dimensional classical systems, the probabilistic structure itself strongly constrains the available geometry.
%%%%%%%%%%%%%%%%%%%%%%%%%%%%%%%%%%%%%%%%%%%%%%%
Related uniqueness results have been extended to broader infinite-dimensional settings \cite{BBM2016,AJLS2017,F2022}, further reinforcing the distinguished role of the Fisher--Rao metric tensor in the classical case.
%%%%%%%%%%%%%%%%%%%%%%%%%%%%%%%%%%%

The quantum case, by contrast, exhibits a richer and more varied geometric landscape.
%%%%%%%%%%%%%%%%%%%%%%%%%%
On manifolds of positive-definite quantum states over a finite-dimensional Hilbert space $\mathcal{H}_{n}$, the uniqueness of the metric is lost.
%%%%%%%%%%%%%%%%%%%%%%%%%%%%%
Instead of a single Riemannian geometry as in the classical case, Morozova and \v{C}encov understood that quantum Markov morphisms are not enough to grant uniqueness \cite{MC1991}, and Petz classified the family of admissible Riemannian geometries through a family of symmetric and normalized operator monotone functions \cite{P1996}.
%%%%%%%%%%%%%%%%%%%%%%%%%%%%%%%%%%%%%%
All these Riemannian metric tensors are invariant under the action of the unitary group on quantum states, and monotone with respect to quantum Markov morphisms, \textit{i.e.}, completely-positive and trace-preserving (CPTP) maps. 
%%%%%%%%%%%%%%%%%%%%%%%%%%%%%%%%%%%%%%%
Among them, \emph{Symmetric Logarithmic Derivative} (SLD) metric tensor is closely related to the Fisher information and to the quantum Cram\'er--Rao bound for quantum parameter estimation \cite{BC1994,BGJ2003,H1967a,H2011,P2009a,S2019}.
%%%%%%%%%%%%%%%%%%%%%%%%%%%%%%%%%%%%%%%%%%%%%%%%%%

A peculiar place is occupied by the Fubini--Study metric tensor on pure quantum states, which is the unique (up to a positive constant) Riemannian metric tensor invariant under the unitary group, both in finite and infinite dimensions \cite{BZ2006,CMP1990}.
%%%%%%%%%%%%%%%%%%%%%%%%%%%%%%%%%%%%%%%%%%%%
This uniqueness is reminiscent of the classical case \cite{FKMMSV2010}, but appears in a completely different mathematical landscape, namely, that of Hilbert spaces.
%%%%%%%%%%%%%%%%%%%%%%%%%%%%%%%%%%%%%%%%%

Indeed, classical and quantum information geometry are traditionally developed with rather different mathematical languages.
%%%%%%%%%%%%%%%%%%%%%%%%%%%%%%%%%%%%%%%%
Classical information geometry is rooted in probability and measure theory, where systems are described using sample spaces, random variables, and probability distributions.
%%%%%%%%%%%%%%%%%%%%%%%%%%%%%%%
Quantum information geometry, on the other hand, is rooted in operator theory on Hilbert spaces, where systems are described using Hilbert spaces, quantum observables (self-adjoint operators), and quantum states (density operators).
%%%%%%%%%%%%%%%%%%%%%%%%%%%%%%%%%%
Despite these differences, both theories are concerned with families of states and with their geometric structures encoding operational distinguishability.
%%%%%%%%%%%%%%%%%%%%%%%%%%%%%%%%%%%%%%%%%%%%%%
This suggests that the contrast between the classical and quantum settings need not prevent a common treatment.
%%%%%%%%%%%%%%%%%%%%%%%%%%%%%%%%%%%%%%%
It is already well known how \emph{positive-operator-valued measures} (POVMs) allow one to determine a classical statistical model from a quantum one \cite{BGJ2003,BC1994,P2009}, but POVMs provide a bridge between the two formalisms rather than a unifying mathematical framework. 
%%%%%%%%%%%%%%%%%%%%%%%%%%%%%%%

On the other hand, the language of $C^*$-algebras provides a natural setting for classical and quantum information geometry to be formulated ``simultaneously'' \cite{CDJS2024,CJS2020,CJS2020a,J2006,K2011,K2014a,K2016}.
%%%%%%%%%%%%%%%%%%%%%%%%%%%%%%%%%%%%%%%%%%%
In this setting, the focus is shifted from the classical outcome space or the quantum Hilbert space to a $C^*$-algebra $\mathscr{A}$, that is, a complex Banach algebra with a continuous anti-linear involution $\dagger$ that satisfies the so-called $C^*$ identity $\Vert \mathbf{x}^{\dagger}\,\mathbf{x}\Vert = \Vert \mathbf{x}\Vert^{2}$.
%%%%%%%%%%%%%%%%%%%%%%%%%%%%%%%%%%%%%%
Probability measures and quantum states are recovered as states on the algebra $\mathscr{A}$, that is, as bounded linear functionals in $\mathscr{A}^*$ such that $\rho(\mathbf{x}^{\dagger}\mathbf{x})\geq 0$ and $\Vert\rho\Vert=1$.
%%%%%%%%%%%%%%%%%%%%%%%%%%%%%%%%%%%%%%%%%%%%%%%
The space of all states of $\mathscr{A}$ is denoted by $\mathcal{S}(\mathscr{A})$ and it is a convex set that is compact in the weak*-topology of $\mathscr{A}^*$ whenever $\mathscr{A}$ has an identity element.
%%%%%%%%%%%%%%%%%%%%%%%%%%%%%%%%%%%%%%%%%%%
The classical case may be recovered using commutative algebras like $\mathcal{L}^{\infty}(\Omega,\mu)$ where $(\Omega,\mu)$ is a $\sigma$-finite measure space, or $C_{0}(X)$ where $X$ is a locally-compact Hausdorff space, in which cases $\mu$-absolutely continuous probability measures on $(\Omega,\mu)$ and Radon measures on $X$ are states in the $C^*$-algebraic sense.
%%%%%%%%%%%%%%%%%%%%%%%%%%%%%%%%%%%
Standard quantum mechanical systems are recovered when $\mathscr{A}=\mathcal{B}(\mathcal{H})$, in which case density operators are states in the $C^*$-algebraic sense.
%%%%%%%%%%%%%%%%%%%%%%%%%%%%%%%%%%%%%%%%%%

In particular, both classical and quantum parametric (statistical) models may be described by a triple $(M,\mathrm{i},\mathscr{A})$ where $M$ is a real smooth manifold, $\mathscr{A}$ is a $C^{*}$-algebra, and $\mathrm{i}\colon M\rightarrow\mathcal{S}(\mathscr{A})$ is the map determining the model \cite{CDJS2024}.
%%%%%%%%%%%%%%%%%%%%%%
From this perspective, one is led to ask whether apparently different
geometric structures such as the Fisher--Rao, SLD, and Fubini--Study
tensors may be understood as different realizations of a common background
construction. A partial positive answer, valid in finite dimensions, is
given in \cite{CJS2020,CJJS2025}, where these tensors are obtained from a
coadjoint-orbit-like construction for the Jordan product of self-adjoint
elements. That construction, however, presupposes a smooth ambient
manifold of states and does not extend beyond finite dimensions.

The present paper differs from that approach in three structural respects.
First, it replaces the ambient-manifold picture by the GNS fibration, so
the construction never requires $\mathcal S(\mathscr A)$ to be a manifold
and remains meaningful when the state space is stratified or has corners.
Second, it operates in genuine infinite-dimensional examples: for faithful
states the relevant representative is the Hilbert--Schmidt amplitude rather
than the possibly unbounded SLD operator, and the compatibility condition
becomes an explicit Hilbert--Schmidt integrability requirement, verified
below for displaced thermal states. Third, alongside the metric tensor it
naturally produces a skew-symmetric two-form $\Omega$, whose closedness is
analyzed structurally in Section~\ref{sec:closed-extensions}.

It is worth noting that the problem of classifying admissible Riemannian geometries on statistical models can be recast in a categorical language reminiscent of \v{C}encov's original works \cite{C1981a,MC1991}.
%%%%%%%%%%%%%%%%%%%%%%%%%%%%%%%%%%%%%%%%
In recent work \cite{CDG2025}, the notion of a \emph{field of covariances} has been introduced as a contravariant functor from a category of non-commutative probability spaces to Hilbert spaces.
%%%%%%%%%%%%%%%%%%%%%%%%%%%%%%
Such functors generalize the classical statistical covariance, and their classification in finite dimensions recovers both \v{C}encov's uniqueness theorem and the Morozova--\v{C}encov--Petz classification of monotone quantum metrics as particular instances, providing also an extension to non-faithful states that are not pure generalizing the radial procedure in \cite{PS1996}.
%%%%%%%%%%%%%%%%%%%%%%%%%%%%%%%%%%%%%%%%%%%%
Among all fields of covariances, the one arising from the \emph{Gelfand--Naimark--Segal} (GNS) construction is distinguished: it associates to each state $\rho$ the GNS Hilbert space $\mathcal{H}_\rho$ with its canonical inner product $\langle\cdot,\cdot\rangle_\rho$, and it corresponds, in the Morozova--\v{C}encov--Petz classification, to the SLD metric.
%%%%%%%%%%%%%%%%%%%%%%%%%%%%%%%%%%%%%%%%
The present paper may be understood as a detailed geometric study of this particular field of covariances from a dual perspective: we develop the pullback construction through which the dual GNS inner product induces geometric tensors on parametric models, and we work out its consequences in the classical, pure-state, and faithful quantum settings.
%%%%%%%%%%%%%%%%%%%%%%%%%%%%%%%%%%%%%%%%%%%%%%%
The extension of this programme to other fields of covariances is a natural direction for future work.
%%%%%%%%%%%%%%%%%%%%%%%%%%%%%%%

The paper provides a GNS-based unifying mechanism for several central geometries of classical and quantum information geometry  that is valid also for parametric models over infinite-dimensional $C^{*}$-algebras,   when the regularity assumptions are verified.
%%%%%%%%%%%%%%%%%%%%%%%%%%%%%%%%%%%%%%
Our main idea is that familiar geometric structures in classical and quantum information geometry can be traced back to a single mathematical object, namely, the (dual of the) GNS inner product.
%%%%%%%%%%%%%%%%%%%%%%%%%%%%%%%%%%%%
Seen from this angle, the metrics mentioned above do not appear as isolated constructions, but rather as different geometric manifestations of the same operator-algebraic background structure.
%%%%%%%%%%%%%%%%%%%%%%%%%%%%%%
The idea is very close to the standard construction of pull-back metrics in differential geometry, where a Riemannian metric tensor is pulled back along a smooth immersion.
%%%%%%%%%%%%%%%%%%%%%%%%%%%%%%%%%%%%%%%%%%%

However, a basic obstacle appears immediately.
The state space $\mathcal S(\mathscr{A})$ is not a smooth manifold.
%%%%%%%%%%%%%%%%%%%%%%%%%%%%%%%%%%%%
Even in finite dimensions, it is either a manifold with boundary when $\mathscr{A}\cong \mathbb{C}^{2}$ or $\mathscr{A}\cong \mathcal{M}_{2}(\mathbb{C})\cong\mathcal{B}(\mathbb{C}^{2})$, or a manifold with corners when $\mathscr{A}\cong \mathbb{C}^{n}$ with $n>2$, or a stratified space whose strata are smooth manifolds \cite{DF2021,GKM2005,CCIMV2019}.
%%%%%%%%%%%%%%%%%%%%%%%%%%%%%%%
Because of this, the usual differential-geometric procedure of pulling back covariant tensor fields is not directly available.
This leads us to look for another geometric object attached to each state and varying coherently enough across the state space to support a pullback construction along statistical models.

The GNS construction suggests such an object.
To each state $\rho\in\mathcal{S}(\mathscr{A})$ it associates a Hilbert space $\mathcal{H}_\rho$ and its dual $\mathcal{H}_{\rho}^{*}$.
%%%%%%%%%%%%%%%%%%%%%%%%%%%%%%
These Hilbert spaces glue together to give rise to a non-locally-trivial Hilbert fibration in the sense of \cite[\S II.13.14]{FD1988} (see Proposition~\ref{prop:GNS-nonlocallytrivial-fibration}).
%%%%%%%%%%%%%%%%%%%%%%%%%%%%

This point of view gives us a geometric environment that remains workable even though the state space is not a manifold.
If a statistical model $(M,\mathrm{i},\mathscr{A})$ satisfies additional regularity conditions (see Definitions~\ref{def:gns-smooth-parametric-model} and \ref{def:regular model}), we can pull back the real and imaginary parts of the dual \emph{GNS}
Hermitian structure to define a smooth weak Riemannian metric tensor $G$ and
a smooth $2$-form $\Omega$ on $M$.  When $M$ is finite-dimensional, the weak
Riemannian metric is automatically strong.
%%%%%%%%%%%%%%%%%%%%%%%%%%%%%%%%%%

A relevant feature of this construction is that it recovers, within the same general mechanism, several of the fundamental geometries of information theory, as explained in sections~\ref{sec:classical-case}, \ref{sec:quantum-pure-states}, and~\ref{sec:quantum-faithful-states}.
%%%%%%%%%%%%%%%%%%%%%%%%%%%%%%%%%%%
In the classical commutative case, the regularity of the model is related to the existence of an $\mathcal{L}^{2}$-integrable score function, and the induced metric tensor is the Fisher--Rao metric tensor, while the $2$-form identically vanishes.
%%%%%%%%%%%%%%%%%%%%%%%%%%%%%%%%%%%%%%%%
For pure quantum states, the construction yields  the Fubini--Study metric tensor together with  its symplectic form up to normalization and sign/order convention.
%%%%%%%%%%%%%%%%%%%%%%%%%%%%%%%
For faithful states on $\mathcal B(\mathcal{H})$, the regularity of the model is related to the existence of a Hilbert--Schmidt amplitude associated with the weak SLD equation, which in finite dimensions reduces to the usual SLD and the induced metric is the SLD metric (equivalently, four times the Bures-Helstrom metric)  \cite{H1967a,P2009a}.
%%%%%%%%%%%%%%%%%%%%%%%%%%%%%%%%
Moreover, for finite-dimensional faithful models, the two-form is proportional to the expected commutator of the SLD representatives, also known as \emph{mean Uhlmann curvature} \cite{RJD2016,CSDV2019,CSV2018,CSV2018a,LVSC2019,BLVSC2019} (see \ref{rem:mean-uhlmann-curvature}). 
%%%%%%%%%%%%%%%%%%%%%%%%%
In particular, models with pairwise commuting
SLDs are isotropic for $\Omega$, although the converse need not hold (see remark \ref{rem:antisymmetric-part-SLD-commutators}).
%%%%%%%%%%%%%%%%%%%%%%%%%%%%%%%%%%%%%%%%%%%%%

The behaviour of the two-form $\Omega$ is sensitive to the model. 
%%%%%%%%%%%%%%%%%%%%%%%%%%%%%%%%%%
It vanishes in the commutative case (see Section \ref{sec:classical-case}); it is closed (and symplectic) in the case of pure quantum states (see \ref{sec:quantum-pure-states}); but there are models of faithful quantum states in which it is not closed (see Section \ref{sec:quantum-faithful-states}). 
%%%%%%%%%%%%%%%%%%%%%%%%%%%%%%%%%%%%%%%
In section~\ref{sec:closed-extensions}, we give a structural explanation of this phenomenon under additional regularity assumptions on the model (see Definition \ref{def:bundle-regular-model}).
%%%%%%%%%%%%%%%%%%%%%%%%%%
The idea is to lift the analysis to the total space of the realified dual GNS bundle, where the fiberwise symplectic form
on the real dual GNS bundle admits connection-dependent closed extensions to the total space, in analogy with the finite-dimensional theory of closed forms on symplectic fibrations \cite{GLSW1983}.
%%%%%%%%%%%%%%%%%%%%%%%%%%%%%%%
The closedness of the two-form on the model is controlled by the covariant exterior derivative of the lift used to define the two-form.
%%%%%%%%%%%%%%%%%%

The paper is organized as follows.
%%%%%%%%%%%%%%%%%%%%%%%%%%%%%%%%%%%%%%%%%%%%%%%%%%%%%%%%%%
Section~\ref{sec: gns fibration} recalls the GNS construction and introduces the GNS fibration together with its dual.
%%%%%%%%%%%%%%%%%%%%%%%%%%%%%%%%%%%%%%%%
Section~\ref{sec: parametric models and pullbacks} defines GNS-smooth and regular parametric models, constructs the real GNS lift, the real dual lift, and the induced Hermitian tensor $K$, with real and imaginary parts $G$ and $\Omega$.
%%%%%%%%%%%%%%%%%%%%%%%%%%%%%%%%%%
Sections~\ref{sec:classical-case}, \ref{sec:quantum-pure-states}, and~\ref{sec:quantum-faithful-states} analyze the main examples: commutative dominated models, pure states, and faithful quantum states, including faithful qubits and displaced thermal states.
%%%%%%%%%%%%%%%%%%%%%%%%%%%%%%%%%%%%%%%%%
Finally, Section~\ref{sec:closed-extensions} studies the closedness problem for $\Omega$ by introducing the stronger bundle-regularity condition and by separating closed extensions on the real dual GNS bundle from closedness of the induced two-form on the parameter manifold.
%%%%%%%%%%%%%%%%%%%%%%%%%%%%%%%%%%%%%%%%%%%%%

\section{The \emph{GNS} fibration and its dual}\label{sec: gns fibration}

We now build two \emph{non-locally trivial Hilbert fibrations} over the space of states $\mathcal{S}(\mathscr{A})$ in the sense of \cite[II.13 and II.14]{FD1988}\footnote{Differently from \cite{FD1988}, we use the term \emph{fibration} rather than \emph{bundle} to avoid suggesting local triviality that is often associated to the word bundle in differential geometric contexts.}.
The starting point  is the \emph{GNS} construction associated with a state we now briefly recall.

Let $\mathscr{A}$ be a $C^{*}$-algebra and $\rho: \mathscr{A}\to \mathbb{C}$ a state on $\mathscr{A}$.
The Gelfand-Naimark-Segal (GNS) construction associates to the pair $(\mathscr{A},\rho)$ a  triple $(\mathcal{H}_\rho,\pi_{\rho}, \Omega^{\rho})$, where $\mathcal{H}_\rho$ is a complex Hilbert space called the \emph{GNS} Hilbert space of $\rho$, $\pi_{\rho}$ is a $*$-homomorphism of $\mathscr{A}$ in $\mathcal{B}(\mathcal{H}_\rho)$ called the \emph{GNS} representation of $\mathscr{A}$ associated with $\rho$, and $\Omega^{\rho}\in\mathcal{H}_\rho$ is a vector that is cyclic with respect to $\pi_{\rho}$ (which means that $\mathrm{span}\left\{\pi_{\rho}(a)\Omega^{\rho}\mid a\in\mathscr{A}\right\}$ is dense in $\mathcal{H}_\rho$).
Specifically, the \emph{GNS} Hilbert space of $\rho$ is defined as
$$
\mathcal{H}_\rho := \overline{\mathscr{A} / \mathcal N_\rho},
$$
where
\begin{equation}\label{eqn:gelfand-ideal} 
\mathcal N_\rho := \{ a \in \mathscr{A} \mid \rho(a^*a)=0 \}
\end{equation}
is the so-called \textit{GNS ideal} of $\rho$, and the closure is taken with respect to the inner product 
$$
\langle [a]_\rho , [b]_\rho \rangle_\rho := \rho(a^*b),
$$
which is well defined because it does not depend on the representative of the classes. 
We denote by $\psi_a^\rho$ the equivalence class $[a]_\rho$ when immersed in the closure $\mathcal{H}_\rho$.
Generic elements of $\mathcal{H}_\rho$ will be denoted by $\psi^{\rho} , \eta^{\rho} ,\xi^{\rho} $, etc.
The \emph{GNS} representation is defined by
\begin{equation}
\pi_\rho(a)\psi_{b}^{\rho} := \psi_{ab}^\rho,
\end{equation}
on $\mathscr{A}/\mathcal{N}_{\rho}$, and then extended by continuity.
In particular, if $\mathscr A$ has an identity element $\mathbb I$, then
the cyclic vector $\Omega^\rho$ is just $\psi^\rho_{\mathbb I}$.  If
$\mathscr A$ is nonunital, one may either pass to the unitization, or,
equivalently, obtain the cyclic vector as the Hilbert-space limit of
$\psi^\rho_{e_\lambda}$ for any approximate identity $(e_\lambda)_\lambda$
of $\mathscr A$; the resulting vector is independent of the chosen
approximate identity.  In the present paper, however, the construction of
the GNS fibration only uses the vectors $\psi^\rho_a$ with
$a\in\mathscr A$.

\begin{remark}
We write $\mathcal H_\rho^{\mathbb R}$ for the realification of $\mathcal{H}_{\rho}$.
%%%%%%%%%%%
This is  the same underlying additive group as $\mathcal H_\rho$, but scalar multiplication is restricted to
real scalars.
%%%%%%%%%%%%%%%%%%%%%%%%%

The real and imaginary parts of the Hermitian product define two canonical real-bilinear forms on $\mathcal H_\rho^{\mathbb R}$. 
%%%%%%%%%%%%%%%%%%%%
Namely, for $\psi^\rho,\eta^\rho\in\mathcal H_\rho$, regarded also as elements of
$\mathcal H_\rho^{\mathbb R}$, set
\begin{equation}\label{eqn:real-hilbert-product-realification-gns}
(\psi^\rho,\eta^\rho)_\rho
:=
\Re\langle \psi^\rho,\eta^\rho\rangle_\rho,
\end{equation}
and
\begin{equation}\label{eqn:symplectic-form-realification-gns}
[\psi^\rho,\eta^\rho]_\rho
:=
\Im\langle \psi^\rho,\eta^\rho\rangle_\rho.
\end{equation}
%%%%%%%%%%%%%%%%%%%%%%%%%
Then $(\cdot,\cdot)_\rho$ is symmetric and positive definite, and it turns 
$\mathcal H_\rho^{\mathbb R}$ into a real Hilbert space. The form $[\cdot,\cdot]_\rho$ is real-bilinear and skew-symmetric.
%%%%%%%%%%%%%%%%%%%

The original multiplication by $\mathrm i$ on the complex Hilbert space $\mathcal H_\rho$ defines a real-linear operator $J_\rho:\mathcal H_\rho^{\mathbb R}\to
\mathcal H_\rho^{\mathbb R}$ by
\begin{equation}\label{eqn:complex-structure-realification-gns}
J_\rho(\psi^\rho):=\mathrm i\psi^\rho,
\end{equation}
where the product $\mathrm i\psi^\rho$ is taken in $\mathcal H_\rho$ and the
result is then regarded as an element of $\mathcal H_\rho^{\mathbb R}$. 
%%%%%%%%%%%%%%%%%%
Thus $J_\rho^2=-\mathrm{id}$.

With the convention that the GNS Hermitian product is linear in the second
argument, these structures satisfy
\begin{equation}\label{eqn:compatibility-symmetric-skew-symmetric-form-realification-gns}
[\psi^\rho,\eta^\rho]_\rho
=
(J_\rho\psi^\rho,\eta^\rho)_\rho
=
-(\psi^\rho,J_\rho\eta^\rho)_\rho
\end{equation}
for all $\psi^\rho,\eta^\rho\in\mathcal H_\rho$.
\end{remark}

\begin{definition}[{\cite[\S II.13.4]{FD1988}}] \label{def: Banach and Hilbert fibration}
Let $X$ be a Hausdorff topological space.
A \textbf{fibration} over $X$ is a triple $(\mathcal E,\pi,X)$ where $\mathcal E$ is a topological space, and $\pi:\mathcal E\to X$ is a continuous open surjection.
For each $x\in X$, we write $\mathcal E_x:=\pi^{-1}(x)$ for the fiber over $x$.
A \emph{section} of $(\mathcal E,\pi,X)$ is a map $s:X\to \mathcal E$ such that
$\pi\circ s=\mathrm{id}_X.$

A \textbf{\emph{Banach fibration}} over $X$ is a \textbf{fibration} over $X$  such that each fiber $\mathcal E_x$ ($x\in X$) is a Banach space,  and the following conditions hold:

\begin{enumerate}
    \item The norm map $ \mathcal E\to \mathbb R$ given by $\xi\mapsto \|\xi\|_{\mathcal E_{\pi(\xi)}}$ is continuous;
    \item the addition map $+:\mathcal E\times_X\mathcal E\to \mathcal E$ given by $(\xi,\eta)\mapsto \xi+\eta$, where 
    $$
    \mathcal E\times_X\mathcal E
    :=
    \{(\xi,\eta)\in\mathcal E\times\mathcal E:\pi(\xi)=\pi(\eta)\},
    $$
    is continuous;
    \item the scalar multiplication map $\mathbb{C}\times \mathcal E \to \mathcal E$ given by $(\lambda,\xi) \mapsto \lambda \xi$ is continuous;    
    \item if $x\in X$ and $\{\xi_i\}$ is a net in $\mathcal E$ such that $\|\xi_i\|\to 0$ and $\pi(\xi_i)\to x$,   then $\xi_i\to 0_x$ in $\mathcal E$, where $0_x$ denotes the zero vector in the fiber $\mathcal E_x$.
\end{enumerate}
A Banach fibration whose fibers are all Hilbert spaces is called a \textbf{\emph{Hilbert fibration}}.
\end{definition}

Sometimes, one practically starts from a set-theoretic surjection $p\colon A\rightarrow X$ for which each fiber is a Banach space, and would like to build a Banach fibration over $X$ in the sense of Definition \ref{def: Banach and Hilbert fibration}.
It turns out such a procedure is possible once a suitably regular set of sections is introduced.

\begin{theorem}[Construction theorem for Banach fibrations  {\cite[\S II.13.18]{FD1988}}]\label{thm: construction thm}
Let $X$ be a Hausdorff topological space, let $A$ be a untopologized set, and let $p:A\to X$ be a surjection such that $A_x:=p^{-1}(x)$  is a Banach space for each $x\in X$.
Assume there is a complex vector space $\Gamma$ of (set-theoretical) sections of $p\colon A\rightarrow X$ satisfying the following conditions:
\begin{enumerate}
\item for every $\sigma\in\Gamma$, the function $x\longmapsto \|\sigma(x)\|$ is continuous on $X$;
\item for every $x\in X$, the set $\{\sigma(x)\mid\sigma\in\Gamma\}$ is dense in $A_x$.
\end{enumerate}
Then there exists a unique topology on $A$ such that, with this topology, the triple $(A,p,X)$ is a Banach fibration and every section $\sigma\in\Gamma$ is continuous.
\end{theorem}

We now exploit Theorem \ref{thm: construction thm} to define the  \emph{GNS} fibration over $\mathcal{S}(\mathscr{A})$ and its dual.
We start with the set-theoretical disjoint union
\begin{equation}\label{eq:GNS fibration}
\mathcal{H}(\mathscr{A}):=\bigsqcup_{\rho\in\mathscr S(\mathscr{A})} \mathcal{H}_\rho,
\end{equation}
equipped with the natural projection
\begin{equation}\label{eq:set theoretical projection}
p:\mathcal{H}(\mathscr{A})\longrightarrow \mathscr S(\mathscr{A}),
\qquad
p(\psi^{\rho} ):=\rho,
\end{equation}
and its dual 
\begin{equation}\label{eq:dual GNS fibration}
\mathcal{H}^\ast(\mathscr{A}):=\bigsqcup_{\rho\in\mathscr S(\mathscr{A})} \mathcal{H}_\rho^\ast,
\end{equation}
with the projection
\begin{equation}\label{eq:set theoretical dual projection}
p^{*}:\mathcal{H}^{*}(\mathscr{A})\longrightarrow \mathscr S(\mathscr{A}),
\qquad
p^{*}(\psi_{\rho} ):=\rho.
\end{equation}

For each $a\in\mathscr{A}$, we define the set-theoretical section of $(\mathcal{H}(\mathscr{A}),p,\mathcal{S}(\mathscr{A}))$ given by
\begin{equation}\label{eq:tautological sections}
\Psi_a:\mathcal{S}(\mathscr{A})\longrightarrow \mathcal{H}(\mathscr{A}),
\qquad
\Psi_a(\rho):=\psi_a^\rho=[a]_\rho.
\end{equation}
We shall call these sections \textbf{tautological sections} of $\mathcal{H}(\mathscr{A})$.
Analogously, for each $a\in\mathscr{A}$, we can also define the dual section 
\begin{equation}\label{eq:dual tautological section}
\Psi_a^\ast:\mathcal{S}(\mathscr{A})\longrightarrow \mathcal{H}^\ast(\mathscr{A}),
\qquad
\Psi_a^\ast(\rho)\equiv\psi^{a}_{\rho} :=R_\rho(\psi_a^\rho)\in\mathcal{H}_\rho^\ast, 
\end{equation}
where $R_\rho$ is the conjugate-linear Riesz identification 
\begin{equation}\label{eq:Riesz fiberwise}
R_\rho:\mathcal{H}_\rho\longrightarrow \mathcal{H}_\rho^\ast,
\qquad
R_\rho(\eta)(\xi):=\langle\eta , \xi\rangle_\rho.
\end{equation}
Equivalently, for every $\xi\in\mathcal{H}_\rho$ one has
$$
\psi_{a}^{\rho,\ast}(\xi)=\langle\psi_a^\rho, \xi\rangle_\rho.
$$
We shall call the sections $\Psi_a^\ast$ the \textbf{tautological sections} of $\mathcal{H}^\ast(\mathscr{A})$.

\begin{proposition}\label{prop:GNS-nonlocallytrivial-fibration}
Let $(\mathcal{H}(\mathscr{A}),p,\mathcal{S}(\mathscr{A}))$ be a set-theoretical Hilbert fibration as in Definition \ref{def: Banach and Hilbert fibration}, where $\mathcal{S}(\mathscr{A})$ is endowed with the weak$^*$-topology, and $\mathcal{H}(\mathscr{A})$ and $p$ are as in equation \eqref{eq:GNS fibration} and equation \eqref{eq:set theoretical projection}, respectively.
Let $\Gamma_{\mathrm{GNS}}$ be the family of tautological sections of $\mathcal{H}(\mathscr{A})$ as in equation \eqref{eq:tautological sections}.
Then Theorem~\ref{thm: construction thm}  applies,  and thus there exists a unique topology on $\mathcal{H}(\mathscr{A})$ for which
$$
(\mathcal{H}(\mathscr{A}),p,\mathcal{S}(\mathscr{A}))
$$
is a Hilbert fibration and every tautological section $\Psi_a$ is continuous.
We call the (non-locally-trivial) Hilbert fibration $(\mathcal{H}(\mathscr{A}),p,\mathcal{S}(\mathscr{A}))$ the  \emph{GNS fibration}. 

Analogous conclusions hold for the dual fibration $(\mathcal{H}^{\ast}(\mathscr{A}),p^{\ast},\mathcal{S}(\mathscr{A}))$, where $\mathcal{H}^{\ast}(\mathscr{A})$ is as in equation \ref{eq:dual GNS fibration}, $p^{*}$ is as in equation \eqref{eq:set theoretical dual projection}, and the family $\Gamma_{\mathrm{GNS}}^\ast$ of tautological sections of  $\mathcal{H}^{\ast}(\mathscr{A})$ is as in equation \eqref{eq:dual tautological section}.
We call the (non-locally-trivial) Hilbert fibration $(\mathcal{H}^{\ast}(\mathscr{A}),p^{\ast},\mathcal{S}(\mathscr{A}))$ the \emph{dual GNS fibration}. 
\end{proposition}

\begin{proof}
For each $\rho\in\mathscr S(\mathscr{A})$,  the fiber $\mathcal{H}_\rho$ is a complex Hilbert space, hence a complex Banach space.
Moreover, by construction, the family $\Gamma_{\mathrm{GNS}}$ is a complex vector space of sections of $p$.
To check condition (1) in theorem~\ref{thm: construction thm} for $(\mathcal{H}(\mathscr{A}),p,\mathcal{S}(\mathscr{A}))$, fix $a\in\mathscr{A}$, so that
\begin{equation}
\|\Psi_a(\rho)\|^2
=
\|\psi_a^\rho\|_\rho^2
=
\rho(a^\ast a).
\end{equation}
The map $\rho\longmapsto \rho(a^\ast a)$ is continuous for every fixed $a\in\mathscr{A}$ with respect to the weak$^*$ topology on $\mathcal{S}(\mathscr{A})$ by definition, and it follows that $\rho\longmapsto \|\Psi_a(\rho)\|$ is continuous.
To check condition (2) in theorem~\ref{thm: construction thm} for $(\mathcal{H}(\mathscr{A}),p,\mathcal{S}(\mathscr{A}))$, fix $\rho\in\mathcal{S}(\mathscr{A})$. 
By construction of the \emph{GNS} Hilbert space, the set
\begin{equation}
\{\psi_a^\rho:a\in\mathscr{A}\}
=
\{\Psi_a(\rho):a\in\mathscr{A}\}
\end{equation}
is dense in $\mathcal{H}_\rho$, and thus all assumptions of theorem~\ref{thm: construction thm} are therefore satisfied for $(\mathcal{H}(\mathscr{A}),p,\mathcal{S}(\mathscr{A}))$. 

We now consider the dual fibration.  Although the Riesz map
$R_\rho:\mathcal H_\rho\to\mathcal H_\rho^*$ is conjugate-linear, the family
\[
\Gamma_{\mathrm{GNS}}^*
:=
\{\Psi_a^*:a\in\mathscr A\}
\]
is still a complex vector space of sections.  Indeed, for
$\lambda,\mu\in\mathbb C$ and $a,b\in\mathscr A$, one has
\[
\lambda\Psi_a^*+\mu\Psi_b^*
=
\Psi_{\overline{\lambda}a+\overline{\mu}b}^* .
\]
Moreover,
\[
\|\Psi_a^*(\rho)\|_{\mathcal H_\rho^*}
=
\|R_\rho(\psi_a^\rho)\|_{\mathcal H_\rho^*}
=
\|\psi_a^\rho\|_\rho
=
\sqrt{\rho(a^*a)} ,
\]
and hence the function
\[
\rho\longmapsto \|\Psi_a^*(\rho)\|_{\mathcal H_\rho^*}
\]
is continuous for the weak$^*$ topology on $\mathcal S(\mathscr A)$.
Finally, for every fixed $\rho\in\mathcal S(\mathscr A)$, the set
\[
\{\Psi_a^*(\rho):a\in\mathscr A\}
=
\{R_\rho(\psi_a^\rho):a\in\mathscr A\}
\]
is dense in $\mathcal H_\rho^*$, because
$\{\psi_a^\rho:a\in\mathscr A\}$ is dense in $\mathcal H_\rho$ and
$R_\rho$ is an isometric, conjugate-linear bijection.  Therefore the two
hypotheses of Theorem~\ref{thm: construction thm} are also satisfied for
$(\mathcal H^*(\mathscr A),p^*,\mathcal S(\mathscr A))$ with
$\Gamma_{\mathrm{GNS}}^*$.
%%%%%%%%%%%%%%%%%%%%%

\end{proof}

It is worth noting that the complexified tangent bundle of a real smooth manifold provides a natural example of a locally trivial Banach fibration in the sense of Definition~\ref{def: Banach and Hilbert fibration}.
Let $M$ be a real smooth manifold. For every $m\in M$, we denote by
\begin{equation}
T_m^{\mathbb C}M:=T_mM\otimes_{\mathbb R}\mathbb C
\end{equation}
the complexification of the real tangent space $T_mM$.
Every element of $T_m^{\mathbb C}M$ can be written uniquely in the form
$$
v_m\otimes 1+w_m\otimes i,
\qquad v_m,w_m\in T_mM,
$$
which we shall denote simply by $v_m+iw_m$.
This notation identifies the underlying real vector space of
$T_m^{\mathbb C}M$ with $T_mM\oplus T_mM$.\footnote{This decomposition should not be understood as a decomposition into complex subspaces.}
The complexified tangent space also carries the canonical complex-antilinear involution determined by
$$
\overline{v_m\otimes z}:=v_m\otimes \overline z.
$$
Equivalently, if $z_m=v_m+iw_m$, then $\overline{z_m}=v_m-iw_m$.
Its fixed-point set is the canonical copy of $T_mM$ inside
$T_m^{\mathbb C}M$, i.e., $\{z_m \in T_m^\mathbb{C} : \overline{z_m} = z_m\} = \{v_m \otimes 1 : v_m \in T_mM\}$.

An additional Hermitian structure can be defined on $T_m^{\mathbb{C}}M$ for instance by choosing a Riemannian metric on $M$ and extending it Hermitianly to $T^{\mathbb{C}}M$.
More generally, suppose that for every $m\in M$ we are given a positive definite Hermitian form
$$
K_m:T_m^{\mathbb C}M\times T_m^{\mathbb C}M\to\mathbb C
$$
depending smoothly on $m$.
Then each fiber $T_m^{\mathbb C}M$ becomes a complex Hilbert space, and the triple
$$
\left(T^{\mathbb C}M,\tau_M^{\mathbb C},(M,K)\right)
$$
is a locally trivial complex Hilbert fibration where
$$
\tau_M^{\mathbb C}:T^{\mathbb C}M\to M
$$
is the canonical projection.
The Hermitian tensor $K$ determines canonically two real tensors on the underlying real manifold $M$, namely
\begin{equation}\label{eqn: real tensors from K}
G_m(v_m,w_m):=\Re K_m(v_m,w_m),
\qquad
\Omega_m(v_m,w_m):=\Im K_m(v_m,w_m),
\end{equation}
for all $v_m,w_m\in T_mM\subset T_m^{\mathbb C}M$.
In this way, the real part of $K$ defines a symmetric bilinear form on $TM$, while its imaginary part defines a skew-symmetric $2$-form.

This construction provides the natural geometric setting in which the Hermitian structure carried by the GNS fibers can be related to the geometry of a smooth parametric model.
More precisely, under suitable compatibility and regularity assumptions, the Hermitian product of the GNS fibers will induce canonically a Hermitian tensor on the complexified tangent bundle of the model.
In the examples considered later, this point of view will allow us to recover familiar geometric structures of information geometry from a single Hermitian tensor: in the classical case, the real part gives the Fisher--Rao metric tensor, while in the quantum cases it yields, according to the model, the Fubini--Study and SLD metric tensors together with their associated imaginary parts.
%%%%%%%%%%%%%%%%%%%%%%%%%%%%%%

\section{Regular parametric models and induced tensors}\label{sec: parametric models and pullbacks}

Recall that in the $C^{*}$-algebraic approach to classical and quantum information geometry, a \emph{parametric (statistical) model} may be described by a triple $(M,\mathrm{i},\mathscr{A})$ where $M$ is a real smooth manifold, $\mathscr{A}$ is a $C^{*}$-algebra, and $\mathrm{i}\colon M\rightarrow\mathcal{S}(\mathscr{A})$ taking values in the set of states is the map determining the model \cite{CDJS2024}.
%%%%%%%%%%%%%%%%%%%%
If $\mathrm{i}$ is continuous with respect to the topology on $M$ and the weak$^*$-topology on $\mathcal{S}(\mathscr{A})$, we can build two non-locally-trivial Hilbert fibrations over $M$  from the \emph{GNS} and the \emph{dual GNS} fibrations introduced in Section \ref{sec: gns fibration} following \cite[\S II.13.17]{FD1988}.
Specifically, we consider 
\begin{equation}\label{eqn: pullback fibrations}
\begin{split}
\mathrm{i}^{*}\mathcal{H}(\mathscr{A})&:=\left\{(m,\psi)\in M\times \mathcal{H}(\mathscr{A})\mid \mathrm{i}(m)=p(\psi)\right\} \\ & \\
\mathrm{i}^{*}\mathcal{H}^{*}(\mathscr{A})&:=\left\{(m,\eta)\in M\times \mathcal{H}^{*}(\mathscr{A})\mid \mathrm{i}(m)=p^{*}(\eta)\right\}    
\end{split}
\end{equation}
and the obvious projections  over $M$ coming from the left projections on $M\times \mathcal{H}(\mathscr{A})$ and $M\times \mathcal{H}^{*}(\mathscr{A})$.

The purpose of this section is to introduce a class of parametric models for which the complexified tangent space $T_m^{\mathbb C}M$ admits a canonical lift to the fibers of the
pullback dual GNS fibration $\mathrm{i}^*\mathcal{H}^*(\mathscr{A})$.
This makes it possible to induce on $T_m^{\mathbb C}M$ a Hermitian form coming from the dual GNS inner product.
Its real and imaginary parts then determine, on the underlying real tangent space $T_mM$, a symmetric bilinear form and a skew-symmetric $2$-form, respectively.

In many examples of practical interest, the map  $\mathrm{i}\colon M\rightarrow\mathcal{S}(\mathscr{A})$ is such that, for every self-adjoint element $a\in \mathscr{A}_{\mathrm{sa}}$, the  expectation-value map $\ell_a \colon M \to \mathbb{R}$ given by  $ \ell_a(m) := \mathrm{i}(m)(a)$ is smooth.   When this holds, every tangent vector $v_m \in T_m M$  determines a real-linear functional  $\eta^{v_m}$ on the real subspace  of $\mathcal{H}_{\mathrm{i}(m)}^\mathbb{R}$ generated by vectors of the form  $\psi_a^{\mathrm{i}(m)}$ with $a \in \mathscr{A}_{\mathrm{sa}}$, according to:
$$
\langle v_m,d\ell_a(m)\rangle =\Re\left( \eta^{v_m}\bigl(\psi_a^{\mathrm{i}(m)}\bigr)\right) \, .
$$
In general, this functional need not extend to a  continuous real-linear functional on the whole GNS Hilbert space.  This motivates the following definition.
%%%%%%%%%%%%%%%%%%%%%%%%%%

\begin{definition}\label{def:gns-smooth-parametric-model}
A parametric model $(M,\mathrm{i},\mathscr{A})$ is called \emph{GNS-smooth} if the following conditions hold:

\begin{enumerate}
\item\label{cond:smoothness-of-expectations}
For every self-adjoint element $a\in\mathscr{A}_{\mathrm{sa}}$, the expectation-value map
\[
\ell_a:M\longrightarrow\mathbb R,
\qquad
\ell_a(m):=\mathrm{i}(m)(a),
\]
is smooth.

\item\label{cond:GNS-boundedness}
For every $m\in M$ and every $v_m\in T_mM$, setting $\rho:=\mathrm{i}(m)$, there exists a constant $C_{v_{m}}>0$ such that
\begin{equation}\label{eqn:GNS-boundedness-condition}
\bigl|\langle v_m,d\ell_a(m)\rangle\bigr|
\leq
C_{v_{m}}\,\|\psi_a^\rho\|_\rho
=
C_{v_{m}}\sqrt{\rho(a^*a)},
\qquad
\forall a\in\mathscr{A}_{\mathrm{sa}}.
\end{equation}

\item\label{cond:injectivity}
For every $m\in M$, if $v_m\in T_mM$ satisfies
\begin{equation}
\langle v_m,d\ell_a(m)\rangle=0,
\qquad
\forall a\in\mathscr{A}_{\mathrm{sa}},
\end{equation}
then $v_m=0$.
\end{enumerate}

A GNS-smooth parametric model is called \emph{normal} if $\mathscr{A}$ is a $W^*$-algebra and $\mathrm{i}(M)\subset\mathcal{S}_n(\mathscr{A})$.
\end{definition}

\begin{proposition}\label{prop:boundedness-canonical-representative}
Let $(M,\mathrm{i},\mathscr{A})$ be a parametric model satisfying condition~\ref{cond:smoothness-of-expectations} in Definition~\ref{def:gns-smooth-parametric-model}. Fix $m\in M$, set $\rho:=\mathrm{i}(m)$, and define
\begin{equation}\label{eqn:V-rho-definition}
V_\rho
:=
\overline{\operatorname{span}_{\mathbb R}
\{\psi_a^\rho:\ a\in\mathscr{A}_{\mathrm{sa}}\}}
\subset
\mathcal{H}_\rho^{\mathbb R}.
\end{equation}
Let $v_m\in T_mM$. 
%%%%%%%%%%%%%%%%%%%%%%%%%%%%
Then, condition \ref{cond:GNS-boundedness} in Definition \ref{def:gns-smooth-parametric-model} is equivalent to the existence of a unique vector $\xi_{v_m}^{\rho}\in V_\rho$ such that
\begin{equation}\label{eqn:canonical-real-representative}
\langle v_m,d\ell_a(m)\rangle
=
\Re\langle\xi_{v_m}^{\rho},\psi_a^\rho\rangle_\rho
\qquad
\forall a\in\mathscr{A}_{\mathrm{sa}}.
\end{equation}
\end{proposition}

\begin{proof}
Assume first that condition  \ref{cond:GNS-boundedness} in Definition \ref{def:gns-smooth-parametric-model} holds. 
%%%%%%%%%%%%%%%%%%%
Let
\begin{equation}
D_\rho
:=
\operatorname{span}_{\mathbb R}
\{\psi_a^\rho:\ a\in\mathscr{A}_{\mathrm{sa}}\}
\subset V_\rho .
\end{equation}
The formula
\begin{equation}
\varphi_{v_m}(\psi_a^\rho)
:=
\langle v_m,d\ell_a(m)\rangle,
\qquad
a\in\mathscr{A}_{\mathrm{sa}},
\end{equation}
defines a real-linear functional on $D_\rho$.
%%%%%%%%%%
Indeed, equation \eqref{eqn:GNS-boundedness-condition} implies that, if $b$ is in the GNS ideal of $\rho$ so that $\psi_{b}^{\rho}=0$, then $\varphi_{v_m}(\psi_a^\rho)=0$, and linearity follows by direct inspection.
%%%%%%%%%%%%%%%%
Moreover, equation \eqref{eqn:GNS-boundedness-condition} also implies that $\varphi_{v_m}$ is bounded on $D_{\rho}$. 
%%%%%%%%%%%%%%%%%%%%%%%%%%
Since $D_\rho$ is dense in $V_\rho$, it extends uniquely to a continuous real-linear functional on $V_\rho$. 
%%%%%%%%%%%%%%%%%%%%
By the real Riesz representation theorem, there exists a unique $\xi_{v_m}^{\rho}\in V_\rho$ such that
\begin{equation}
\varphi_{v_m}(\zeta)
=
\Re\langle\xi_{v_m}^{\rho},\zeta\rangle_\rho
\qquad
\forall \zeta\in V_\rho.
\end{equation}
Restricting this identity to $\zeta=\psi_a^\rho$ gives equation~\eqref{eqn:canonical-real-representative}.
%%%%%%%%%%%%%%%%%%%%%%%%%%%

Conversely, if $\xi_{v_m}^{\rho}\in V_\rho$ satisfies equation~\eqref{eqn:canonical-real-representative}, then
\begin{equation}
\bigl|\langle v_m,d\ell_a(m)\rangle\bigr|
=
\bigl|\Re\langle\xi_{v_m}^{\rho},\psi_a^\rho\rangle_\rho\bigr|
\leq
\|\xi_{v_m}^{\rho}\|_\rho\,\|\psi_a^\rho\|_\rho
\end{equation}
for all $a\in\mathscr{A}_{\mathrm{sa}}$. 
%%%%%%%%%%%%%%%%
Thus condition \ref{cond:GNS-boundedness} in Definition \ref{def:gns-smooth-parametric-model} holds with $C_{v_{m}}:=\|\xi_{v_m}^{\rho}\|_\rho$.
\end{proof}
 
Let $(M,\mathrm{i},\mathscr{A})$ be a GNS-smooth parametric model as in Definition \ref{def:gns-smooth-parametric-model}.
%%%%%%%%%%%%%%%%%%%%%%%%%%
From Proposition~\ref{prop:boundedness-canonical-representative}, we first obtain the real-linear map
\begin{equation}\label{eqn:real-gns-vector-lift}
\Xi_m:T_mM\longrightarrow V_{\mathrm{i}(m)}\subset\mathcal{H}_{\mathrm{i}(m)}^{\mathbb R},
\qquad
\Xi_m(v_m):=\xi_{v_m}^{\mathrm{i}(m)}.
\end{equation}
%%%%%%%%%%%%%%%%%%%%%
The vector $\Xi_m(v_m)$ is referred to as the \emph{real GNS vector lift} of $v_{m}$.
%%%%%%%%%%%%%%%%%%%%%%%%%%%%%
Then, applying the real Riesz isomorphism, we obtain the linear map
\begin{equation}\label{eqn:real-gns-lift}
L_m^{\mathbb R}:T_mM\longrightarrow
(\mathcal{H}_{\mathrm{i}(m)}^{\mathbb R})^*,
\qquad
L_m^{\mathbb R}(v_m)
:=
R_{\mathrm{i}(m)}^{\mathbb R}\bigl(\Xi_m(v_m)\bigr)=\eta_{\mathbb R}^{v_m}.
\end{equation}
%%%%%%%%%%%%%%%%%
The linear functional in equation \eqref{eqn:real-gns-lift} is referred to as the \emph{real GNS lift} of $v_{m}$.
%%%%%%%%%%%%%%%%%%%%%%%%%%%%%%%

Consider the canonical complex structure $J_{\mathrm{i}(m)}$ defined in Since $\mathcal{H}_{\mathrm{i}(m)}^{\mathbb{R}}$ as in equation \eqref{eqn:complex-structure-realification-gns}. 
%%%%%%%%%%%%%%%%%%%%%%
We shall identify real-dual representatives with
complex-dual representatives through the canonical real-linear isomorphism
\begin{equation}\label{eqn:canonical-complexification-real-dual}
\mathfrak C_{\mathrm{i}(m)}:
(\mathcal{H}_{\mathrm{i}(m)}^{\mathbb R})^*
\longrightarrow
\mathcal{H}_{\mathrm{i}(m)}^*,
\qquad
\mathfrak C_{\mathrm{i}(m)}(\alpha)(\zeta)
:=
\alpha(\zeta)-\mathrm i\,\alpha(J_{\mathrm{i}(m)}\zeta),
\end{equation}
where, with an evident abuse of notation, whenever $\zeta\in\mathcal H_{\mathrm i(m)}$ appears as an argument of $\alpha$ or $J_{\mathrm i(m)}$, it is understood as the same vector regarded as
an element of the realification $\mathcal H_{\mathrm i(m)}^{\mathbb R}$.
Equivalently, $\mathfrak C_{\mathrm{i}(m)}(\alpha)$ is the unique complex-linear
functional on $\mathcal{H}_{\mathrm{i}(m)}$ whose real part is $\alpha$.
%%%%%%%%%%%%%%%%%%%
We then set
\begin{equation}\label{eqn:complex-gns-lift-on-real-vectors}
L_m(v_m):=
\mathfrak C_{\mathrm{i}(m)}\bigl(L_m^{\mathbb R}(v_m)\bigr)
\in \mathcal{H}_{\mathrm{i}(m)}^* .
\end{equation}
%%%%%%%%%%%%%%%%%%%%%%%%%%
By construction, for every $a\in\mathscr{A}_{\mathrm{sa}}$, it holds
\begin{equation}\label{eqn:compatibility-gns-lift-on-real-vectors}
\langle v_m,d\ell_a(m)\rangle=\Re\left(
L_m(v_m)(\psi_a^{\mathrm i(m)})
\right).
\end{equation}
%%%%%%%%%%%%%%%%%%%%%%%%%%%%%%
By complexification, we obtain a complex-linear map
\begin{equation}\label{eqn:complex-gns-lift-on-complex-vectors}
L_m^{\mathbb C}:T_m^{\mathbb C}M
\longrightarrow
\mathcal{H}_{\mathrm{i}(m)}^*,
\qquad
L_m^{\mathbb C}(v_m+\mathrm{i}w_m)
:=
L_m(v_m)+\mathrm{i}L_m(w_m),
\end{equation} 
for each $m\in M$.
%%%%%%%%%%%%%%%%%%%%%%%%%%%%%%
For real tangent vectors one has
\begin{equation}\label{eqn:complex-gns-lift-on-real-vectors-2}
L_m(v_m)=R_{\mathrm i(m)}(\Xi_m(v_m)),
\end{equation}
where $\Xi_m(v_m)$ is seen as an element of $\mathcal{H}_{\mathrm{i}(m)}$ rather than of its realification $\mathcal H_{\mathrm i(m)}^{\mathbb R}$.
%%%%%%%%%%%%%%%%%%%
Since the Riesz map $R_{\mathrm i(m)}:\mathcal{H}_{\mathrm i(m)}\to\mathcal{H}_{\mathrm i(m)}^*$ is conjugate-linear with our convention for the GNS inner product, the Riesz vector corresponding to
$L_m^{\mathbb C}(v_m+\mathrm i w_m)$ is
\begin{equation}\label{eqn:complex-gns-vector-lift}
\Xi_m^{\mathbb C}(v_m+\mathrm i w_m)
:=
\Xi_m(v_m)-\mathrm i\Xi_m(w_m),
\end{equation}
and it holds
\begin{equation}
L_m^{\mathbb C}(v_m+\mathrm i w_m)
=
R_{\mathrm i(m)}
\left(
\Xi_m^{\mathbb C}(v_m+\mathrm i w_m)
\right).
\end{equation}
%%%%%%%%%%%%%%%%%%%%%
\begin{remark}
The construction above selects a canonical complex-linear lift by first selecting the Riesz vector
$\Xi_m(v_m)\in V_{\mathrm i(m)}$ and then applying the Riesz map. If, instead, one only asks for a complex-linear functional
$\eta\in\mathcal H_{\mathrm i(m)}^*$ satisfying
\begin{equation}
\langle v_m,d\ell_a(m)\rangle
=
\Re\left(\eta(\psi_a^{\mathrm i(m)})\right),
\qquad
\forall a\in\mathscr A_{\mathrm{sa}},
\end{equation}
then $\eta$ is generally not unique. Its real part is determined only on
$V_{\mathrm i(m)}$, and it may be changed by adding any continuous real-linear functional vanishing on
$V_{\mathrm i(m)}$.
Equivalently, the ambiguity is modeled by the annihilator
\begin{equation}
\operatorname{Ann}(V_{\mathrm i(m)})
\subset
(\mathcal H_{\mathrm i(m)}^{\mathbb R})^*.
\end{equation}
Thus the canonical object is not an arbitrary complex representative, but the vector
$\Xi_m(v_m)\in V_{\mathrm i(m)}$ selected by Proposition~\ref{prop:boundedness-canonical-representative}.
\end{remark}

For $m\in M$, the maps in equation~\eqref{eqn:complex-gns-lift-on-complex-vectors} assemble into a fiberwise complex-linear lift
\begin{equation}\label{eqn:global-complex-gns-lift}
L^{\mathbb C}:T^{\mathbb C}M
\longrightarrow
\mathrm{i}^*\bigl(\mathcal H^*(\mathscr A)\bigr),
\end{equation}
where $\mathrm{i}^*(\mathcal H^*(\mathscr A))$ denotes the pullback of the dual GNS fibration along $\mathrm{i}$.
%%%%%%%%%%%%%%%%%%%%%%%%%%%%%%%%%
We now pull back the fiberwise Hermitian products of the dual GNS fibration along $L^{\mathbb C}$. 
%%%%%%%%%%%%%%%%%%%
The Hilbert product on the dual GNS Hilbert space is the one transported by the conjugate-linear Riesz identification $R_\rho(\xi)(\zeta)=\langle\xi,\zeta\rangle_\rho$, namely
\begin{equation}\label{eqn:dual-gns-inner-product}
\left\langle
R_\rho(\xi),
R_\rho(\eta)
\right\rangle_\rho^*
:=
\langle \eta,\xi\rangle_\rho .
\end{equation}
For every $m\in M$, with $\rho=\mathrm{i}(m)$, define
\begin{equation}\label{eqn:induced-hermitian-form-on-complex-vectors}
K_m(z_m,z_m')
:=
\bigl\langle
L_m^{\mathbb C}(z_m),
L_m^{\mathbb C}(z_m')
\bigr\rangle_{\rho}^*,
\qquad
z_m,z_m'\in T_m^{\mathbb C}M.
\end{equation}
Equivalently, using equation~\eqref{eqn:complex-gns-vector-lift}, it holds
\begin{equation}\label{eqn:induced-hermitian-form-on-complex-vectors-riesz}
K_m(z_m,z_m')
=
\left\langle
\Xi_m^{\mathbb C}(z_m'),
\Xi_m^{\mathbb C}(z_m)
\right\rangle_{\rho}.
\end{equation}

Restricting the real and imaginary parts of $K_m$ to the real tangent space
$T_mM\subseteq T_m^{\mathbb C}M$, we obtain
\begin{equation}\label{eq:induced-bilinear-forms}
\begin{split}
G_m(v_m,w_m)
&:=
\Re K_m(v_m,w_m),\\
\Omega_m(v_m,w_m)
&:=
\Im K_m(v_m,w_m).
\end{split}
\end{equation}
The bilinear forms are well defined by the  construction of $L_m$ in equation~\eqref{eqn:complex-gns-lift-on-real-vectors}, or equivalently by the canonical construction of
the Riesz-vector map $\Xi_m$ in equation~\eqref{eqn:real-gns-vector-lift}.
%%%%%%%%%%%%%%%%%%%%%%%%%%%%%%%%%%%%%%%%%%%%
Condition~\ref{cond:injectivity} in Definition~\ref{def:gns-smooth-parametric-model} implies that $L_m$ is injective, and hence $G_m$ is positive definite.
Indeed, if $G_m(v_m,v_m)=0$, then $L_m(v_m)=0$, and therefore
\begin{equation}
\langle v_m,d\ell_a(m)\rangle=0,
\qquad
\forall a\in\mathscr{A}_{\mathrm{sa}}.
\end{equation}
By condition~\ref{cond:injectivity} in Definition~\ref{def:gns-smooth-parametric-model}, this gives $v_m=0$.
%%%%%%%%%%%%%%%%%%%%%%%%%%%%%%%
Moreover, by construction, $G_m$ is symmetric and $\Omega_m$ is skew-symmetric.
%%%%%%%%%%%%%%%%%%%%%%%%%%%%%

In principle, there is no guarantee that the forms in equation~\eqref{eq:induced-bilinear-forms}
glue together to smooth tensor fields on $M$, and we are thus led to the following definition.
%%%%%%%%%%%%%%%%%%%%%%%%%%%%%%%%%

\begin{definition}\label{def:regular model}
Let $(M,\mathrm{i},\mathscr{A})$ be a GNS-smooth parametric model as in Definition~\ref{def:gns-smooth-parametric-model}, and let $L^{\mathbb C}:T^{\mathbb C}M\to \mathrm{i}^*(\mathcal{H}^*(\mathscr{A}))$ be as in equation \eqref{eqn:global-complex-gns-lift}.
The model is said to be:
\begin{itemize}
\item \emph{Hermitian regular} if the Hermitian form $K_{m}$ in equation \eqref{eqn:induced-hermitian-form-on-complex-vectors} is jointly continuous on $T_{m}^{\mathbb{C}}M\times T_{m}^{\mathbb{C}}M$ and, for all\footnote{Here, $\mathfrak X^{\mathbb C}(M)
:=
\Gamma(T^{\mathbb C}M)
\cong
\mathfrak X(M)\otimes_{\mathbb R}\mathbb C$, and every $Z,W\in \mathfrak X^{\mathbb C}(M)$ may be written as
$$
Z=X+\mathrm{i}Y,
\qquad
W=U+\mathrm{i}V,
\qquad
X,Y,U,V\in\mathfrak X(M).
$$
}
$Z,W\in \mathfrak X^{\mathbb C}(M)$, the function
$$
m\longmapsto K_m(Z(m),W(m))
$$
is smooth as a complex-valued function\footnote{This means smoothness of both its real and imaginary parts.} on the real smooth manifold $M$;
\item \textbf{\emph{symmetric regular}} if the bilinear form $G_{m}$ in equation \eqref{eq:induced-bilinear-forms} is jointly continuous on $T_{m}M\times T_{m}M$ and, for all $X,Y\in\mathfrak X(M)$, the assignment
\begin{equation}\label{eqn: Riemannian metric}
m\longmapsto
G_m(X(m),Y(m))
:=
\Re K_m(X(m),Y(m))
\end{equation}
defines a smooth real-valued function on $M$, in which case $G$ is a smooth
weak\footnote{When $M$ is finite-dimensional, weak Riemannian metric tensors
are automatically strong.} Riemannian metric tensor on $M$;

\item \textbf{\emph{skew-symmetric regular}} if the bilinear form $\Omega_{m}$  in equation \eqref{eq:induced-bilinear-forms} is jointly continuous on $T_{m}M\times T_{m}M$ and, for all $X,Y\in\mathfrak X(M)$, the assignment
\begin{equation}\label{eqn: skew-tensor}
m\longmapsto
\Omega_m(X(m),Y(m))
:=
\Im K_m(X(m),Y(m))
\end{equation}
defines a smooth real-valued function on $M$, in which case $\Omega$ is a smooth $2$-form on $M$.
\end{itemize}
\end{definition}

\begin{remark}\label{rem:regularity-finite-vs-infinite}
The interplay between the conditions in Definitions~\ref{def:gns-smooth-parametric-model} and~\ref{def:regular model} depends strongly on the dimension.
%%%%%%%%%%%%%%%%%%%%%%%%%%%%%%%%%%%%%%%%%
For finite-dimensional faithful models on $\mathscr{A}=\mathcal{B}(\mathcal{H})$ with $\dim\mathcal{H}<\infty$, condition~\ref{cond:GNS-boundedness} in Definition~\ref{def:gns-smooth-parametric-model} is automatic, and Hermitian regularity follows from the smoothness of the map $\varrho\mapsto \mathcal{J}_\varrho^{-1}$, as shown in Proposition~\ref{prop:finite-dimensional-faithful-states}.
%%%%%%%%%%%%%%%%%%%%%%
In infinite dimensions, condition~\ref{cond:GNS-boundedness} becomes a genuine analytic restriction (the Hilbert--Schmidt integrability condition~\eqref{eqn:HS-integrability-SLD}), and regularity must be verified case by case, as in the displaced thermal model of subsection~\ref{subsec:displaced-thermal-gaussian-model}.
%%%%%%%%%%%%%%%%%%%%%%%%%%%%%%%

Moreover, the joint continuity of $K_{m}$, $G_{m}$, and $\Omega_{m}$ in Definition \ref{def:regular model} needs to be checked only if the parameter manifold $M$ is infinite-dimensional.
%%%%%%%%%%%%%%%%%%%

%%%%%%%%%%%%%%%%%%%%%%%%%%%%%%%%%%%%%
\end{remark}

\begin{remark}\label{rem: hermitianregularity-equivalence}
Hermitian regularity is equivalent to the simultaneous validity of symmetric
regularity and skew-symmetric regularity. Indeed, on real tangent vectors one
has
\[
K_m(v_m,w_m)=G_m(v_m,w_m)+i\Omega_m(v_m,w_m),
\]
and the values of \(K_m\) on \(T_m^{\mathbb C}M\) are recovered by
sesquilinearity. More explicitly, if
\[
Z=X+iY,\qquad W=U+iV,
\]
then, with the convention that \(K\) is conjugate-linear in the first argument
and linear in the second,
\[
\begin{aligned}
K(Z,W)
={}&
G(X,U)+G(Y,V)-\Omega(X,V)+\Omega(Y,U)
\\
&+i\big(
\Omega(X,U)+G(X,V)-G(Y,U)+\Omega(Y,V)
\big).
\end{aligned}
\]
Consequently, smoothness and pointwise continuity of \(K\) are equivalent to
the corresponding properties of both \(G\) and \(\Omega\).
\end{remark}

\begin{remark}\label{rem:partial-regularities-independent}
The preceding remark does not mean that symmetric regularity and skew-symmetric regularity are equivalent separately.
For instance, if $M$ is one-dimensional, then every skew-symmetric bilinear form on $T_mM$ vanishes. 
Consequently, whenever the pointwise construction of
$\Omega$ is defined on a one-dimensional model, one has
$\Omega = 0$, and the model is automatically skew-symmetric regular.
However, the symmetric tensor may still fail to be smooth for rank-changing models \cite{S2017a,SAGP2020}.
%%%%%%%%%%%%%%%%%%%%%
A typical source of such examples appears in qubit models whose image is a curve in the closed Bloch ball meeting a rank-changing stratum.
%%%%%%%%%%%%%%%%%%%%%%%%%
More explicitly, let $\mathscr{A}=M_2(\mathbb C)$ and consider a one-dimensional model
$$
\rho_t(a)=\operatorname{Tr}(\varrho_t a), \qquad \varrho_t=\frac{1}{2}\left(I+r(t)\cdot\sigma\right),
$$
where $r(t)$ parametrizes an ellipse in the closed Bloch ball whose width equals the Bloch radius, and whose height is strictly smaller than the Bloch radius.
%%%%%%%%%%%%%%%%%
Since the parameter space is one-dimensional, the induced two-form satisfies $\Omega=0$.
%%%%%%%%%%%%%%%%%%%%%%%%%%%
Nevertheless, if the canonical representatives fail to depend smoothly on $t$ at the rank-changing points of the curve, then the function $t\longmapsto G_t(\partial_t,\partial_t)$
need not be smooth.
\end{remark}

\begin{remark}
In the following sections, we present examples in which $\Omega$ identically vanishes (Section~\ref{sec:classical-case}), is symplectic (Section~\ref{sec:quantum-pure-states}), or fails to be closed (example~\ref{exmp:qubit-case-omega-not-closed} and subsection~\ref{subsec:displaced-thermal-gaussian-model}). 
%%%%%%%%%%%%%%%%%%%%%%%%
A structural analysis of the closedness question is given in Section~\ref{sec:closed-extensions}.
\end{remark}

An immediate consequence of the construction is that regularity is stable under restriction to immersed submanifolds.

\begin{proposition}\label{prop:submanifold}
Let $(M,\mathrm{i},\mathscr{A})$ be a GNS-smooth parametric model as in Definition~\ref{def:gns-smooth-parametric-model}, and let $k:N\hookrightarrow M$ be a smooth injective immersion.
Then the restricted model $(N,\mathrm{i}\circ k,\mathscr{A})$ is again GNS-smooth.
%%%%%%%%%%%%%%%%%%%%%%%%%%%%%%%%%%%%%%%%%
Moreover, the following hold:
\begin{enumerate}
\item if $(M,\mathrm{i},\mathscr{A})$ is Hermitian regular as in Definition \ref{def:regular model}, then
$(N,\mathrm{i}\circ k,\mathscr{A})$ is Hermitian regular and
\begin{equation}
K^N=k^*K^M;
\end{equation}

\item if $(M,\mathrm{i},\mathscr{A})$ is symmetric regular as in Definition \ref{def:regular model}, then
$(N,\mathrm{i}\circ k,\mathscr{A})$ is symmetric regular and
\begin{equation}
G^N=k^*G^M;
\end{equation}

\item if $(M,\mathrm{i},\mathscr{A})$ is skew-symmetric regular as in Definition \ref{def:regular model}, then
$(N,\mathrm{i}\circ k,\mathscr{A})$ is skew-symmetric regular and
\begin{equation}
\Omega^N=k^*\Omega^M.
\end{equation}
\end{enumerate}
\end{proposition}

\begin{proof}
For $a\in\mathscr{A}_{\mathrm{sa}}$, set $\ell_a^M(m):=\mathrm{i}(m)(a)$ and $\ell_a^N(n):=(\mathrm{i}\circ k)(n)(a)$.
%%%%%%%%%%%%%%%%%%%%%%%%%%%%
Then $\ell_a^N=\ell_a^M\circ k$, so condition~\ref{cond:smoothness-of-expectations} of
Definition~\ref{def:gns-smooth-parametric-model} is inherited from the smoothness of
$\ell_a^M$ and $k$.

Let $n\in N$ and $u_n\in T_nN$. Since
\begin{equation}
\langle u_n,d \ell_a^N(n)\rangle
=
\langle T_nk(u_n),d \ell_a^M(k(n))\rangle ,
\qquad
\forall a\in\mathscr{A}_{\mathrm{sa}},
\end{equation}
condition~\ref{cond:GNS-boundedness} for the model $(M,\mathrm i,\mathscr A)$ gives
\begin{equation}
|\langle u_n,d \ell_a^N(n)\rangle|
\leq
C_{T_nk(u_n)}\,
\|\psi_a^{\mathrm i(k(n))}\|_{\mathrm i(k(n))}
=
C_{T_nk(u_n)}\,
\|\psi_a^{(\mathrm i\circ k)(n)}\|_{(\mathrm i\circ k)(n)} .
\end{equation}
Thus condition~\ref{cond:GNS-boundedness} in Definition \ref{def:gns-smooth-parametric-model} holds for the restricted model. 
%%%%%%%%%%%%%%%%%%%%%%%%%%%%%%%%%%%%%%%%%%%%
If
\begin{equation}
\langle u_n,d \ell_a^N(n)\rangle=0,
\qquad
\forall a\in\mathscr{A}_{\mathrm{sa}},
\end{equation}
then
\begin{equation}
\langle T_nk(u_n),d \ell_a^M(k(n))\rangle=0,
\qquad
\forall a\in\mathscr{A}_{\mathrm{sa}}.
\end{equation}
Condition \ref{cond:injectivity} in Definition \ref{def:gns-smooth-parametric-model} for $(M,\mathrm{i},\mathscr{A})$ gives $T_nk(u_n)=0$, and since $k$ is an immersion,
$u_n=0$. 
%%%%%%%%%%%%%%%%%%%
Hence the restricted model is GNS-smooth as in Definition \ref{def:gns-smooth-parametric-model}.
%%%%%%%%%%%%%%%%%%%%%%%%%%%
Now, since $(\mathrm{i}\circ k)(n)=\mathrm{i}(k(n))$, the original and restricted models have the same GNS fiber and the same subspace
$V_{\mathrm{i}(k(n))}$ at the corresponding point. By uniqueness of the canonical representative,
\begin{equation}
L_n^{N,\mathbb C}(z_n)
=
L_{k(n)}^{M,\mathbb C}
\bigl(T_n^{\mathbb C}k(z_n)\bigr),
\qquad
\forall z_n\in T_n^{\mathbb C}N.
\end{equation}
Therefore
\begin{equation}
\begin{split}
K_n^N(z_n,z_n')
&=
\left\langle
L_n^{N,\mathbb C}(z_n),
L_n^{N,\mathbb C}(z_n')
\right\rangle_{\mathrm{i}(k(n))}^*
\\
&=
K_{k(n)}^M
\left(
T_n^{\mathbb C}k(z_n),
T_n^{\mathbb C}k(z_n')
\right).
\end{split}
\end{equation}
Thus
\begin{equation}
K^N=k^*K^M.
\end{equation}
Consequently, Hermitian regularity is preserved under restriction.
%%%%%%%%%%%%%%%%%%%%%%
Taking real and imaginary parts on real tangent vectors gives
\begin{equation}
G^N=k^*G^M,
\qquad
\Omega^N=k^*\Omega^M.
\end{equation}
Thus symmetric regularity and skew-symmetric regularity are also preserved under restriction.
\end{proof}

\section{Classical models}\label{sec:classical-case}

The purpose of this section is to show that an important class of statistical models of probability measures used in classical information geometry are Hermitian regular GNS-smooth parametric models in the sense of Definitions \ref{def:gns-smooth-parametric-model} and \ref{def:regular model}.
%%%%%%%%%%%%%%%%%%%%%%%%%

In classical information geometry, one usually works with statistical models satisfying regularity assumptions strong enough to make logarithmic derivatives and the Fisher--Rao metric tensor well defined.
%%%%%%%%%%%%%%%%%%%%%
A natural and sufficiently broad framework is provided by the $k$-integrable statistical models of Ay, Jost, L{\^e}, and Schwachh{\"o}fer \cite{AJLS2015,AJLS2017}.
%%%%%%%%%%%%%%%%%%%%%%
In particular, for $k=2$, logarithmic derivatives are square-integrable and their $\mathcal{L}^2$-pairing gives the Fisher--Rao tensor (see \cite[Def.~4.4 and Def.~4.5]{AJLS2015} and \cite[Ch.~4]{AJLS2017}).
%%%%%%%%%%%%%%%%%%%%%%%%%%%%%%%

In this section we consider the dominated faithful case where all probability measures in the model are absolutely continuous with respect to a fixed $\sigma$-finite reference measure $\nu$, and their densities are strictly positive $\nu$-almost everywhere.
%%%%%%%%%%%%%%%%%%%%%%%%%%%%%%%%%%%%%%%%%%%%%%%%%%%
This is precisely the setting naturally associated with normal faithful states on $\mathscr{A}=\mathcal{L}^\infty(X,\nu)$.
%%%%%%%%%%%%%%%%%%%%%%%%%%%%%%%%%%%%%%%%

\begin{definition}\label{def:dominated-2-integrable-statistical-model}
Let $(X,\nu)$ be a $\sigma$-finite measure space, and let $M$ be a real
smooth Banach manifold. A faithful dominated $2$-integrable statistical
model with regular density function is a map
\begin{equation}
p:X\times M\longrightarrow\mathbb R
\end{equation}
satisfying the following conditions.

\begin{enumerate}
\item For every $m\in M$, the function $p(\cdot;m)$ is measurable and satisfies
\begin{equation}
p(x;m)>0
\quad \nu\text{-a.e.},
\qquad
\int_X p(x;m)\,\mathrm d\nu(x)=1.
\end{equation}
We denote by $\rho_m$ the corresponding probability measure such that
\begin{equation}
\frac{\mathrm d\rho_m}{\mathrm d\nu}(x):=p(x;m)\,.
\end{equation}

\item The associated map of density classes
\begin{equation}
P:M\longrightarrow\mathcal{L}^{1}_{\mathbb{R}}(X,\nu),
\qquad
P(m):=[p(\cdot;m)]_\nu,
\end{equation}
is smooth.

\item\label{cond:L1-partial-derivatives} For every $m\in M$ and every $v_m\in T_mM$, the directional derivative $\partial_{v_m}p(x;m)$ exists for $\nu$-almost every $x\in X$,   and is such that 
\begin{equation}
[\partial_{v_m}p(\cdot;m)]_\nu
=
T_{m}P(v_m)\in T_{P(m)}\mathcal{L}^{1}_{\mathbb{R}}(X,\nu)\cong \mathcal{L}^{1}_{\mathbb{R}}(X,\nu).
\end{equation}

\item For every $m\in M$ and every $v_m\in T_mM$, the logarithmic derivative, or score,
\begin{equation}\label{eqn:classical-score-definition}
s_{v_m}(x)
:=
\frac{\partial_{v_m}p(x;m)}{p(x;m)}
\end{equation}
belongs to $\mathcal{L}^2_{\mathbb R}(X,\rho_m)$.

\item\label{cond:continuity-fisher-rao} For every $m\in M$, the formula
\begin{equation}\label{eqn:classical-FR-pointwise}
G_m^{\mathrm{FR}}(v_m,w_m)
:=
\int_X s_{v_m}(x)s_{w_m}(x)\,\mathrm d\rho_m(x)
\end{equation}
defines a continuous bilinear form on $T_mM\times T_mM$.

\item\label{cond:smoothness-fisher-rao} For every pair of smooth vector fields $U,V\in\mathfrak X(M)$, the function
\begin{equation}\label{eqn:score-tensors-classical}
m\longmapsto
G_m^{\mathrm{FR}}(U(m),V(m))
=
\int_X
s_{U(m)}(x)s_{V(m)}(x)
\,\mathrm d\rho_m(x)
\end{equation}
is smooth.
\end{enumerate}

The tensor
\begin{equation}\label{eqn:classical-FR-pre}
G^{\mathrm{FR}}(U,V)(m)
:=
\int_X
s_{U(m)}(x)s_{V(m)}(x)
\,\mathrm d\rho_m(x)
\end{equation}
is called the \emph{Fisher--Rao tensor} of the model.
%%%%%%%%%%%%%%
The model is called \emph{nondegenerate} if $G^{\mathrm{FR}}$ is positive definite, hence a weak Riemannian metric tensor on $M$.
\end{definition}

\begin{remark}
The continuity requirement in equation~\eqref{eqn:classical-FR-pointwise}
is included because, on an infinite-dimensional parameter manifold, it is
not a consequence of square-integrability of the scores alone.  Indeed, the
condition $s_{v_m}\in\mathcal L^2(X,\rho_m)$ for every $v_m\in T_mM$ only
says that the score map
\[
S_m:T_mM\longrightarrow \mathcal L^2_{\mathbb R}(X,\rho_m),
\qquad
S_m(v_m):=s_{v_m},
\]
is everywhere defined.  It does not imply that $S_m$ is continuous.  The
continuity of $G_m^{\mathrm{FR}}$ is equivalent to the continuity of the
quadratic form
\[
v_m\longmapsto \|S_m(v_m)\|^2_{\mathcal L^2(X,\rho_m)}
\]
and, in particular, is guaranteed if $S_m$ is continuous.

If $M$ is finite-dimensional, then every everywhere-defined linear map
$S_m:T_mM\to\mathcal L^2_{\mathbb R}(X,\rho_m)$ is automatically
continuous.  Thus, in finite dimension, the continuity condition in point 5. is automatic
once the scores are defined for all tangent vectors.
\end{remark}

\begin{remark}\label{rem:derivative-under-integral}
Condition~\ref{cond:L1-partial-derivatives} in Definition \ref{def:dominated-2-integrable-statistical-model} makes explicit the usual passage of derivatives under the integral sign.
%%%%%%%%%%%%%%%%%%%%%%%%%%%%%%%%
Let $a\in\mathcal{L}^\infty_{\mathbb R}(X,\nu)$ and consider the expectation-value function $\ell_a:M\to\mathbb R$ given by
\begin{equation}
\ell_a(m):=\int_X a(x)p(x;m)\,\mathrm d\nu(x)\,.
\end{equation}
%%%%%%%%%%%%%%%%%%%%%%%%%
If $\gamma_m^{v_m}:(-\epsilon,\epsilon)\to M$ is a smooth curve such that $\gamma_m^{v_m}(0)=m$ and $\dot\gamma_m^{v_m}(0)=v_m$, then
\begin{equation}
\langle v_{m}, \mathrm{d}\ell_{a}(m)\rangle=\frac{\mathrm d}{\mathrm dt}\bigg|_{t=0}\ell_a(\gamma_m^{v_m}(t))=\frac{\mathrm d}{\mathrm dt} \bigg|_{t=0}\int_X a(x)p(x;\gamma_m^{v_m}(t))\,\mathrm d\nu(x).
\end{equation}
%%%%%%%%%%%%%%%%%%%%%
By Condition~\ref{cond:L1-partial-derivatives} in Definition \ref{def:dominated-2-integrable-statistical-model}, the curve $t\longmapsto [p(\cdot;\gamma_m^{v_m}(t))]_\nu$ is differentiable in $\mathcal{L}^1(X,\nu)$ at $t=0$, with derivative $[\partial_{v_m}p(\cdot;m)]_\nu$.
%%%%%%%%%%%%%%%%%%%%%%%%%%%%%%
Since $a\in\mathcal{L}^{\infty}_{\mathbb{R}}(X,\nu)$ defines a continuous linear functional on
$\mathcal{L}^{1}_{\mathbb{R}}(X,\nu)$ by integration, it follows that
\begin{equation}
\frac{\mathrm d}{\mathrm dt}\bigg|_{t=0}
\int_X a(x)p(x;\gamma_m^{v_m}(t))\,\mathrm d\nu(x)
=
\int_X a(x)\partial_{v_m}p(x;m)\,\mathrm d\nu(x).
\end{equation}
Using the score in equation \eqref{eqn:classical-score-definition}, this can be rewritten as
\begin{equation}\label{eqn:compatibility-condition-classical-score}
\langle v_{m}, \mathrm{d}\ell_{a}(m)\rangle
=
\int_X a(x)s_{v_m}(x)\,\mathrm d\rho_m(x).
\end{equation}
\end{remark}

Every faithful dominated $2$-integrable statistical model determines a parametric model $(M,\mathrm{i},\mathscr{A})$ of faithful normal states in the sense of the operator-algebraic approach to information geometry \cite{CDJS2024}.
%%%%%%%%%%%%%%%%%%%%%%%%%
Indeed, set $\mathscr{A}=\mathcal{L}^\infty(X,\nu)$, and, for each $m\in M$, define a state $\rho_m\equiv \mathrm{i}(m)$ on $\mathscr{A}$ by
\begin{equation}\label{eqn:classical-state-density}
\rho_m(a)
=
\int_X a(x)p(x;m)\,\mathrm d\nu(x),
\qquad
a\in\mathscr{A}.
\end{equation}
%%%%%%%%%%%%%%%%%%%%%%%%%%%%%%%%%
Since $p(\cdot;m)>0$ $\nu$-almost everywhere, $\rho_m$ is a normal faithful state on $\mathscr{A}$.
%%%%%%%%%%%%%%%%%%%%%

The next proposition shows that, when the  faithful dominated $2$-integrable statistical model is nondegenerate, the associated operator-algebraic model is GNS-smooth and Hermitian regular in the sense of Definitions~\ref{def:gns-smooth-parametric-model} and~\ref{def:regular model}.
Moreover, the canonical GNS representative is the classical score, the induced Riemannian metric is the Fisher--Rao metric, and the induced skew-symmetric tensor vanishes identically.
%%%%%%%%%%%%%%%%%%%%%%%%%%%%%%%%%%%%%%%

\begin{proposition}\label{prop:classical-case}
Let $(M,X,\nu,p)$ be a nondegenerate faithful dominated $2$-integrable
statistical model with regular density function in the sense of
Definition~\ref{def:dominated-2-integrable-statistical-model}.
Let $\mathscr{A}=\mathcal{L}^\infty(X,\nu)$ and let
$(M,\mathrm i,\mathscr{A})$ be the associated operator-algebraic parametric
model of normal faithful states defined by
equation~\eqref{eqn:classical-state-density}.
Then $(M,\mathrm i,\mathscr{A})$ is a Hermitian regular GNS-smooth parametric model as in Definitions~\ref{def:gns-smooth-parametric-model} and \ref{def:regular model}.
%%%%%%%%%%%%%%%%%%%%%%%%%%%%%%

For every $v_m\in T_mM$, equation~\eqref{eqn:real-gns-vector-lift} is given by
\begin{equation}\label{eqn:real-gns-vector-lift-classical-case}
\Xi_m(v_m)=s_{v_m},
\end{equation}
where $s_{v_m}$ is the logarithmic derivative in
$\mathcal{L}^2_{\mathbb R}(X,\rho_m)$, while equation \eqref{eqn:complex-gns-vector-lift} is given by
\begin{equation}\label{eqn:complex-gns-vector-lift-classical-case}
\Xi_{m}^{\mathbb{C}}(z_m)=\widehat s_{z_m}:=s_{v_m}-\mathrm i s_{w_m},
\end{equation}
with $z_m=v_m+\mathrm i w_m\in T_m^{\mathbb C}M$.
%%%%%%%%%%%%%%%%%%%%%%%%%%%%%%%%
Moreover, the induced Hermitian form in equation \eqref{eqn:induced-hermitian-form-on-complex-vectors-riesz} is
\begin{equation}\label{eqn:classical-Hermitian-tensor}
K_m(z_m,z_m')
=
\int_X
\overline{\widehat s_{z_m'}(x)}
\widehat s_{z_m}(x)
\,\mathrm d\rho_m(x).
\end{equation}
%%%%%%%%%%%%%%%%%%%%%%%%%%%
On real tangent vectors, its imaginary part $\Omega_m$ vanishes, and its real part is
\begin{equation}\label{eqn:classical-FR}
G_m(v_m,w_m)
=
\int_X
s_{v_m}(x)s_{w_m}(x)
\,\mathrm d\rho_m(x).
\end{equation}
Thus $G=G^{\mathrm{FR}}$ and $\Omega=0$.
%%%%%%%%%%%%%%%%
\end{proposition}

\begin{proof}
For each $m\in M$, the GNS Hilbert space $\mathcal{H}_{\rho_m}$ is canonically
identified with $\mathcal{L}^2(X,\rho_m)$, and the GNS vector
$\psi_a^{\rho_m}$ associated with $a\in\mathcal{L}^\infty(X,\nu)$ is the class of $a$ in $\mathcal{L}^2(X,\rho_m)$ because
\begin{equation}
\langle \psi_a^{\rho_m},\psi_b^{\rho_m}\rangle_{\rho_m}
=
\rho_m(a^*b)
=
\int_X \overline{a(x)}b(x)\,\mathrm d\rho_m(x).
\end{equation}
Since $\rho_m$ is faithful, the null ideal is the usual null space modulo
$\rho_m$-almost everywhere equality.
%%%%%%%%%%%%%%%%%%%%%%%%%%%

Since $\mathscr{A}$ is commutative, self-adjoint elements correspond to real-valued functions. 
%%%%%%%%%%%%%%%%%
Moreover, bounded real-valued functions are dense in $\mathcal{L}^2_{\mathbb R}(X,\rho_m)$. 
%%%%%%%%%%%%%%%%%%%%
Hence the distinguished real subspace in equation~\eqref{eqn:V-rho-definition} is
\begin{equation}\label{eqn:vrho-classical-case}
V_{\rho_m}
=
\mathcal{L}^2_{\mathbb R}(X,\rho_m).
\end{equation}

We first prove that $(M,\mathrm i,\mathscr{A})$ is GNS-smooth as in Definition~\ref{def:gns-smooth-parametric-model}. 
%%%%%%%%%%%%%%%%%%%%%%%%%%
Let $a\in\mathcal{L}^\infty_{\mathbb R}(X,\nu)$ and consider the expectation-value function
\begin{equation}
\ell_a(m)=\rho_m(a)
=
\int_X a(x)p(x;m)\,\mathrm d\nu(x).
\end{equation}
By Definition~\ref{def:dominated-2-integrable-statistical-model}, the map
$P\colon M\rightarrow\mathcal{L}^1_{\mathbb R}(X,\nu)$ is smooth, and
$a\in\mathcal{L}^\infty_{\mathbb R}(X,\nu)$ defines a continuous linear functional $\Lambda_a$ on $\mathcal{L}^1_{\mathbb R}(X,\nu)$ by integration.
Since $\ell_a=\Lambda_a\circ P$, the map $\ell_a$ is smooth. This proves
condition~\ref{cond:smoothness-of-expectations} in
Definition~\ref{def:gns-smooth-parametric-model}.

Let now $m\in M$ and $v_m\in T_mM$. By
Remark~\ref{rem:derivative-under-integral}, one has
\begin{equation}\label{eqn:classical-score-representation-proof}
\langle v_m,\mathrm d\ell_a(m)\rangle
=
\int_X a(x)s_{v_m}(x)\,\mathrm d\rho_m(x),
\end{equation}
where $s_{v_m}$ is the score in
equation~\eqref{eqn:classical-score-definition}. 
%%%%%%%%%%%%%%%%%%%%%%%
Since
$s_{v_m}\in\mathcal{L}^2_{\mathbb R}(X,\rho_m)\subset\mathcal{L}^{2}(X,\rho_{m})$, condition~\ref{cond:GNS-boundedness} in
Definition~\ref{def:gns-smooth-parametric-model} holds with $C_{v_{m}}=\Vert s_{v_m} \Vert$.
%%%%%%%%%%%%%%%%%

It remains to verify condition~\ref{cond:injectivity} in
Definition~\ref{def:gns-smooth-parametric-model}. 
%%%%%%%%%%%%%%%%%%%%%%%%%
Assume that
\begin{equation}
\langle v_m,\mathrm d\ell_a(m)\rangle
=
\int_X a(x)s_{v_m}(x)\,\mathrm d\rho_m(x)
=
0,
\end{equation}
for all $a\in\mathscr{A}_{\mathrm{sa}}=\mathcal{L}^\infty_{\mathbb R}(X,\nu)$.
%%%%%%%%%%%%%%
Since $s_{v_m}\in\mathcal{L}^2(X,\rho_m)\subseteq\mathcal{L}^1(X,\rho_m)$, it
follows that $s_{v_m}=0$ $\rho_m$-almost everywhere. Hence
\begin{equation}
G_m^{\mathrm{FR}}(v_m,v_m)=0.
\end{equation}
Since the model is nondegenerate, $G^{\mathrm{FR}}$ is positive definite, and therefore $v_m=0$,  thus proving condition~\ref{cond:injectivity}.
%%%%%%%%%%%%%%%%%%%%%%%%%%%%%%%

Equation \eqref{eqn:real-gns-vector-lift-classical-case} follows from equations \eqref{eqn:real-gns-vector-lift}, \eqref{eqn:vrho-classical-case} and \eqref{eqn:classical-score-representation-proof}.
%%%%%%%%%%%%%%%%%%%%%%%%
Consequently, equation \eqref{eqn:complex-gns-vector-lift-classical-case} follows from equation \eqref{eqn:complex-gns-vector-lift}. 
%%%%%%%%%%%%%%%%%%%%%%%%%%
Using equation~\eqref{eqn:induced-hermitian-form-on-complex-vectors-riesz}, we obtain
\begin{equation}
K_m(z_m,z_m')
=
\left\langle
\widehat s_{z_m'},
\widehat s_{z_m}
\right\rangle_{\rho_m}
=
\int_X
\overline{\widehat s_{z_m'}(x)}
\widehat s_{z_m}(x)
\,\mathrm d\rho_m(x).
\end{equation}
%%%%%%%%%%%%%%%%%%%%%%%%%
If $v_m,w_m\in T_mM$ are real tangent vectors, then
$\widehat s_{v_m}=s_{v_m}$ and $\widehat s_{w_m}=s_{w_m}$ are real-valued.
%%%%%%%%%%%%%%%%%%%%%%%%%%%%
Therefore equation~\eqref{eq:induced-bilinear-forms} gives
\begin{equation}
G_m(v_m,w_m)
=
G_m^{\mathrm{FR}}(v_m,w_m),
\qquad
\Omega_m(v_m,w_m)=0.
\end{equation}
%%%%%%%%%%%%%%%%%%%%%%%%%%%%
By conditions \ref{cond:continuity-fisher-rao} and \ref{cond:smoothness-fisher-rao} in Definition~\ref{def:dominated-2-integrable-statistical-model}, it follows that $G^{\mathrm{FR}}$ is a smooth weak Riemannian metric tensor.
%%%%%%%%%%%%%%%%%%%%%%%%%%%%%%%%%%%%%%%%
Since $\Omega=0$, the skew-symmetric tensor is trivially smooth. 
%%%%%%%%%%%%%%%%%%
Therefore, by Remark~\ref{rem: hermitianregularity-equivalence}, the model is Hermitian regular as in Definition~\ref{def:regular model}.
%%%%%%%%%%%%%%%%%%%%%%
\end{proof}

\section{Quantum pure states}\label{sec:quantum-pure-states}

We now consider the model of pure normal states on the algebra of bounded operators on a separable complex Hilbert space $\mathcal{H}$, where the Hilbert product is conjugate-linear in the first argument and linear in the second.
%%%%%%%%%%%%%%
The purpose of the section is to show that these states form a Hermitian regular GNS-smooth parametric model as in Definitions \ref{def:gns-smooth-parametric-model} and \ref{def:regular model}, and that the resulting geometric tensors coincide (up to normalization and sign/order convention) with the standard Fubini--Study metric tensor and symplectic form.
%%%%%%%%%%%%%%%%%%%%%%%%%%%

Let $\mathscr{A}=\mathcal{B}(\mathcal{H})$ and $M=\mathbb{CP}(\mathcal{H})$, where $\mathbb{CP}(\mathcal{H}):=\mathbb{S}(\mathcal{H})/U(1)$ denotes the projective Hilbert space of $\mathcal{H}$, with $\mathbb{S}(\mathcal{H})$ the unit sphere in $\mathcal{H}$. 
%%%%%%%%%%%%%%%%%%%
To make contact with parametric models in information geometry which use real manifolds, we regard $M=\mathbb{CP}(\mathcal{H})$ as a real smooth Hilbert manifold, although it also carries a natural complex Hilbert manifold structure.
%%%%%%%%%%%%%%%%%%%%%%%%%%%%%%%
Define the map $\mathrm{i}\colon M
\longrightarrow
\mathcal{S}(\mathscr{A})$ setting
\begin{equation}
\mathrm{i}([\psi])(a)=\rho_\psi(a)=\langle\psi,a\psi\rangle, \quad \forall\;a\in\mathscr{A},
\end{equation}
where $\psi\in\mathcal{H}$ is any normalized representative of the ray $[\psi]$.
%%%%%%%%%%%%%%%%%%%%%%%%%
This map is well defined because $\rho_{e^{\mathrm i\theta}\psi}=\rho_\psi$.
%%%%%%%%%%%%%%%%%%%%%%%%%%

For each normalized representative $\psi$ of the ray $[\psi]$, the GNS Hilbert
space $\mathcal{H}_{\rho_\psi}$ may be identified with $\mathcal{H}$ by the map
\begin{equation}\label{eqn:pure-state-GNS-identification}
U_\psi:\mathcal{H}_{\rho_\psi}\longrightarrow\mathcal{H},
\qquad
U_\psi(\psi_a^{\rho_\psi})=a\psi,
\qquad a\in \mathscr{A}.
\end{equation}
Indeed, it holds
\begin{equation}
\left\langle
\psi_a^{\rho_\psi},
\psi_b^{\rho_\psi}
\right\rangle_{\rho_\psi}
=
\rho_\psi(a^*b)
=
\langle a\psi,b\psi\rangle,
\end{equation}
so $U_\psi$ is an isometry on the dense subspace of GNS vectors. 
%%%%%%%%%%%%%%%%%%%%%
Its range is all of $\mathcal{H}$, because for every $\varphi\in\mathcal{H}$ the rank-one operator
\begin{equation}
\theta_{\varphi,\psi}(\eta):=\varphi\langle\psi,\eta\rangle
\end{equation}
satisfies $\theta_{\varphi,\psi}\psi=\varphi$. 
%%%%%%%%%%%%%%%%%%%%%%%%%%%
Hence $U_\psi$ extends to a unitary isomorphism. 
%%%%%%%%%%%%%%%%%%%%%%
This identification depends on the choice of normalized
representative because, if $\widetilde\psi=e^{\mathrm i\theta}\psi$, then
\begin{equation}
U_{\widetilde\psi}
=
e^{\mathrm i\theta}U_\psi .
\end{equation}
%%%%%%%%%%%%%%%%%%%%%%%%%%%%%%%
The tensors obtained below are independent of this choice.
%%%%%%%%%%%%%%%%%%%%%%%%%%%%%%%
Under the identification induced by $U_\psi$, the distinguished real subspace
$V_{\rho_\psi}$ in equation~\eqref{eqn:V-rho-definition} is
\begin{equation}\label{eqn:pure-state-Vrho}
U_\psi(V_{\rho_\psi})
=
\{\varphi\in\mathcal{H}:\langle\psi,\varphi\rangle\in\mathbb R\}.
\end{equation}
%%%%%%%%%%%%%%%%%%%%%%%%%%%%%%
Indeed, if $\varphi=U_{\psi}(\psi_{a}^{\rho_{\psi}})=a\psi$ with $a=a^*$, then
\begin{equation}
\langle\psi,\varphi\rangle
=
\langle\psi,a\psi\rangle
\in\mathbb R.
\end{equation}
%%%%%%%%%%%%%%%%%%%%%%%%%%%%%
Conversely, if $\varphi\in\mathcal{H}$ satisfies
$\langle\psi,\varphi\rangle\in\mathbb R$, we can write
\begin{equation}
\varphi=\alpha\psi+\chi,
\qquad
\alpha\in\mathbb R,
\qquad
\chi\perp\psi,
\end{equation}
and the self-adjoint operator
\begin{equation}
a
=
\alpha\,\theta_{\psi,\psi}
+
\theta_{\chi,\psi}
+
\theta_{\psi,\chi}
\end{equation}
satisfies $U_{\psi}(\psi_{a}^{\rho_{\psi}})=a\psi=\varphi$. 
%%%%%%%%%%%%%%%%%%%%%%
Moreover, it also holds
\begin{equation}\label{eqn:pure-state-vertical-direction}
U_\psi(V_{\rho_\psi})^{\perp_{\mathbb R}}
=
\mathbb R\,\mathrm i\psi,
\end{equation}
which is precisely the vertical phase direction.

Let $v_{[\psi]}\in T_{[\psi]}\mathbb{CP}(\mathcal{H})$. Choose a smooth curve
of normalized vectors $\psi(t)$ such that
\begin{equation}
\psi(0)=\psi,
\qquad
\frac{\mathrm d}{\mathrm dt}[\psi(t)]\bigg|_{t=0}
=
v_{[\psi]},
\end{equation}
and set
\begin{equation}
\dot\psi
:=
\frac{\mathrm d}{\mathrm dt}\psi(t)\bigg|_{t=0}.
\end{equation}
%%%%%%%%%%%%%%%%%%%%%%%
Since $\|\psi(t)\|=1$, it follows that
\begin{equation}
\Re\langle\psi,\dot\psi\rangle=0.
\end{equation}
%%%%%%%%%%%%%%%%%%%%%%%%%
We define the horizontal component of $\dot\psi$ by
\begin{equation}\label{eqn:pure-state-horizontal-component}
\dot\psi_{\mathrm h}
:=
\dot\psi-\langle\psi,\dot\psi\rangle\psi,
\end{equation}
so that
\begin{equation}
\langle\psi,\dot\psi_{\mathrm h}\rangle=0.
\end{equation}
%%%%%%%%%%%%%%%%%%%%%%%%%%
If another normalized lift is chosen, say
\begin{equation}
\widetilde\psi(t)=e^{\mathrm i\theta(t)}\psi(t),
\end{equation}
then
\begin{equation}
\dot{\widetilde\psi}_{\mathrm h}
=
e^{\mathrm i\theta(0)}\dot\psi_{\mathrm h}.
\end{equation}
%%%%%%%%%%%%%%%%%%%%%%%%%%%%%%%%
Together with $U_{e^{\mathrm i\theta}\psi}=e^{\mathrm i\theta}U_\psi$, this
shows that the GNS vector determined below does not depend on the chosen
normalized representative.
%%%%%%%%%%%%%%%%%%%%%%%%%%%%%%%%%%%%

Recall that the projective Hilbert space $\mathbb{CP}(\mathcal{H})=M$ carries its standard Fubini--Study K\"ahler structure \cite{CMP1990,AS1999,BZ2006,EMM2010}. 
%%%%%%%%%%%%%%%%%%%%%
With the convention that the Hilbert product is conjugate-linear in the first argument and linear in the second, we normalize the Fubini--Study Hermitian tensor by
\begin{equation}\label{eqn:FS-Hermitian-tensor}
h_{FS,[\psi]}(v_{[\psi]},w_{[\psi]})
:=
\left\langle
\dot\psi_{\mathrm h},\dot\phi_{\mathrm h}
\right\rangle,
\end{equation}
where $\dot\psi_{\mathrm h}$ and $\dot\phi_{\mathrm h}$ are the horizontal components of normalized lifts representing $v_{[\psi]}$ and $w_{[\psi]}$.
%%%%%%%%%%%%%%%%%%%%%%%%%%%%%%
Its real and imaginary parts are denoted by
\begin{equation}\label{eqn:FS-real-imaginary-parts}
g_{FS}:=\Re h_{FS},
\qquad
\omega_{FS}:=\Im h_{FS},
\end{equation}
and are, respectively, the Fubini--Study metric tensor and symplectic form.
%%%%%%%%%%%%%%%%%%%%%%%%%%%%%%%%
Other normalizations of the Fubini--Study metric differ from this one by an overall positive constant.
%%%%%%%%%%%%%%%%%%%%%%%%%%%%%%%%%%%%%%%%

\begin{proposition}\label{prop:Fubini--Study}
Let $\mathscr{A}=\mathcal{B}(\mathcal{H})$ and $M=\mathbb{CP}(\mathcal{H})$.
%%%%%%%%%%%%%%%%%%%%%%%%%%%
The pure-state model $(M,\mathrm{i},\mathscr{A})$ defined above is a Hermitian regular GNS-smooth parametric model of normal states as in Definitions~\ref{def:gns-smooth-parametric-model} and~\ref{def:regular model}.
%%%%%%%%%%%%%%%%%%%%%%%%%%%%%%%%%%%%%%

Moreover, if $v_{[\psi]},w_{[\psi]}\in T_{[\psi]}M$ are represented by normalized curves $\psi(t)$ and $\phi(t)$ with
\begin{equation}
\psi(0)=\phi(0)=\psi,
\end{equation}
and if $\dot\psi_{\mathrm h}$ and $\dot\phi_{\mathrm h}$ denote the horizontal components of their derivatives at $t=0$, then, on real tangent vectors, the Hermitian form in equations \eqref{eqn:induced-hermitian-form-on-complex-vectors} and \eqref{eqn:induced-hermitian-form-on-complex-vectors-riesz} reads
\begin{equation}\label{eqn:pure-state-K-real}
K_{[\psi]}(v_{[\psi]},w_{[\psi]})
=
4\langle\dot\phi_{\mathrm h},\dot\psi_{\mathrm h}\rangle
=
4\left(
\langle\dot\phi,\dot\psi\rangle
-
\langle\dot\phi,\psi\rangle
\langle\psi,\dot\psi\rangle
\right).
\end{equation}
Consequently, the real and imaginary parts in equation \eqref{eq:induced-bilinear-forms} read
\begin{equation}\label{eqn:pure-state-G}
G_{[\psi]}(v_{[\psi]},w_{[\psi]})
=
4\Re\langle\dot\phi_{\mathrm h},\dot\psi_{\mathrm h}\rangle
\end{equation}
and
\begin{equation}\label{eqn:pure-state-Omega}
\Omega_{[\psi]}(v_{[\psi]},w_{[\psi]})
=
4\Im\langle\dot\phi_{\mathrm h},\dot\psi_{\mathrm h}\rangle,
\end{equation}
respectively.
%%%%%%%%%%%%%%
Therefore, with the normalization in
equation~\eqref{eqn:FS-Hermitian-tensor}, one has
\begin{equation}
K=4\overline{h_{FS}},
\qquad
G=4g_{FS},
\qquad
\Omega=-4\omega_{FS}.
\end{equation}
%%%%%%%%%%%%%%%%%%%%%%
\end{proposition}

\begin{proof}
We first verify GNS-smoothness in the sense of
Definition~\ref{def:gns-smooth-parametric-model}. 
%%%%%%%%%%%%%%%%%%%%%%%%%%%%%
Let $a\in\mathscr A_{\mathrm{sa}}$ and consider the function
$F_a:\mathbb S(\mathcal H)\longrightarrow\mathbb R$ given by $F_a(\psi)=\langle\psi,a\psi\rangle$.
%%%%%%%%%%%%%%%%%%%%%%%%%%%
This function is smooth because it is the restriction to $\mathbb S(\mathcal H)$ of the smooth quadratic map $\psi\mapsto\langle\psi,a\psi\rangle$ on $\mathcal H$.
%%%%%%%%%%%%%%%%%%%
Moreover, $F_a(e^{\mathrm i\theta}\psi)=F_a(\psi)$, which means it descends to a well-defined function on $M=\mathbb{CP}(\mathcal{H})=\mathbb{S}(\mathcal{H})/U(1)$.
%%%%%%%%%%%%%%%%%%%%
This function is precisely the expectation-value function $\ell_a([\psi])=\langle\psi,a\psi\rangle$, and satisfies $F_a=\ell_a\circ q$, where $q:\mathbb{S}(\mathcal{H})\to\mathbb{CP}(\mathcal H)=M$ is the quotient
projection. 
%%%%%%%%%%%%%%%%%
Since $q$ is the projection of a principal $U(1)$-bundle, smoothness of the invariant function $F_a$ implies smoothness of the descended function $\ell_a$.
%%%%%%%%%%%%%%%%%%%%%%%%%%
This proves condition~\ref{cond:smoothness-of-expectations} in
Definition~\ref{def:gns-smooth-parametric-model}.
%%%%%%%%%%%%%%%%%%%%%%%%%%%%%%%%%%%%%%

Let $v_{[\psi]}\in T_{[\psi]}M$ and let $\psi(t)$ be a
smooth normalized lift as above. 
%%%%%%%%%%%%%%%%%%%%
For $a=a^*$, it holds
\begin{equation}\label{eqn:pure-state-expectation-derivative}
\begin{split}
\left\langle v_{[\psi]},\mathrm{d}\ell_a([\psi])\right\rangle
&=
\frac{\mathrm d}{\mathrm dt}\bigg|_{t=0}
\langle\psi(t),a\psi(t)\rangle
\\
&=
\langle\dot\psi,a\psi\rangle
+
\langle\psi,a\dot\psi\rangle
\\
&=
2\Re\langle\dot\psi,a\psi\rangle .
\end{split}
\end{equation}
Since $\langle\psi,\dot\psi\rangle$ is purely imaginary and
$\langle\psi,a\psi\rangle$ is real, equations \eqref{eqn:pure-state-horizontal-component} and \eqref{eqn:pure-state-expectation-derivative} imply that
\begin{equation}\label{eqn:horizontal-compatibility-pure-states}
\left\langle v_{[\psi]},\mathrm{d}\ell_a([\psi])\right\rangle
=
2\Re\langle\dot\psi_{\mathrm h},a\psi\rangle.
\end{equation}
%%%%%%%%%%%%%%%%%%%%%%%
Using the GNS identification $U_\psi$, equation~\eqref{eqn:horizontal-compatibility-pure-states} can be rewritten as
\begin{equation}
\left\langle v_{[\psi]},\mathrm{d}\ell_a([\psi])\right\rangle
=
\Re\left\langle
2\dot\psi_{\mathrm h},
U_\psi(\psi_a^{\rho_\psi})
\right\rangle
\qquad
\forall a\in\mathscr A_{\mathrm{sa}}.
\end{equation}
Since $\dot\psi_{\mathrm h}$ is orthogonal to $\psi$, equation~\eqref{eqn:pure-state-Vrho} implies $2\dot\psi_{\mathrm h}\in U_\psi(V_{\rho_\psi})$, and thus equation~\eqref{eqn:canonical-real-representative} holds with
\begin{equation}\label{eqn:pure-state-Xi}
U_\psi\left(\Xi_{[\psi]}(v_{[\psi]})\right)
=
2\dot\psi_{\mathrm h}.
\end{equation}
By Proposition~\ref{prop:boundedness-canonical-representative}, condition~\ref{cond:GNS-boundedness} in Definition~\ref{def:gns-smooth-parametric-model} follows.

We now verify condition~\ref{cond:injectivity} in Definition \ref{def:gns-smooth-parametric-model}. 
%%%%%%%%%%%%%%%%%%%%%%%%%%%%%%%
Assume that
\begin{equation}
\left\langle v_{[\psi]},\mathrm{d}\ell_a([\psi])\right\rangle=0,
\qquad
\forall a\in\mathscr{A}_{\mathrm{sa}}.
\end{equation}
%%%%%%%%%%%%%%%%%%%%%%%
Equation \eqref{eqn:horizontal-compatibility-pure-states} implies
\begin{equation}
\Re\langle 2\dot\psi_{\mathrm h},a\psi\rangle=0,
\qquad
\forall a\in\mathscr{A}_{\mathrm{sa}},
\end{equation}
%%%%%%%%%%%%%%%%%
which means that  $2\dot\psi_{\mathrm h}$ is real-orthogonal to
$U_\psi(V_{\rho_\psi})$ since the latter is the closed real span of vectors of the form $a\psi$ with $a=a^*$.
%%%%%%%%%%%%%%%%%%%%%%%%%%%%%%%%
But $2\dot\psi_{\mathrm h}$ belongs to
$U_\psi(V_{\rho_\psi})$, because it is orthogonal to $\psi$, and thus $\dot\psi_{\mathrm h}=0$, which means $v_{[\psi]}=0$, thus proving condition~\ref{cond:injectivity} in Definition \ref{def:gns-smooth-parametric-model}.
%%%%%%%%%%%%%%%%%%%%%%%%%%%%%%%%
Since $\mathscr{A}=\mathcal{B}(\mathcal{H})$ is a $W^*$-algebra and pure vector
states are normal, the model $(M,\mathrm{i},\mathscr{A})$ is normal.
%%%%%%%%%%%%%%%%%%%%%%%%%%%

It remains to compute the induced Hermitian tensor. Let
$v_{[\psi]},w_{[\psi]}\in T_{[\psi]}M$ be represented by
normalized curves $\psi(t)$ and $\phi(t)$ with
$\psi(0)=\phi(0)=\psi$, and let $\dot\psi_{\mathrm h}$ and
$\dot\phi_{\mathrm h}$ be the corresponding horizontal components. 
%%%%%%%%%%%%%%%%%
From equation~\eqref{eqn:pure-state-Xi}, it follows that
\begin{equation}
U_\psi\left(\Xi_{[\psi]}(v_{[\psi]})\right)
=
2\dot\psi_{\mathrm h},
\qquad
U_\psi\left(\Xi_{[\psi]}(w_{[\psi]})\right)
=
2\dot\phi_{\mathrm h}.
\end{equation}
%%%%%%%%%%%%%%%%%%%%%%%% 
Using equation~\eqref{eqn:induced-hermitian-form-on-complex-vectors-riesz}, for real tangent vectors we obtain
\begin{equation}
\begin{split}
K_{[\psi]}(v_{[\psi]},w_{[\psi]})
&=
\left\langle
\Xi_{[\psi]}(w_{[\psi]}),
\Xi_{[\psi]}(v_{[\psi]})
\right\rangle_{\rho_\psi}
\\
&=
4\langle\dot\phi_{\mathrm h},\dot\psi_{\mathrm h}\rangle 
\\
&=4\left(
\langle\dot\phi,\dot\psi\rangle
-
\langle\dot\phi,\psi\rangle
\langle\psi,\dot\psi\rangle\right),
\end{split}
\end{equation}
because $\dot\psi_{\mathrm h}=\dot\psi-\langle\psi,\dot\psi\rangle\psi$ and $\dot\phi_{\mathrm h}=\dot\phi-\langle\psi,\dot\phi\rangle\psi$ according to equation \eqref{eqn:pure-state-horizontal-component}, thus proving equation~\eqref{eqn:pure-state-K-real}. 
%%%%%%%%%%%%%%%%%%%%%%%%%%%%%%%
By sesquilinearity, this identifies the induced Hermitian tensor with
\begin{equation}
K=4\overline{h_{FS}},
\end{equation}
where $h_{FS}$ is as in equation~\eqref{eqn:FS-Hermitian-tensor}.
%%%%%%%%%%%%%%%%%%%%%%%
Taking real and imaginary parts, equation \eqref{eq:induced-bilinear-forms} gives
\begin{equation}
G=4g_{FS},
\qquad
\Omega=-4\omega_{FS}.
\end{equation}
%%%%%%%%%%%%%%%%%%%%%%%%
Since $g_{FS}$ and $\omega_{FS}$ are the smooth Fubini--Study metric tensor and symplectic form, respectively, it follows that $G$ and $\Omega$ are smooth. 
%%%%%%%%%%%%%%%%%%%%%%%%%
Therefore, the model is symmetric regular and skew-symmetric regular in the sense of Definition~\ref{def:regular model}, and also Hermitian regular by Remark~\ref{rem: hermitianregularity-equivalence}.
%%%%%%%%%%%%%%%%%%%%%%%%%%%%%%%%%%
\end{proof}

\begin{remark}
For a complexified tangent vector $z_{[\psi]}=v_{[\psi]}+\mathrm i w_{[\psi]}$, the GNS lift in equation \eqref{eqn:complex-gns-lift-on-complex-vectors} is
\begin{equation} 
L_{[\psi]}^{\mathbb C}(z_{[\psi]})
=
L_{[\psi]}(v_{[\psi]})
+
\mathrm iL_{[\psi]}(w_{[\psi]}).
\end{equation}
Since the Riesz map is conjugate-linear with our convention for the Hilbert product, the corresponding Riesz vector is represented under $U_\psi$ by
\begin{equation}
U_\psi\left(
\widehat\Xi_{[\psi]}^{\mathbb C}(z_{[\psi]})
\right)
=
2\left(
\dot\psi_{\mathrm h}
-
\mathrm i\dot\phi_{\mathrm h}
\right),
\end{equation}
where $\dot\psi_{\mathrm h}$ and $\dot\phi_{\mathrm h}$ represent
$v_{[\psi]}$ and $w_{[\psi]}$, respectively. This is the same convention
responsible for the formula $\widehat s_{v+\mathrm i w}=s_v-\mathrm i s_w$
in the commutative case.
\end{remark}

\begin{remark}\label{rem:relation-with-quantum-geometric-tensor}
The  Hermitian tensor obtained in Proposition~\ref{prop:Fubini--Study} recovers, up to normalization and ordering conventions, the usual \emph{quantum geometric tensor} on pure-state models \cite{PV1980,B1984a,CMP1990a,AS1999,BZ2006,EMM2010,BH2001a,KSMP2017}.
%%%%%%%%%%%%%%%%%%%%%%%%%%%%%
Recall that, for a smooth family of normalized vectors
$\lambda\mapsto\psi(\lambda)$ representing rays in
$\mathbb{CP}(\mathcal{H})$, the quantum geometric tensor is usually written locally as
\begin{equation}
Q_{\mu\nu}
=
\left\langle
\partial_\mu\psi_{\mathrm h},
\partial_\nu\psi_{\mathrm h}
\right\rangle,
\end{equation}
where
\begin{equation}
\partial_\mu\psi_{\mathrm h}
=
\partial_\mu\psi
-
\langle\psi,\partial_\mu\psi\rangle\psi
\end{equation}
is the horizontal component of $\partial_\mu\psi$. 
%%%%%%%%%%%%%%%
Equivalently 
\begin{equation}
Q_{\mu\nu}
=
\left\langle
\partial_\mu\psi,
\left(\mathbb I-|\psi\rangle\langle\psi|\right)
\partial_\nu\psi
\right\rangle .
\end{equation}
%%%%%%%%%%%%%%%%%%%%%%%%%%%%
From Proposition~\ref{prop:Fubini--Study}, and recalling the Fubini--Study Hermitian tensor in equation \eqref{eqn:FS-Hermitian-tensor}, it follows that $K=4\overline{Q}$. 
%%%%%%%%%%%%%%%%%%

The point of the present construction is to show that the usual quantum geometric tensor is obtained from the same dual GNS mechanism that gives the Fisher--Rao tensor in the commutative case and the SLD metric for faithful quantum states in finite dimensions. 
%%%%%%%%%%%%%%%%%%%%%%%%%%
\end{remark}

\section{Faithful quantum states}\label{sec:quantum-faithful-states}

We now turn to the case of parametric models of faithful normal states on $\mathscr{A}\cong\mathcal B(\mathcal{H})$, where $\mathcal H$ is a separable complex Hilbert space.
%%%%%%%%%%%%%%%%%%%%%%%%%%%%%%%
A normal state on $\mathcal B(\mathcal H)$ is uniquely represented by a positive trace-class operator of unit trace.
%%%%%%%%%%%%%%%%%%%%%%%%%%%%%%%%%%
Thus, a normal parametric model in this setting consists of a smooth manifold $M$ together with a map $\mathrm{i}:M\to \mathscr S(\mathscr{A})$, $m\mapsto \rho_m$ 
such that, for each $m\in M$, there exists a density operator
\begin{equation}
\varrho_m\in \mathcal{B}_1(\mathcal H),
\qquad
\varrho_m\geq 0,
\qquad
\operatorname{Tr}(\varrho_m)=1,
\end{equation}
where $\mathcal{B}_1(\mathcal H)$ denotes the space of trace-class operators on $\mathcal{H}$, satisfying
\begin{equation}
\rho_m(a)=\operatorname{Tr}(\varrho_m a),
\qquad
a\in\mathscr{A}.
\end{equation}
The state $\rho_m$ is faithful precisely when $\varrho_m$ is injective. 
%%%%%%%%%%%%%%%%%%%%%%
Notice that, in infinite dimension, injectivity does not imply that $\varrho_m^{-1}$ is bounded; equivalently, a faithful density operator need not be bounded below.
%%%%%%%%%%%%%%%%%%%%%%%%%%%%%%%%%%%%%%%%%

Fix $m\in M$ and write
\begin{equation}
\rho:=\rho_m,
\qquad
\varrho:=\varrho_m.
\end{equation}
%%%%%%%%%%%%%%%%%%%%%%%%
Because $\rho$ is faithful, its GNS ideal in equation \eqref{eqn:gelfand-ideal} is trivial. 
%%%%%%%%%%%%%%%%%%%%%%%
Indeed, if $a\in\mathscr{A}$ satisfies
\begin{equation}
\rho(a^*a)=\operatorname{Tr}(\varrho a^*a)=0,
\end{equation}
then $a\varrho^{1/2}=0$. 
%%%%%%%%%%%%%%%%%%%%%%%%%%%
Since $\varrho$ is injective, $\varrho^{1/2}$ is injective and has dense range, and hence $a=0$. 
%%%%%%%%%%%%%%%%%%%%%%%%%
Therefore the GNS Hilbert space $\mathcal H_\rho$ is the completion of $\mathscr{A}=\mathcal B(\mathcal H)$ with respect to the inner product
\begin{equation}
\langle a,b\rangle_\rho
=
\rho(a^*b)
=
\operatorname{Tr}(\varrho a^*b).
\end{equation}
%%%%%%%%%%%%%%%%%%%%%%%%%%

There is a canonical Hilbert--Schmidt realization of this GNS space. 
%%%%%%%%%%%%%%%%%%%%%
On the dense subspace of GNS vectors define
\begin{equation}
U_\varrho:\{\psi_a^\rho=[a]_\rho:\ a\in\mathscr{A}=\mathcal B(\mathcal H)\}
\longrightarrow \mathcal{B}_2(\mathcal H),
\qquad
U_\varrho(\psi_a^\rho):=a\varrho^{1/2},
\end{equation}
where $\mathcal{B}_2(\mathcal H)$ is the space of Hilbert--Schmidt operators on $\mathcal{H}$.
%%%%%%%%%%%%%%%%%%%%%%%%%%%%%%%%%%%%%%
For all $a,b\in\mathscr{A}$ one has
\begin{equation}
\langle \psi_a^\rho,\psi_b^\rho\rangle_\rho
=
\operatorname{Tr}(\varrho a^*b)
=
\operatorname{Tr}\!\left((a\varrho^{1/2})^*(b\varrho^{1/2})\right),
\end{equation}
so $U_\varrho$ is an isometry on a dense subspace. 
%%%%%%%%%%%%%%%%%%%%%
It remains to observe that its range is dense in $\mathcal{B}_2(\mathcal H)$. 
%%%%%%%%%%%%%%%%%%%%%%%%%%%%
Since $\varrho^{1/2}$ is injective, one has
\begin{equation}\label{eqn:varrho-square-root-dense-range}
\overline{\operatorname{Ran}(\varrho^{1/2})}
=
\ker(\varrho^{1/2})^\perp
=
\mathcal H.
\end{equation}
%%%%%%%%%%%%%%%%%%%
Therefore, if $|u\rangle\langle v|$ is a rank-one operator, we may choose a sequence $(w_n)_n\subset\mathcal H$ such that $\varrho^{1/2}w_n\to v$, and then
\begin{equation}
|u\rangle\langle w_n|\,\varrho^{1/2}
\longrightarrow
|u\rangle\langle v|
\end{equation}
in Hilbert--Schmidt norm. 
%%%%%%%%%%%%%%%%%%%%%%%
Hence all rank-one operators, and therefore all finite-rank operators, belong to the Hilbert--Schmidt closure of
\begin{equation}
\{a\varrho^{1/2}:a\in\mathscr{A}=\mathcal B(\mathcal H)\}.
\end{equation}
%%%%%%%%%%%%%%%%%%%
Since finite-rank operators are dense in $\mathcal{B}_2(\mathcal H)$, the isometry $U_\varrho$ extends uniquely to a unitary isomorphism
\begin{equation}\label{eqn:GNS-HS-identification-faithful}
U_\varrho:\mathcal H_\rho\longrightarrow \mathcal{B}_2(\mathcal H).
\end{equation}
%%%%%%%%%%%%%%%%%%%%%%%%%%%%%%%%%%%%

We now transport the distinguished real subspace $V_\rho$ of equation~\eqref{eqn:V-rho-definition} through the unitary map $U_\varrho$. 
%%%%%%%%%%%%%%%%%%%%%%%%%%
Since $U_\varrho$ is a real-linear homeomorphism on the realified Hilbert spaces, one obtains
\begin{equation}\label{eq:HS_image_of_Vrho}
U_\varrho(V_\rho)
=
\overline{\operatorname{span}_{\mathbb R}
\{a\varrho^{1/2}:a=a^*\in\mathcal B(\mathcal H)\}}^{\|\cdot\|_{HS}^{\mathbb{R}}}
\subset \mathcal{B}_{2}^{\mathbb{R}}(\mathcal H).
\end{equation}
This closed real Hilbert subspace is the faithful quantum analogue of the real $\mathcal{L}^2$ space in which the classical score function lives.
%%%%%%%%%%%%%%%%%%%%%%%

Let now $(M,\mathrm i,\mathscr{A})$ be a GNS-smooth parametric model of normal faithful states as in Definition~\ref{def:gns-smooth-parametric-model}. 
%%%%%%%%%%%%%%%%%%%%%%%%%%%%%
For every $m\in M$ and every $v_m\in T_m M$, Proposition~\ref{prop:boundedness-canonical-representative} gives the real GNS vector lift $\xi_{v_{m}}^{\rho_{m}}=\Xi_m(v_m)\in V_{\rho_m}\subset\mathcal{H}_{\rho_{m}}^{\mathbb{R}}$, and we define its Hilbert--Schmidt representative by
\begin{equation}\label{eq:def_X_vm_pre}
X_{v_m}
:=
U_{\varrho_m}\bigl(\Xi_m(v_m)\bigr)
\in \mathcal{B}_{2}^{\mathbb{R}}(\mathcal H).
\end{equation}
%%%%%%%%%%%%%%%%%%%
By construction, it holds
\begin{equation}
X_{v_m}
\in
\overline{\operatorname{span}_{\mathbb R}
\{a\varrho_m^{1/2}:a=a^*\in\mathcal B(\mathcal H)\}}^{\|\cdot\|_{HS}^{\mathbb{R}}}.
\end{equation}
%%%%%%%%%%%%%%%%%%%%%%%%%%
Moreover, in Hilbert--Schmidt form,   equation \eqref{eqn:canonical-real-representative} becomes 
\begin{equation}\label{eq:HS_quantum_compatibility_pre}
\langle v_m,\mathrm{d}\ell_a(m)\rangle
=
\Re\operatorname{Tr}\!\left(X_{v_m}^*a\varrho_m^{1/2}\right),
\end{equation}
for all $a\in\mathscr{A}_{\mathrm{sa}}=\mathcal{B}(\mathcal{H})_{\mathrm{sa}}$, where $\ell_a(m):=\rho_m(a)=\operatorname{Tr}(\varrho_m a)$.
%%%%%%%%%%%%%%%%%%%

This formula should be read as the faithful quantum counterpart of the classical score identity in equation \eqref{eqn:compatibility-condition-classical-score}.
%%%%%%%%%%%%%%%%%%%%%%%%%%%%%%
The role of the classical score $s_{v_m}$ is now played first by the Hilbert--Schmidt operator $X_{v_m}$.
%%%%%%%%%%%%%%%%%%%

When the map $\varrho:M\longrightarrow \mathcal B_1(\mathcal H)_{\mathrm{sa}}$ determined by the model through
$m\mapsto \varrho_m$ is smooth, where $\mathcal B_1(\mathcal H)_{\mathrm{sa}}$ is regarded as a real Banach
space with the trace norm, then for every $v_m\in T_mM$ one has
\begin{equation}\label{eqn:rho-dot-faithful}
\partial_{v_m}\varrho_m
:=
T_m\varrho(v_m)
\in \mathcal B_1(\mathcal H)_{\mathrm{sa}}.
\end{equation}
Moreover, since $\operatorname{Tr}(\varrho_m)=1$ for all $m\in M$, differentiating
the normalization condition gives
\begin{equation}
\operatorname{Tr}(\partial_{v_m}\varrho_m)=0.
\end{equation}
%%%%%%%%%%%%%%%%%%%%%%%%%%

For every $a=a^*\in\mathcal B(\mathcal H)$, the map
\begin{equation}
\Lambda_a:\mathcal B_1(\mathcal H)_{\mathrm{sa}}\longrightarrow \mathbb R,
\qquad
\Lambda_a(W):=\operatorname{Tr}(Wa),
\end{equation}
is a continuous real-linear functional. 
%%%%%%%%%%%%%%%%%
Hence, for the expectation-value map
\begin{equation}
\ell_a(m):=\rho_m(a)=\operatorname{Tr}(\varrho_m a),
\end{equation}
one obtains
\begin{equation}\label{eq:quantum_derivative_under_trace_pre}
\langle v_m,\mathrm{d}\ell_a(m)\rangle
=
\operatorname{Tr}\!\left((\partial_{v_m}\varrho_m)a\right)
\end{equation}
for all $a\in\mathscr{A}_{\mathrm{sa}}=\mathcal B(\mathcal H)_{\mathrm{sa}}$.
%%%%%%%%%%%%%%%%%%%%%%%%%%%
Combining this identity with equation~\eqref{eq:HS_quantum_compatibility_pre}
leads to the weak symmetric logarithmic derivative relation for the Hilbert--Schmidt representative $X_{v_m}$. 
%%%%%%%%%%%%%%%%%%%%%%%%
This is made precise in the proposition following the Lemma.
%%%%%%%%%%%%%%%%%%%

\begin{lemma}\label{lem:trace-class-separated-by-selfadjoints}
Let $T\in\mathcal B_1(\mathcal H)_{\mathrm{sa}}$ be a self-adjoint
trace-class operator.  If
\[
\operatorname{Tr}(Ta)=0,
\qquad
\forall a\in\mathcal B(\mathcal H)_{\mathrm{sa}},
\]
then $T=0$.
\end{lemma}

\begin{proof}
If $b\in\mathcal B(\mathcal H)$ is arbitrary, write
\[
b=b_1+\mathrm i b_2,
\qquad
b_1:=\frac12(b+b^*),
\qquad
b_2:=\frac{1}{2\mathrm i}(b-b^*),
\]
with $b_1,b_2\in\mathcal B(\mathcal H)_{\mathrm{sa}}$.  By assumption,
\[
\operatorname{Tr}(Tb_1)=\operatorname{Tr}(Tb_2)=0,
\]
and hence $\operatorname{Tr}(Tb)=0$ for all
$b\in\mathcal B(\mathcal H)$.  Since
$\mathcal B_1(\mathcal H)^*\cong\mathcal B(\mathcal H)$ through the trace
duality pairing, this implies $T=0$.
\end{proof}

\begin{proposition}\label{prop:quantum_HS_score}
Let $(M,\mathrm{i},\mathscr{A})$ be a GNS-smooth parametric model of normal faithful states on
$\mathscr{A}\cong\mathcal B(\mathcal H)$ as in
Definition~\ref{def:gns-smooth-parametric-model}. 
%%%%%%%%%%%%%
Assume that the density-operator map
\begin{equation}
\varrho:M\longrightarrow \mathcal B_1(\mathcal H)_{\mathrm{sa}},
\qquad
m\longmapsto\varrho_m,
\end{equation}
is smooth, where $\mathcal B_1(\mathcal H)_{\mathrm{sa}}$ is regarded as a real Banach space with the trace norm.
%%%%%%%%%%%%%%%%%%%%%%%%%%%

Then, for every $m\in M$ and every $v_m\in T_mM$, the Hilbert--Schmidt representative $X_{v_m}$ defined in equation~\eqref{eq:def_X_vm_pre} satisfies
the weak symmetric logarithmic derivative relation
\begin{equation}\label{eq:weak_sld_X}
2\,\partial_{v_m}\varrho_m
=
\varrho_m^{1/2}X_{v_m}^*
+
X_{v_m}\varrho_m^{1/2},
\end{equation}
where $\partial_{v_m}\varrho_m$ is as in equation \eqref{eqn:rho-dot-faithful}.
%%%%%%%%%%%%%%%%%%%%%%%%

Moreover, for a complexified tangent vector $z_m=v_m+\mathrm i w_m\in T_m^{\mathbb C}M$, the Hilbert--Schmidt image of equation \eqref{eqn:complex-gns-vector-lift} is
\begin{equation}\label{eq:X_complexified_correct}
\widehat X_{z_m}\equiv  U_\varrho\bigl(\Xi_m^{\mathbb C}(z_m)\bigr)
:=
X_{v_m}-\mathrm i X_{w_m}.
\end{equation}
%%%%%%%%%%%%%%%%%%%%%%%%%%%%%%%%
Consequently, for $z_m,z'_m\in T_m^{\mathbb C}M$, the induced Hermitian form in equation \eqref{eqn:induced-hermitian-form-on-complex-vectors-riesz} is
\begin{equation}\label{eq:K_quantum_HS}
K_m(z_m,z'_m)
=
\operatorname{Tr}\!\left(\widehat X_{z'_m}^{\,*}\widehat X_{z_m}\right).
\end{equation}
%%%%%%%%%%%%%%%%%%%%%%%%%%%
In particular, for real tangent vectors $v_m,w_m\in T_mM$, one has
\begin{equation}\label{eq:K_quantum_HS_real}
K_m(v_m,w_m)
=
\operatorname{Tr}\!\left(X_{w_m}^{*}X_{v_m}\right),
\end{equation}
and therefore equation \eqref{eq:induced-bilinear-forms} gives
\begin{equation}\label{eq:G_Omega_quantum_HS}
\begin{split}
G_m(v_m,w_m)
&=
\Re\operatorname{Tr}\!\left(X_{w_m}^{*}X_{v_m}\right), 
\\
\Omega_m(v_m,w_m)
&=
\Im\operatorname{Tr}\!\left(X_{w_m}^{*}X_{v_m}\right).
\end{split}
\end{equation}
\end{proposition}

\begin{proof}
Fix $m\in M$ and $v_m\in T_mM$, and write $\varrho:=\varrho_m$.
%%%%%%%%%%%%%%%%%%%%%%%
Recalling that  $\varrho:M\to\mathcal B_1(\mathcal H)_{\mathrm{sa}}$ is smooth in trace norm, it holds
\begin{equation}\label{eq:HS_compatibility_used_in_proof}
\operatorname{Tr}\!\left((\partial_{v_m}\varrho_m)a\right)\stackrel{\eqref{eq:quantum_derivative_under_trace_pre}}{=}\langle v_m,\mathrm d\ell_a(m)\rangle
\stackrel{\eqref{eq:HS_quantum_compatibility_pre}}{=}
\Re\operatorname{Tr}\!\left(X_{v_m}^*a\varrho^{1/2}\right)
\end{equation}
for all $a\in\mathscr{A}_{\mathrm{sa}}=\mathcal B(\mathcal H)_{\mathrm{sa}}$.
%%%%%%%%%%%%%%%%%%
Since $X_{v_m}$ and $\varrho^{1/2}$ are Hilbert--Schmidt, the products $X_{v_m}^*a\varrho^{1/2}$ and $\varrho^{1/2}aX_{v_m}$ are trace-class for all $a\in\mathscr{A}=\mathcal{B}(\mathcal{H})$. 
%%%%%%%%%%%%%%%%%
Therefore, for $a=a^*$, it holds
\begin{align}
\Re\operatorname{Tr}\!\left(X_{v_m}^*a\varrho^{1/2}\right)
&=
\frac12\left(
\operatorname{Tr}\!\left(X_{v_m}^*a\varrho^{1/2}\right)
+
\overline{\operatorname{Tr}\!\left(X_{v_m}^*a\varrho^{1/2}\right)}
\right)
\nonumber\\
&=
\frac12\left(
\operatorname{Tr}\!\left(X_{v_m}^*a\varrho^{1/2}\right)
+
\operatorname{Tr}\!\left(\varrho^{1/2}aX_{v_m}\right)
\right)
\nonumber\\
&=
\frac12
\operatorname{Tr}\!\left(
(\varrho^{1/2}X_{v_m}^*+X_{v_m}\varrho^{1/2})a
\right).
\end{align}
%%%%%%%%%%%%%%%%%%%%%%%%%%%%
Equation~\eqref{eq:HS_compatibility_used_in_proof} thus gives
\begin{equation}
\operatorname{Tr}\!\left((\partial_{v_m}\varrho_m)a\right)
=
\frac12
\operatorname{Tr}\!\left(
(\varrho^{1/2}X_{v_m}^*+X_{v_m}\varrho^{1/2})a
\right)
\end{equation}
for all $a\in\mathscr{A}_{\mathrm{sa}}=\mathcal B(\mathcal H)_{\mathrm{sa}}$, 

Set
\[
T
:=
2\,\partial_{v_m}\varrho_m
-
\varrho^{1/2}X_{v_m}^*
-
X_{v_m}\varrho^{1/2}.
\]
The operator $T$ is self-adjoint and trace-class.  The preceding identity
says that
\[
\operatorname{Tr}(Ta)=0,
\qquad
\forall a\in\mathcal B(\mathcal H)_{\mathrm{sa}}.
\]
By Lemma~\ref{lem:trace-class-separated-by-selfadjoints}, $T=0$, and hence
\begin{equation}
2\,\partial_{v_m}\varrho_m
=
\varrho^{1/2}X_{v_m}^*
+
X_{v_m}\varrho^{1/2}.
\end{equation}%%%%%%%%%%%%%%%%%
This proves equation~\eqref{eq:weak_sld_X}.
%%%%%%%%%%%%%%%%%%%

Consider now $z_m=v_m+\mathrm i w_m\in T_m^{\mathbb C}M$.
%%%%%%%%%%%%%%%%%%%%%%%%%%%%%%%%%%%%%%%%%%%%%%%%%%%%%%
Since $U_\varrho$ is complex-linear, equation~\eqref{eqn:complex-gns-vector-lift} leads to
\begin{equation}
U_\varrho\bigl(\Xi_m^{\mathbb C}(z_m)\bigr)
\stackrel{\eqref{eq:def_X_vm_pre}}{=}
X_{v_m}-\mathrm iX_{w_m},
\end{equation}
thus proving equation~\eqref{eq:X_complexified_correct}.
%%%%%%%%%%%%%%%%%%%%%%%%%%%

Finally, using equation~\eqref{eqn:induced-hermitian-form-on-complex-vectors-riesz} and the unitary
identification $U_\varrho:\mathcal H_\rho\to\mathcal B_2(\mathcal H)$, we get
\begin{equation}
K_m(z_m,z'_m)
=
\left\langle
\Xi_m^{\mathbb C}(z'_m),
\Xi_m^{\mathbb C}(z_m)
\right\rangle_{\rho}
=
\operatorname{Tr}\!\left(\widehat X_{z'_m}^{\,*}\widehat X_{z_m}\right),
\end{equation}
thus proving equation~\eqref{eq:K_quantum_HS}. 
%%%%%%%%%%%%%%%%
Restricting to real tangent vectors gives
\begin{equation}
K_m(v_m,w_m)
=
\operatorname{Tr}\!\left(X_{w_m}^*X_{v_m}\right),
\end{equation}
and taking real and imaginary parts gives equation~\eqref{eq:G_Omega_quantum_HS}.
\end{proof}

\begin{remark}\label{rem:quasi SLD from HS representative}
The Hilbert--Schmidt representative $X_{v_m}$ is the primary object in the faithful infinite-dimensional setting. Since $\varrho_m$ is faithful, $\varrho_m^{1/2}$ is injective and has dense range. Hence $X_{v_m}$ determines a densely defined logarithmic-derivative-type operator on $\operatorname{Ran}(\varrho_m^{1/2})$ by
\begin{equation}\label{eq:quasi-SLD}
L_{v_m}:\operatorname{Ran}(\varrho_m^{1/2})\longrightarrow\mathcal H,
\qquad
L_{v_m}(\varrho_m^{1/2}\xi):=X_{v_m}\xi .
\end{equation}
Equivalently,
\begin{equation}
L_{v_m}=X_{v_m}\varrho_m^{-1/2},
\qquad
X_{v_m}=L_{v_m}\varrho_m^{1/2}
\end{equation}
on $\operatorname{Ran}(\varrho_m^{1/2})$.

The operator $L_{v_m}$ is symmetric on $\operatorname{Ran}(\varrho_m^{1/2})$. Indeed, by construction,
\begin{equation}
X_{v_m}\in
\overline{\operatorname{span}_{\mathbb R}
\{a\varrho_m^{1/2}:a=a^*\in\mathcal B(\mathcal H)\}}^{\|\cdot\|_{HS}} .
\end{equation}
Thus there are bounded self-adjoint operators $a_n=a_n^*$ such that
\begin{equation}
a_n\varrho_m^{1/2}\longrightarrow X_{v_m}
\end{equation}
in Hilbert--Schmidt norm, hence also in operator norm. For $\xi,\eta\in\mathcal H$, one has
\begin{equation}
\langle a_n\varrho_m^{1/2}\xi,\varrho_m^{1/2}\eta\rangle
=
\langle \varrho_m^{1/2}\xi,a_n\varrho_m^{1/2}\eta\rangle .
\end{equation}
Passing to the limit gives
\begin{equation}
\langle L_{v_m}\varrho_m^{1/2}\xi,\varrho_m^{1/2}\eta\rangle
=
\langle \varrho_m^{1/2}\xi,L_{v_m}\varrho_m^{1/2}\eta\rangle ,
\end{equation}
which is the symmetry of $L_{v_m}$ on $\operatorname{Ran}(\varrho_m^{1/2})$.

The weak SLD relation of Proposition~\ref{prop:quantum_HS_score} may therefore be read as
\begin{equation}
2\,\partial_{v_m}\varrho_m
=
\varrho_m^{1/2}X_{v_m}^*
+
X_{v_m}\varrho_m^{1/2},
\end{equation}
that is, as the SLD equation written in terms of the Hilbert--Schmidt amplitude $X_{v_m}=L_{v_m}\varrho_m^{1/2}$. In general, however, $L_{v_m}$ need not be bounded, and products such as $\varrho_m L_{v_m}L_{w_m}$ need not be intrinsically meaningful. Thus, in infinite dimension, $X_{v_m}$ is the robust object, while $L_{v_m}$ should be treated as a derived densely defined operator.

If, in addition,
\begin{equation}
\varrho_m=\sum_j p_j |e_j\rangle\langle e_j|,
\qquad
p_j>0,
\end{equation}
and $\dot\varrho:=\partial_{v_m}\varrho_m$, then the formal matrix coefficients of $L_{v_m}$ are
\begin{equation}
(L_{v_m})_{ij}
=
\frac{2\dot\varrho_{ij}}{p_i+p_j}.
\end{equation}
Consequently,
\begin{equation}\label{eqn:HS-integrability-SLD}
\|X_{v_m}\|_{HS}^2
=
4\sum_{i,j}
\frac{|\dot\varrho_{ij}|^2p_j}{(p_i+p_j)^2}.
\end{equation}
Thus, in this eigenbasis representation, GNS-boundedness is equivalent to the finiteness of the right-hand side of equation~\eqref{eqn:HS-integrability-SLD}. In finite dimensions this condition is automatic on faithful states; in infinite dimensions it is a genuine restriction on admissible tangent directions.

Finally, if $L_{v_m}$ is bounded on $\operatorname{Ran}(\varrho_m^{1/2})$ with respect to the ambient Hilbert norm, then it extends uniquely to a bounded self-adjoint operator on $\mathcal H$, and the weak SLD relation becomes the usual symmetric logarithmic derivative equation
\begin{equation}
2\,\partial_{v_m}\varrho_m
=
\varrho_m L_{v_m}+L_{v_m}\varrho_m
\end{equation}
in $\mathcal B_1(\mathcal H)_{\mathrm{sa}}$.
\end{remark}

\subsection{Faithful quantum states in finite dimensions}
\label{subsec:faithful-quantum-states-finite-dimensional}

We now focus on the finite-dimensional case where  $\mathcal H$ is a finite-dimensional complex Hilbert space, and
\begin{equation}
\mathscr S_f(\mathcal H)
:=
\left\{
\varrho\in\mathcal B(\mathcal H)_{\mathrm{sa}}
:\;
\varrho>0,\;
\operatorname{Tr}(\varrho)=1
\right\}
\end{equation}
denotes the space of faithful density operators on $\mathcal H$, equivalently, the space of faithful states\footnote{Being $\mathcal{B}(\mathcal{H})$ finite-dimensional, all states are normal.} on $\mathscr{A}=\mathcal{B}(\mathcal{H})$. 
%%%%%%%%%%%%%%%%%%%%%%%%
In this finite-dimensional situation, 
$\mathscr S_f(\mathcal H)$ is an open submanifold of the affine hyperplane $\left\{
A\in\mathcal B(\mathcal H)_{\mathrm{sa}}:
\operatorname{Tr}(A)=1
\right\}$ \cite{GKM2005,CCIMV2019,CDILM2017}.
%%%%%%%%%%%%%%%%%%%%%%%%%%%%%
Thus, for every $\varrho\in\mathscr S_f(\mathcal H)$, it holds
\begin{equation}
\begin{split}
T_\varrho\mathscr S_f(\mathcal H)
&=
\left\{
v\in\mathcal B(\mathcal H)_{\mathrm{sa}}:
\operatorname{Tr}(v)=0
\right\},\\ & \\
T_\varrho^{\mathbb C}\mathscr S_f(\mathcal H)
&\cong
\left\{
z\in\mathcal B(\mathcal H):
\operatorname{Tr}(z)=0
\right\}.
\end{split}
\end{equation}
%%%%%%%%%%%%%

\begin{proposition}\label{prop:finite-dimensional-faithful-states}
Let $\mathcal{H}$ be finite-dimensional.
%%%%%%%%%%%%%%%%%%%%
The parametric model $(M,\mathrm{i},\mathscr{A})$ of faithful normal states where $M=\mathscr S_f(\mathcal H)$, $\mathscr{A}=\mathcal{B}(\mathcal{H})$, and $\mathrm i:\mathscr S_f(\mathcal H)\longrightarrow \mathscr S(\mathscr A)$ is given by 
\begin{equation}
\mathrm i(\varrho)=\rho_\varrho,
\qquad
\rho_\varrho(a):=\operatorname{Tr}(\varrho a),
\end{equation}
is a Hermitian regular GNS-smooth normal parametric model as in Definitions~\ref{def:gns-smooth-parametric-model} and~\ref{def:regular model}.
%%%%%%%%%%%%%%%%%

For every $\varrho\in M$ and every
$v\in T_\varrho M$, the Hilbert--Schmidt representative
$X_v$ defined in equation~\eqref{eq:def_X_vm_pre} is of the form
\begin{equation}\label{eqn:SLD-X-fin-dim-prop}
X_v=L_v\varrho^{1/2},
\end{equation}
where $L_v=L_v^*$ is the unique self-adjoint solution of the symmetric logarithmic derivative equation
\begin{equation}\label{eq:finite-dimensional-SLD-equation}
\frac{1}{2}(\varrho L_v+L_v\varrho)=v.
\end{equation}
%%%%%%%%%%%%%%%%%
Equivalently, if
\begin{equation}
\mathcal{J}_\varrho:\mathcal B(\mathcal H)_{\mathrm{sa}}
\longrightarrow
\mathcal B(\mathcal H)_{\mathrm{sa}},
\qquad
\mathcal{J}_\varrho(A):=\frac12(\varrho A+A\varrho),
\end{equation}
then
\begin{equation}
L_v=\mathcal{J}_\varrho^{-1}(v).
\end{equation}
%%%%%%%%%%%%%%%%%%%
For $z=v+\mathrm i w\in T_\varrho^{\mathbb C} M$, set $L_z:=L_v+\mathrm iL_w$.
%%%%%%%%%%%%%
Then, equation \eqref{eq:X_complexified_correct} reads
\begin{equation}
\widehat X_z
=
L_{\overline z}\varrho^{1/2},
\qquad
L_{\overline z}:=L_v-\mathrm iL_w,
\end{equation}
and  equations \eqref{eq:K_quantum_HS}, \eqref{eq:K_quantum_HS_real}, and \eqref{eq:G_Omega_quantum_HS} read
\begin{equation} \label{eqn:induced-tensor-fin-dim-faithful-states}
\begin{split} 
K_\varrho(z,z')
&=
\operatorname{Tr}\!\left(\widehat X_{z'}^{\,*}\widehat X_z\right)
=
\operatorname{Tr}\!\left(
\varrho\,L_{\overline{z'}}^{\,*}L_{\overline z}
\right) \\
K_\varrho(v,w)
&=
\operatorname{Tr}\!\left( X_{w}^{\,*}  X_v\right)
=
\operatorname{Tr}(\varrho L_wL_v) \\
G_\varrho(v,w)
&=
\Re\operatorname{Tr}(\varrho L_wL_v) =
\frac{1}{2}\operatorname{Tr}\!\left(
\varrho(L_vL_w+L_wL_v)
\right)\\
\Omega_\varrho(v,w)
&=
\Im\operatorname{Tr}(\varrho L_wL_v)
=
-\frac{1}{2i}
\operatorname{Tr}\!\left(\varrho[L_v,L_w]\right),
\end{split}
\end{equation}
%%%%%%%%%%%%%%%%%%%%%%%%%%%
In particular, $G$ coincides with the SLD quantum Fisher information metric \cite{P2009a,H1967a}. 
%%%%%%%%%%%%%%%%%%%%%%%
With the convention in which the Bures metric is normalized through the infinitesimal Bures distance, the latter is $G/4$.
%%%%%%%%%%%%%%%%%%%%%%%%%%%%%%%%%%%%
\end{proposition}

\begin{proof}
Recalling that $M$ is an open submanifold of the affine hyperplane  $\left\{
A\in\mathcal B(\mathcal H)_{\mathrm{sa}}:
\operatorname{Tr}(A)=1
\right\}$, for every $a\in\mathscr{A}_{\mathrm{sa}}=\mathcal B(\mathcal H)_{\mathrm{sa}}$, the expectation-value map $\ell_a(\varrho)=\operatorname{Tr}(\varrho a)$ is smooth, thus proving condition \ref{cond:smoothness-of-expectations} in Definition \ref{def:gns-smooth-parametric-model}.
%%%%%%%%%%%%%%%%%%%
Moreover, for
$v\in T_\varrho M$ one has
\begin{equation}\label{eqn:derivative-quant-exp-value-fin-dim}
\langle v,\mathrm d\ell_a(\varrho)\rangle
=
\operatorname{Tr}(va).
\end{equation}
%%%%%%%%%%%%%%%%%%%%%%%%

Since $\varrho>0$, the linear map
\begin{equation}
\mathcal{J}_\varrho:\mathcal B(\mathcal H)_{\mathrm{sa}}
\longrightarrow
\mathcal B(\mathcal H)_{\mathrm{sa}},
\qquad
\mathcal{J}_\varrho(a):=\frac12(\varrho a+a\varrho),
\end{equation}
is an isomorphism. 
%%%%%%%%%%%%%%%%
Indeed, using the spectral decomposition of $\varrho$ given by 
\begin{equation}
\varrho=\sum_i p_i |e_i\rangle\langle e_i|,
\qquad
p_i>0,
\end{equation}
then
\begin{equation}
(\mathcal{J}_\varrho(a))_{ij}
=
\frac{p_i+p_j}{2}a_{ij},
\end{equation}
where the components are computed with respect to the orthonormal basis diagonalizing $\varrho$.
%%%%%%%%%%%%%%%%%%%%%%%%%
Hence $\mathcal{J}_\varrho(a)=0$ implies $a=0$, and finite-dimensionality gives surjectivity.
%%%%%%%%%%%%%%%%%%%%%%%%%%%%%%%%%%%%%
Therefore, for every
$v\in T_\varrho M$, there exists a unique self-adjoint
operator $L_v$ satisfying
\begin{equation}\label{eqn:SLD-fin-dim}
\frac{1}{2}(\varrho L_v+L_v\varrho)=v.
\end{equation}
%%%%%%%%%%%%%%%%%%%%%%%
In the eigenbasis of $\varrho$, this operator is given explicitly by
\begin{equation}
(L_v)_{ij}=\frac{2v_{ij}}{p_i+p_j}.
\end{equation}
%%%%%%%%%%%%%%%%%%%%%%%%%%%%%%%%

We now identify the canonical Hilbert--Schmidt representative $X_v$ in equation \eqref{eq:def_X_vm_pre}.
Set
\begin{equation}\label{eqn:SLD-X-fin-dim}
X_v:=L_v\varrho^{1/2}.
\end{equation}
Since $L_v=L_v^*$, one has
\begin{equation}
X_v\in
\left\{
a\varrho^{1/2}:a\in\mathcal B(\mathcal H)_{\mathrm{sa}}
\right\}
=
U_\varrho(V_{\rho_\varrho}).
\end{equation}
Moreover, for every $a=a^*\in\mathcal B(\mathcal H)$,
\begin{align}
\Re\operatorname{Tr}(X_v^*a\varrho^{1/2})
&=
\Re\operatorname{Tr}(\varrho^{1/2}L_v a\varrho^{1/2})
\nonumber\\
&=
\Re\operatorname{Tr}(\varrho L_v a)
\nonumber\\
&=
\frac12\operatorname{Tr}\!\left((\varrho L_v+L_v\varrho)a\right)
\nonumber\\
&=
\operatorname{Tr}(va)
\nonumber\\
&=
\langle v,\mathrm d\ell_a(\varrho)\rangle .
\end{align}
Thus equation~\eqref{eqn:canonical-real-representative} holds with Hilbert--Schmidt image $X_v=L_v\varrho^{1/2}$. By Proposition~\ref{prop:boundedness-canonical-representative}, condition~\ref{cond:GNS-boundedness} in Definition~\ref{def:gns-smooth-parametric-model} follows, and $X_v$ is the Hilbert--Schmidt representative of the canonical GNS vector lift.

Condition \ref{cond:injectivity} in Definition~\ref{def:gns-smooth-parametric-model} can be proved considering  equation \eqref{eqn:derivative-quant-exp-value-fin-dim} with $a=v$.
%%%%%%%%%%%%%%%%%%%%%%%%%%%%%%%%

We now compute the induced bilinear forms. 
%%%%%%%%%%%%%%%%%%%%%%%
Let $z=v+\mathrm i w\in T_\varrho^{\mathbb C}M$.
%%%%%%%%%%%%%%%%%%%%%%%%%%%%%%%%%%%%%%%%
Equations~\eqref{eq:X_complexified_correct} and \eqref{eqn:SLD-X-fin-dim-prop} imply
\begin{equation}
\widehat X_z
=
X_v-\mathrm iX_w
=
(L_v-\mathrm iL_w)\varrho^{1/2}
=
L_{\overline z}\varrho^{1/2}.
\end{equation}
%%%%%%%%%%%%%%%%%%%%%%%%%%
Equation \eqref{eqn:induced-tensor-fin-dim-faithful-states} then follows by direct inspection.
%%%%%%%%%%%%%%%%%%%%%%%%%%

It remains only to note regularity of the model in the sense of Definition \ref{def:regular model}.
%%%%%%%%%%%%%%%%%%%%%%
The map
\begin{equation}
(\varrho,v)\longmapsto L_v=\mathcal{J}_\varrho^{-1}(v)
\end{equation}
is smooth on the finite-dimensional vector bundle
$T\mathscr S_f(\mathcal H)$, because $\varrho\mapsto\mathcal{J}_\varrho$ is
smooth and takes values in the open set of invertible linear maps on
$\mathcal B(\mathcal H)_{\mathrm{sa}}$. 
%%%%%%%%%%%%%%
Hence, the formulas in equation \eqref{eqn:induced-tensor-fin-dim-faithful-states} imply the model is Hermitian regular according to Definition \ref{def:regular model}.
\end{proof}

\begin{remark}\label{rem:antisymmetric-part-SLD-commutators}
The antisymmetric part of the GNS Hermitian tensor has a natural interpretation
in the multiparameter setting. 
%%%%%%%%%%%%%
In the finite-dimensional faithful case, for
$\varrho\in M$ and $u,w\in T_\varrho M$, equation
\eqref{eqn:induced-tensor-fin-dim-faithful-states} gives
\begin{equation}
\Omega_\varrho(u,w)
=
-\frac{1}{2i}
\operatorname{Tr}\!\left(\varrho[L_u,L_w]\right).
\end{equation}
%%%%%%%%%%%%%%%%%%%%%%%%%%%%%
Therefore,  $\Omega_\varrho(u,w)$ is proportional to the expectation value of the commutator of the SLD operators associated with the tangent directions $u$ and $w$.
%%%%%%%%%%%%%%%%%%%%%%%%%%%
Consequently, $\Omega$ measures a local obstruction, at the level of the SLD representatives, to treating different tangent directions as mutually compatible
classical observables. 
%%%%%%%%%%%%%%%%%%%%%%%
If the SLD operators commute pairwise at $\varrho$, namely
if
\begin{equation}
[L_u,L_w]=0,
\qquad
\forall u,w\in T_\varrho M,
\end{equation}
then $\Omega_\varrho$ vanishes identically. 
%%%%%%%%%%%%%%%%%%%%%%
In multiparameter quantum estimation, this is the usual pointwise SLD-commutativity condition associated with local
quasi-classicality \cite{S2019}. 
%%%%%%%%%%%%%%%%%%%%%%
In finite dimension, pairwise commuting self-adjoint SLDs are simultaneously diagonalizable, and the corresponding
SLD-commutator obstruction to joint attainability disappears \cite{P2009a}.
%%%%%%%%%%%%%%%%%%%%%%%%%%%%%%%%%%%%

The converse, however, does not hold in general. 
%%%%%%%%%%%%%%%%%%%%%%%%%%%%%%%%%
The condition
$\Omega_\varrho=0$ only says that
\begin{equation}
\operatorname{Tr}\!\left(\varrho[L_u,L_w]\right)=0,
\qquad
\forall u,w\in T_\varrho M,
\end{equation}
not that the operators $L_u$ and $L_w$ commute. 
%%%%%%%%%%%%%%%%%%%%%%%%%%%%%%%
Therefore isotropic submanifolds of $\Omega$ need not be quasiclassical models in the strong sense of pairwise
commuting SLDs. 
%%%%%%%%%%%%%%%%%%%%%%%%
Moreover, this is a pointwise tangent-space condition. 
%%%%%%%%%%%%%%%%%%%%%%%%%%%%%%
In particular, in one-dimensional models the pairwise commutativity condition is vacuous, and should not be confused with a global statement that the whole family of states is simultaneously diagonalizable.
%%%%%%%%%%%%%%%%%%%%%%%%%%%%%%%%%%
\end{remark}

\begin{remark}[Relation with mean Uhlmann curvature]\label{rem:mean-uhlmann-curvature}
In the finite-dimensional faithful setting, the skew-symmetric tensor obtained
above is closely related to a quantity already used in multiparameter quantum
estimation. In the notation common in that literature, the mean Uhlmann
curvature is often written as
\begin{equation}
U_{\mu\nu}
:=
-\frac{i}{4}\operatorname{Tr}\!\left(\rho[L_\mu,L_\nu]\right),
\end{equation}
where $L_\mu$ denotes the SLD associated with the parameter $\lambda^\mu$.
This quantity appears, for instance, in the compatibility condition for
multiparameter quantum estimation and in the study of mixed-state quantum holonomies and finite-temperature geometric phases \cite{RJD2016,CSDV2019,CSV2018,CSV2018a,LVSC2019,BLVSC2019}.
%%%%%%%%%%%%%%%%%%%%%%%%%%%%%%%%
In particular, equation~\eqref{eqn:induced-tensor-fin-dim-faithful-states}
gives
\begin{equation}
\Omega_{\mu\nu}
=
-\frac{1}{2i}\operatorname{Tr}\!\left(\rho[L_\mu,L_\nu]\right)
=
-2U_{\mu\nu}.
\end{equation}
Thus the two-form $\Omega$ recovers, up to the conventional normalization and
sign, the mean Uhlmann curvature. The additional point of the present
construction is that this skew tensor is not introduced through SLD
commutators or mixed-state holonomy, but arises   as the imaginary
part of the Hermitian tensor induced by the dual GNS fibration.
\end{remark}

\begin{example}[Faithful qubits]\label{exmp:qubit-case-omega-not-closed}
We now specialize Proposition~\ref{prop:finite-dimensional-faithful-states} to the case $\mathcal H=\mathbb C^2$.
%%%%%%%%%%%%%%%%%%%%%
The aim is to show that $\Omega$ is not closed here, in sharp contrast with what happens in the case of pure quantum states in Section \ref{sec:quantum-pure-states}.
%%%%%%%%%%%%%%%%%%%%%%%%%%%%%%%

Let $\sigma_1,\sigma_2,\sigma_3$ be the Pauli matrices, and identify the manifold of faithful qubit states with the interior of the Bloch ball, namely
\begin{equation}
M
=
\left\{
\mathbf r=(r_1,r_2,r_3)\in\mathbb R^3:
|\mathbf r|<1
\right\}.
\end{equation}
The corresponding density operator is
\begin{equation}
\varrho_{\mathbf r}
:=
\frac12(I+\mathbf r\cdot\sigma),
\end{equation}
where
\begin{equation}
\mathbf r\cdot\sigma
:=
r_1\sigma_1+r_2\sigma_2+r_3\sigma_3,
\end{equation}
and the associated normal state on $\mathscr A=\mathcal B(\mathbb C^2)$ is
\begin{equation}
\rho_{\mathbf r}(a):=\operatorname{Tr}(\varrho_{\mathbf r}a).
\end{equation}

Let
\begin{equation}
u=(u_1,u_2,u_3)\in T_{\mathbf r}M\cong\mathbb R^3.
\end{equation}
Then
\begin{equation}
\partial_u\varrho_{\mathbf r}
=
\frac12\,u\cdot\sigma.
\end{equation}
By Proposition~\ref{prop:finite-dimensional-faithful-states},  we have $X_u=L_u\varrho_{\mathbf r}^{1/2}$, where $L_u=L_u^*$ is the unique solution of the SLD equation \eqref{eq:finite-dimensional-SLD-equation}.
%%%%%%%%%%%%%%%%%%%%%%%%%%%%%%%
Writing
\begin{equation}
L_u=\alpha_u I+\beta_u\cdot\sigma,
\end{equation}
and using
\begin{equation}
(\mathbf a\cdot\sigma)(\mathbf b\cdot\sigma)
=
(\mathbf a\cdot\mathbf b)I
+
i(\mathbf a\times\mathbf b)\cdot\sigma,
\end{equation}
one obtains
\begin{equation}
\alpha_u+\mathbf r\cdot\beta_u=0,
\qquad
\beta_u+\alpha_u\mathbf r=u.
\end{equation}
Therefore
\begin{equation}
\alpha_u
=
-\frac{\mathbf r\cdot u}{1-|\mathbf r|^2},
\qquad
\beta_u
=
u+
\frac{\mathbf r\cdot u}{1-|\mathbf r|^2}\mathbf r,
\end{equation}
and hence
\begin{equation}\label{eqn:qubit-SLD-Lu}
L_u
=
-\frac{\mathbf r\cdot u}{1-|\mathbf r|^2}I
+
\left(
u+
\frac{\mathbf r\cdot u}{1-|\mathbf r|^2}\mathbf r
\right)\cdot\sigma.
\end{equation}

For a complexified tangent vector
\begin{equation}
z=u+iw\in T_{\mathbf r}^{\mathbb C}M\cong\mathbb C^3,
\qquad
u,w\in T_{\mathbf r}M,
\end{equation}
we use the complex-linear extension
\begin{equation}
L_z:=L_u+iL_w.
\end{equation}
Equivalently, writing
\begin{equation}
L_z=\alpha_z I+\beta_z\cdot\sigma,
\end{equation}
one has
\begin{equation}
\alpha_z
=
-\frac{\mathbf r\cdot z}{1-|\mathbf r|^2},
\qquad
\beta_z
=
z+
\frac{\mathbf r\cdot z}{1-|\mathbf r|^2}\mathbf r.
\end{equation}
However, the Hilbert--Schmidt Riesz-vector representative corresponding to the
complexified dual lift is not $L_z\varrho_{\mathbf r}^{1/2}$, but rather
\begin{equation}
\widehat X_z
=
L_{\overline z}\varrho_{\mathbf r}^{1/2},
\end{equation}
as in Proposition~\ref{prop:finite-dimensional-faithful-states}.

Consequently, for
$z,z'\in T_{\mathbf r}^{\mathbb C}M$, the induced Hermitian tensor is
\begin{equation}
K_{\mathbf r}(z,z')
=
\operatorname{Tr}\!\left(\widehat X_{z'}^{\,*}\widehat X_z\right)
=
\operatorname{Tr}\!\left(
\varrho_{\mathbf r}L_{\overline{z'}}^{\,*}L_{\overline z}
\right).
\end{equation}
A direct computation gives
\begin{equation}\label{eqn:qubit-Hermitian-tensor}
K_{\mathbf r}(z,z')
=
\overline z\cdot z'
+
\frac{(\mathbf r\cdot\overline z)(\mathbf r\cdot z')}{1-|\mathbf r|^2}
-
i\,\mathbf r\cdot(\overline z\times z').
\end{equation}
The sign of the last term is fixed by the convention
$K_{\mathbf r}(u,v)=\operatorname{Tr}(\varrho_{\mathbf r}L_vL_u)$ on real tangent
vectors.

In particular, for the coordinate vector fields $\partial_{r_i}$, one has
\begin{equation}
L_i:=L_{\partial_{r_i}}
=
\sigma_i+
\frac{r_i}{1-|\mathbf r|^2}
(\mathbf r\cdot\sigma-I).
\end{equation}
Set
\begin{equation}
b_i
:=
e_i+
\frac{r_i}{1-|\mathbf r|^2}\mathbf r,
\end{equation}
so that
\begin{equation}
L_i
=
-\frac{r_i}{1-|\mathbf r|^2}I+b_i\cdot\sigma.
\end{equation}

For real tangent vectors $u,v\in T_{\mathbf r}M$, equation
\eqref{eqn:qubit-Hermitian-tensor} gives
\begin{equation}
G_{\mathbf r}(u,v)
=
\Re K_{\mathbf r}(u,v)
=
u\cdot v+
\frac{(\mathbf r\cdot u)(\mathbf r\cdot v)}{1-|\mathbf r|^2}.
\end{equation}
Thus $G$ is the SLD quantum Fisher metric on the manifold of faithful qubit
states. With the convention in which the Bures metric is normalized through the
infinitesimal Bures distance, this is four times the Bures metric
\cite{BZ2006,D1999a}.

For the antisymmetric tensor, the convention in
Proposition~\ref{prop:finite-dimensional-faithful-states} gives
\begin{equation}
\Omega_{\mathbf r}(u,v)
=
\Im K_{\mathbf r}(u,v)
=
-\frac{1}{2i}
\operatorname{Tr}\!\left(
\varrho_{\mathbf r}[L_u,L_v]
\right).
\end{equation}
Since
\begin{equation}
[L_i,L_j]
=
2i(b_i\times b_j)\cdot\sigma,
\end{equation}
we get
\begin{equation}
\Omega_{ij}(\mathbf r)
:=
\Omega_{\mathbf r}(\partial_{r_i},\partial_{r_j})
=
-\operatorname{Tr}\!\left(
\varrho_{\mathbf r}(b_i\times b_j)\cdot\sigma
\right).
\end{equation}
Using
\begin{equation}
\operatorname{Tr}(\varrho_{\mathbf r}\,\mathbf a\cdot\sigma)
=
\mathbf r\cdot\mathbf a ,
\qquad
\forall \mathbf a\in\mathbb R^3,
\end{equation}
it follows that
\begin{equation}
\Omega_{ij}(\mathbf r)
=
-\mathbf r\cdot(b_i\times b_j).
\end{equation}
The terms containing $\mathbf r$ twice vanish by antisymmetry of the vector
product, so
\begin{equation}
\mathbf r\cdot(b_i\times b_j)
=
\mathbf r\cdot(e_i\times e_j)
=
\varepsilon_{ijk}r_k.
\end{equation}
Therefore
\begin{equation}\label{eqn:qubit-Omega}
\Omega
=
-r_3\,dr_1\wedge dr_2
+
r_2\,dr_1\wedge dr_3
-
r_1\,dr_2\wedge dr_3.
\end{equation}
Consequently,
\begin{equation}
d\Omega
=
-3\,dr_1\wedge dr_2\wedge dr_3
\neq 0.
\end{equation}
Thus, although the faithful qubit model is a finite-dimensional regular model,
the corresponding antisymmetric tensor $\Omega$ need not be closed. This is in
sharp contrast with the pure-state case, where $\Omega$ is, up to normalization,
the Fubini--Study symplectic form.
\end{example}
%%%%%%%%%%%%%%

\subsection{Displaced thermal states and fixed-temperature submodels}
\label{subsec:displaced-thermal-gaussian-model}

Let $\mathcal H:=\ell^2(\mathbb N_0)$, with canonical orthonormal basis
$(e_n)_{n\geq 0}$.\footnote{If $(h_n)_{n\geq 0}$ denotes the Hermite basis of
$L^2(\mathbb R)$, the unitary map $U:L^2(\mathbb R)\to\ell^2(\mathbb N_0)$
given by $U(h_n)=e_n$ identifies the Schr\"odinger realization of the harmonic
oscillator with the number-basis realization used here. We work on
$\ell^2(\mathbb N_0)$ because the thermal reference states are diagonal in this
basis.}
On the dense finite-particle subspace
\begin{equation}
\mathcal D_{\mathrm{fin}}
:=
\operatorname{span}\{e_n:\ n\in\mathbb N_0\}
\end{equation}
we have the annihilation and creation operators defined, respectively, by
\begin{equation}
ae_n=\sqrt n\,e_{n-1},
\qquad
a^\dagger e_n=\sqrt{n+1}\,e_{n+1}.
\end{equation}
We set
\begin{equation}
N:=a^\dagger a,
\qquad
Q:=\frac{a+a^\dagger}{\sqrt2},
\qquad
P:=\frac{a-a^\dagger}{i\sqrt2}.
\end{equation}
Then, on $\mathcal D_{\mathrm{fin}}$,
\begin{equation}
[Q,P]=iI,
\qquad
N=\frac12(Q^2+P^2-I).
\end{equation}

For $(x,y)\in\mathbb R^2$, let
\begin{equation}\label{eqn:weyl-displacement-operator}
W(x,y):=\exp\!\left(i(yQ-xP)\right)
\end{equation}
be the Weyl displacement operator. With this convention,
\begin{equation}
W(x,y)^*QW(x,y)=Q+xI,
\qquad
W(x,y)^*PW(x,y)=P+yI.
\end{equation}
Equivalently,
\begin{equation}
W(x,y)QW(x,y)^*=Q-xI,
\qquad
W(x,y)PW(x,y)^*=P-yI.
\end{equation}

Let
\begin{equation}
M:=\mathbb R^2\times(0,\infty),
\qquad
m=(x,y,\beta).
\end{equation}
For $\beta>0$, set
\begin{equation}
q_\beta:=e^{-\beta},
\qquad
\bar n_\beta:=\frac{q_\beta}{1-q_\beta}
=
\frac{1}{e^\beta-1},
\qquad
t_\beta:=\tanh\!\left(\frac\beta2\right)
=
\frac{1-q_\beta}{1+q_\beta}.
\end{equation}
The thermal density operator is
\begin{equation}\label{eqn:thermal-operator}
\tau_\beta:=(1-q_\beta)q_\beta^N.
\end{equation}
Since
\begin{equation}
\tau_\beta e_n=(1-q_\beta)q_\beta^n e_n,
\end{equation}
all eigenvalues of $\tau_\beta$ are strictly positive and sum to $1$. Thus
$\tau_\beta$ is a faithful density operator on $\mathcal H$.

We define the displaced thermal density operator by
\begin{equation}\label{eqn:thermal_displaced_density_operator}
\varrho_{x,y,\beta}
:=
W(x,y)\tau_\beta W(x,y)^*.
\end{equation}
The associated normal state on $\mathscr A:=\mathcal B(\mathcal H)$ is
\begin{equation}
\rho_{x,y,\beta}(b)
:=
\operatorname{Tr}(\varrho_{x,y,\beta}b),
\qquad
b\in\mathscr A.
\end{equation}
We thus obtain a normal parametric model $(M,\mathrm i,\mathscr A)$ by setting
\begin{equation}
\mathrm i:M\longrightarrow\mathscr S(\mathscr A),
\qquad
\mathrm i(x,y,\beta):=\rho_{x,y,\beta}.
\end{equation}
Since $\tau_\beta$ is faithful and $W(x,y)$ is unitary, each
$\varrho_{x,y,\beta}$ is faithful. The family $\rho_{x,y,\beta}$ is the
three-dimensional family of one-mode displaced thermal states.

We shall prove that this model is GNS-smooth and Hermitian regular. The proof
has three steps. First, we prove smoothness of all expectation-value functions.
Second, we identify the canonical Hilbert--Schmidt representatives of the
coordinate tangent vectors. Third, we compute the induced Hermitian tensor and
its real and imaginary parts.

\begin{lemma}
\label{lem:displaced-thermal-smooth-expectations}
For every $b\in\mathscr{A}_{\mathrm{sa}}=\mathcal B(\mathcal H)_{\mathrm{sa}}$, the function
\begin{equation}
\ell_b:M\longrightarrow\mathbb R,
\qquad
\ell_b(x,y,\beta)
:=
\operatorname{Tr}(\varrho_{x,y,\beta}b)
\end{equation}
is smooth.
\end{lemma}

\begin{proof}
We use the coherent-state representation of the thermal state. Set
\begin{equation}
\alpha:=\frac{x+iy}{\sqrt2},
\qquad
D(\alpha):=W(\sqrt2\,\operatorname{Re}\alpha,\sqrt2\,\operatorname{Im}\alpha),
\end{equation}
so that
\begin{equation}
D(\alpha)=\exp(\alpha a^\dagger-\overline{\alpha}a).
\end{equation}
The normalized coherent states are  
\begin{equation}
|\alpha\rangle:=D(\alpha)e_0,
\end{equation}
and satisfy
\begin{equation}
|\alpha\rangle
=
e^{-|\alpha|^2/2}
\sum_{n=0}^{\infty}
\frac{\alpha^n}{\sqrt{n!}}e_n.
\end{equation}
The thermal state admits the Glauber--Sudarshan representation \cite{G1963a,S1963a,F1989}
\begin{equation}\label{eqn:thermal-state-coherent-representation}
\tau_\beta
=
\frac{1}{\pi\bar n_\beta}
\int_{\mathbb C}
e^{-|\eta|^2/\bar n_\beta}
|\eta\rangle\langle\eta|\,\mathrm{d}\eta .
\end{equation}
Equivalently, using the Weyl relation
\begin{equation}
D(\alpha)D(\eta)
=
\exp\left(
\frac{\alpha\overline{\eta}-\overline{\alpha}\eta}{2}
\right)
D(\alpha+\eta),
\end{equation}
we get
\begin{equation}\label{eqn:displaced-thermal-coherent-representation}
\varrho_{x,y,\beta}
=
\frac{1}{\pi\bar n_\beta}
\int_{\mathbb C}
e^{-|\eta-\alpha|^2/\bar n_\beta}
|\eta\rangle\langle\eta|\,\mathrm{d}\eta .
\end{equation}
Hence
\begin{equation}\label{eqn:ell-b-gaussian-convolution}
\ell_b(x,y,\beta)
=
\frac{1}{\pi\bar n_\beta}
\int_{\mathbb C}
e^{-|\eta-\alpha|^2/\bar n_\beta}
B_b(\eta)\,\mathrm{d}\eta,
\end{equation}
where
\begin{equation}
B_b(\eta):=\langle \eta|b|\eta\rangle.
\end{equation}
The function $B_b$ is bounded by $\|b\|$ and is continuous. Indeed, the
coherent-state map $\eta\mapsto|\eta\rangle$ is norm-continuous because
\begin{equation}
\langle \zeta|\eta\rangle
=
\exp\left(
-\frac{|\zeta|^2}{2}
-\frac{|\eta|^2}{2}
+\overline{\zeta}\eta
\right),
\end{equation}
and therefore
\begin{equation}
\||\zeta\rangle-|\eta\rangle\|^2
=
2-2\operatorname{Re}\langle \zeta|\eta\rangle
\longrightarrow 0
\qquad
\text{as }\eta\to\zeta.
\end{equation}

Identify $\mathbb C$ with $\mathbb R^2$ and write, for $s>0$,
\begin{equation}
\gamma_s(u):=\frac{1}{\pi s}e^{-|u|^2/s}.
\end{equation}
Then equation~\eqref{eqn:ell-b-gaussian-convolution} becomes
\begin{equation}
\ell_b(x,y,\beta)
=
\int_{\mathbb R^2}
\gamma_{\bar n_\beta}(\eta-\alpha)B_b(\eta)\,\mathrm{d}\eta.
\end{equation}
%%%%%%%%%%%%%%%%
Let $K\Subset\mathbb R^2$ be compact, and $[s_0,s_1]\Subset(0,\infty)$. For every
multi-index $\mu$ in the displacement variables and every $j\in\mathbb N_0$,
the derivative
\begin{equation}
\partial_\alpha^\mu\partial_s^j\gamma_s(\eta-\alpha)
\end{equation}
is a polynomial in $\eta-\alpha$ and $s^{-1}$ times
$e^{-|\eta-\alpha|^2/s}$. Hence, uniformly for
$\alpha\in K$ and $s\in[s_0,s_1]$, it is dominated by an integrable function of
$\eta$. Since $B_b\in \mathcal{L}^\infty(\mathbb R^2)$, dominated convergence allows
differentiation under the integral sign. The same local majorants also imply
continuity of all differentiated integrals. Thus
\begin{equation}
(\alpha,s)\longmapsto
\int_{\mathbb R^2}
\gamma_s(\eta-\alpha)B_b(\eta)\,\mathrm{d}\eta
\end{equation}
is smooth on $\mathbb R^2\times(0,\infty)$.
%%%%%%%%%%%
Finally, the map
\begin{equation}
\beta\longmapsto\bar n_\beta=\frac{1}{e^\beta-1}
\end{equation}
is smooth from $(0,\infty)$ to $(0,\infty)$, and
$(x,y)\mapsto \alpha=(x+iy)/\sqrt2$ is real-linear. Therefore
$(x,y,\beta)\mapsto\ell_b(x,y,\beta)$ is smooth.
\end{proof}

\begin{lemma}
\label{lem:displaced-thermal-HS-representatives}
Set
\begin{equation}
S_{x,y,\beta}:=\varrho_{x,y,\beta}^{1/2},
\qquad
\sigma_\beta:=\tau_\beta^{1/2}.
\end{equation}
Then the following operators are Hilbert--Schmidt:
\begin{equation}\label{eqn:displaced-thermal-Xx}
X_x
:=
2t_\beta W(x,y)Q\sigma_\beta W(x,y)^*
=
2t_\beta(Q-xI)S_{x,y,\beta},
\end{equation}
\begin{equation}\label{eqn:displaced-thermal-Xy}
X_y
:=
2t_\beta W(x,y)P\sigma_\beta W(x,y)^*
=
2t_\beta(P-yI)S_{x,y,\beta},
\end{equation}
and
\begin{equation}\label{eqn:displaced-thermal-Xbeta}
\begin{split} 
X_\beta
&:=
W(x,y)(\bar n_\beta I-N)\sigma_\beta W(x,y)^* \\
&=
\left(
\bar n_\beta I
-
\frac12\bigl((Q-xI)^2+(P-yI)^2-I\bigr)
\right)
S_{x,y,\beta}.
\end{split}
\end{equation}
Moreover,
\begin{equation}
X_x,\ X_y,\ X_\beta
\in
U_{\varrho_{x,y,\beta}}\bigl(V_{\rho_{x,y,\beta}}\bigr),
\end{equation}
where $V_{\rho_{x,y,\beta}}$ is as in equation \eqref{eq:HS_image_of_Vrho}.
%%%%%%%%%%%%%%%%%%%%%%%

For every $b\in\mathscr{A}_{\mathrm{sa}}=\mathcal B(\mathcal H)_{\mathrm{sa}}$, the coordinate derivatives of the
expectation-value functions satisfy
\begin{equation}\label{eqn:displaced-thermal-coordinate-compatibility}
\partial_j\ell_b(x,y,\beta)
=
\Re\operatorname{Tr}\!\left(X_j^*bS_{x,y,\beta}\right),
\qquad
j=x,y,\beta.
\end{equation}
Equivalently,
\begin{equation}\label{eqn:displaced-thermal-weak-sld-coordinate}
2\,\partial_j\varrho_{x,y,\beta}
=
S_{x,y,\beta}X_j^*+X_jS_{x,y,\beta},
\qquad
j=x,y,\beta,
\end{equation}
as trace-class identities.
\end{lemma}

\begin{proof}
We first check that the displayed operators are Hilbert--Schmidt. 
%%%%%%%%%%%%
Since
$\sigma_\beta$ is diagonal in the number basis, a direct computation shows that
\begin{equation}
\|Q\sigma_\beta\|_{HS}^2
=
\operatorname{Tr}(\tau_\beta Q^2)
<\infty,
\qquad
\|P\sigma_\beta\|_{HS}^2
=
\operatorname{Tr}(\tau_\beta P^2)
<\infty,
\end{equation}
and
\begin{equation}
\|N\sigma_\beta\|_{HS}^2
=
(1-q_\beta)\sum_{n=0}^{\infty} n^2q_\beta^n
<\infty.
\end{equation}
Thus $X_x$, $X_y$, and $X_\beta$ are Hilbert--Schmidt.
%%%%%%%%%%%%%%%%%%%%%%%%%%%

The alternative expressions in terms of $Q-xI$ and $P-yI$ follow from
\begin{equation}
S_{x,y,\beta}
=
W(x,y)\sigma_\beta W(x,y)^*
\end{equation}
and
\begin{equation}
W(x,y)QW(x,y)^*=Q-xI,
\qquad
W(x,y)PW(x,y)^*=P-yI.
\end{equation}
Similarly,
\begin{equation}
W(x,y)NW(x,y)^*
=
\frac12\bigl((Q-xI)^2+(P-yI)^2-I\bigr),
\end{equation}
on the corresponding displaced core, which gives the second expression for $X_\beta$.
%%%%%%%%%%%%%%%%%

We now prove that $X_x$, $X_y$, and $X_\beta$ belong to the distinguished real Hilbert--Schmidt subspace
$U_{\varrho_{x,y,\beta}}(V_{\rho_{x,y,\beta}})$. 
%%%%%%%%%%%%%%%%
Let $P_N$ denote the projection
onto $\operatorname{span}\{e_0,\ldots,e_N\}$. 
%%%%%%%%%%%%%%%%
Then, a direct computation shows that
\begin{equation}
P_NQP_N\sigma_\beta\longrightarrow Q\sigma_\beta,
\qquad
P_NPP_N\sigma_\beta\longrightarrow P\sigma_\beta,
\end{equation}
and
\begin{equation}
P_NNP_N\sigma_\beta\longrightarrow N\sigma_\beta
\end{equation}
in Hilbert--Schmidt norm. 
%%%%%%%%%%%%%
Since $P_NQP_N$, $P_NPP_N$, and $P_NNP_N$ are bounded
self-adjoint operators, it follows that $Q\sigma_\beta$, $P\sigma_\beta$, and $(\bar n_\beta I-N)\sigma_\beta$ belong to the real Hilbert--Schmidt closure of
\begin{equation}
\{a\tau_\beta^{1/2}:a=a^*\in\mathcal B(\mathcal H)\}.
\end{equation}
Conjugating by $W(x,y)$ gives the asserted membership for
$X_x$, $X_y$, and $X_\beta$.

It remains to verify the compatibility identities in equation \eqref{eqn:displaced-thermal-coordinate-compatibility}. 
%%%%%%%%%%%%%
Let
\begin{equation}
B:=W(x,y)^*bW(x,y).
\end{equation}
By the Weyl relations, varying $x$ or $y$ amounts, up to a scalar phase which
disappears under conjugation, to conjugating $\tau_\beta$ by
$e^{-ihP}$ or $e^{ihQ}$, respectively. Hence
\begin{equation}
\ell_b(x+h,y,\beta)
=
\operatorname{Tr}\left(
e^{-ihP}\tau_\beta e^{ihP}B
\right),
\end{equation}
and
\begin{equation}
\ell_b(x,y+h,\beta)
=
\operatorname{Tr}\left(
e^{ihQ}\tau_\beta e^{-ihQ}B
\right).
\end{equation}

To justify differentiation, let
\begin{equation}
\tau_\beta^{(N)}
:=
(1-q_\beta)\sum_{n=0}^N q_\beta^n |e_n\rangle\langle e_n|.
\end{equation}
For the finite-rank operators $\tau_\beta^{(N)}$, differentiation under
unitary conjugation gives
\begin{equation}
\frac{d}{dh}\bigg|_{h=0}
e^{-ihP}\tau_\beta^{(N)}e^{ihP}
=
-i[P,\tau_\beta^{(N)}],
\end{equation}
and
\begin{equation}
\frac{d}{dh}\bigg|_{h=0}
e^{ihQ}\tau_\beta^{(N)}e^{-ihQ}
=
i[Q,\tau_\beta^{(N)}].
\end{equation}
Moreover, $\tau_\beta^{(N)}\to\tau_\beta$ in trace norm and the commutators
converge in trace norm:
\begin{equation}
-i[P,\tau_\beta^{(N)}]\longrightarrow -i[P,\tau_\beta],
\qquad
i[Q,\tau_\beta^{(N)}]\longrightarrow i[Q,\tau_\beta].
\end{equation}
Indeed, these commutators have only first off-diagonal matrix coefficients,
with weights bounded by a constant times $\sqrt{n+1}q_\beta^n$, which are
summable. Therefore
\begin{equation}\label{eqn:dx-expectation-commutator}
\partial_x\ell_b(x,y,\beta)
=
\operatorname{Tr}\left(
W(x,y)\bigl(-i[P,\tau_\beta]\bigr)W(x,y)^*b
\right),
\end{equation}
and
\begin{equation}\label{eqn:dy-expectation-commutator}
\partial_y\ell_b(x,y,\beta)
=
\operatorname{Tr}\left(
W(x,y)i[Q,\tau_\beta]W(x,y)^*b
\right).
\end{equation}

For the $\beta$-direction, the diagonal expansion gives trace-norm
differentiability directly:
\begin{equation}
\partial_\beta\tau_\beta
=
(\bar n_\beta I-N)\tau_\beta.
\end{equation}
Hence
\begin{equation}\label{eqn:dbeta-expectation-commutator}
\partial_\beta\ell_b(x,y,\beta)
=
\operatorname{Tr}\left(
W(x,y)(\partial_\beta\tau_\beta)W(x,y)^*b
\right).
\end{equation}

We now rewrite the trace-norm derivatives in symmetric form. On
$\mathcal D_{\mathrm{fin}}$ one has
\begin{equation}
a\tau_\beta=q_\beta\tau_\beta a,
\qquad
\tau_\beta a^\dagger=q_\beta a^\dagger\tau_\beta.
\end{equation}
Using
\begin{equation}
Q=\frac{a+a^\dagger}{\sqrt2},
\qquad
P=\frac{a-a^\dagger}{i\sqrt2},
\end{equation}
one obtains
\begin{equation}\label{eqn:x-commutator-symmetric-form}
-i[P,\tau_\beta]
=
t_\beta(Q\tau_\beta+\tau_\beta Q),
\end{equation}
and
\begin{equation}\label{eqn:y-commutator-symmetric-form}
i[Q,\tau_\beta]
=
t_\beta(P\tau_\beta+\tau_\beta P).
\end{equation}
Finally,
\begin{equation}\label{eqn:beta-symmetric-form}
\partial_\beta\tau_\beta
=
(\bar n_\beta I-N)\tau_\beta
=
\frac12\left(
(\bar n_\beta I-N)\tau_\beta+\tau_\beta(\bar n_\beta I-N)
\right),
\end{equation}
because $\bar n_\beta I-N$ commutes with $\tau_\beta$.

Using equations~\eqref{eqn:displaced-thermal-Xx}--\eqref{eqn:displaced-thermal-Xbeta}
and $S_{x,y,\beta}=W(x,y)\sigma_\beta W(x,y)^*$, we compute
\begin{equation}
S_{x,y,\beta}X_x^*+X_xS_{x,y,\beta}
=
2t_\beta W(x,y)(\tau_\beta Q+Q\tau_\beta)W(x,y)^*.
\end{equation}
By equation~\eqref{eqn:x-commutator-symmetric-form}, this gives
\begin{equation}
S_{x,y,\beta}X_x^*+X_xS_{x,y,\beta}
=
2W(x,y)\bigl(-i[P,\tau_\beta]\bigr)W(x,y)^*.
\end{equation}
Together with equation~\eqref{eqn:dx-expectation-commutator}, this proves the
$x$-identity in equation~\eqref{eqn:displaced-thermal-weak-sld-coordinate}.

The $y$-direction is identical:
\begin{equation}
S_{x,y,\beta}X_y^*+X_yS_{x,y,\beta}
=
2t_\beta W(x,y)(\tau_\beta P+P\tau_\beta)W(x,y)^*
=
2W(x,y)i[Q,\tau_\beta]W(x,y)^*,
\end{equation}
which proves the $y$-identity by equation~\eqref{eqn:dy-expectation-commutator}.

Finally,
\begin{equation}
S_{x,y,\beta}X_\beta^*+X_\beta S_{x,y,\beta}
=
2W(x,y)(\bar n_\beta I-N)\tau_\beta W(x,y)^*
=
2W(x,y)(\partial_\beta\tau_\beta)W(x,y)^*,
\end{equation}
which proves the $\beta$-identity by equation~\eqref{eqn:dbeta-expectation-commutator}.

Taking the trace pairing with $b=b^*$ and using cyclicity gives
\begin{equation}
\partial_j\ell_b(x,y,\beta)
=
\Re\operatorname{Tr}\!\left(X_j^*bS_{x,y,\beta}\right),
\qquad
j=x,y,\beta,
\end{equation}
which is equation~\eqref{eqn:displaced-thermal-coordinate-compatibility}.
\end{proof}

\begin{proposition}\label{prop:displaced-thermal-model}
The displaced thermal model $(M,\mathrm i,\mathscr A)$ defined above is a
Hermitian regular GNS-smooth normal parametric model in the sense of
Definitions~\ref{def:gns-smooth-parametric-model} and~\ref{def:regular model}.

For
\begin{equation}
v_m=u\,\partial_x+w\,\partial_y+\dot\beta\,\partial_\beta
\in T_mM,
\end{equation}
the Hilbert--Schmidt representative of the canonical real GNS lift is
\begin{equation}\label{eqn:displaced-thermal-Xv}
X_{v_m}
=
uX_x+wX_y+\dot\beta X_\beta.
\end{equation}
with $X_{x}$, $X_{y}$, and $X_{\beta}$ as in equations \eqref{eqn:displaced-thermal-Xx}, \eqref{eqn:displaced-thermal-Xy}, and \eqref{eqn:displaced-thermal-Xbeta}, respectively.
%%%%%%%%%%%%%%%%%%%%%%%%%%%
If
\begin{equation}
z_m=z_x\partial_x+z_y\partial_y+z_\beta\partial_\beta
\in T_m^{\mathbb C}M,
\end{equation}
then the Hilbert--Schmidt image of the complexified Riesz-vector lift is
\begin{equation}\label{eqn:displaced-thermal-Xhat-z}
\widehat X_{z_m}
=
\overline{z_x}X_x+\overline{z_y}X_y+\overline{z_\beta}X_\beta.
\end{equation}
The induced Hermitian form in equation \eqref{eqn:induced-hermitian-form-on-complex-vectors} reads
\begin{equation}\label{eqn:displaced-thermal-K}
\begin{split}
K_m(z_m,z'_m)
&=
\operatorname{Tr}\!\left(\widehat X_{z'_m}^{\,*}\widehat X_{z_m}\right)
\\
&=
2t_\beta\left(\overline{z_x}z'_x+\overline{z_y}z'_y\right)
+
\frac{\overline{z_\beta}z'_\beta}{4\sinh^2(\beta/2)}
\\
&\quad
-
2i\,t_\beta^2
\left(\overline{z_x}z'_y-\overline{z_y}z'_x\right).
\end{split}
\end{equation}
Consequently, on real tangent vectors, the induced bilinear forms in equation \eqref{eq:induced-bilinear-forms} read
\begin{equation}\label{eqn:displaced-thermal-G}
G
=
2\tanh\!\left(\frac\beta2\right)(dx^2+dy^2)
+
\frac{1}{4\sinh^2(\beta/2)}\,d\beta^2,
\end{equation}
and
\begin{equation}\label{eqn:displaced-thermal-Omega}
\Omega
=
-2\tanh^2\!\left(\frac\beta2\right)\,dx\wedge dy.
\end{equation}
In particular,
\begin{equation}\label{eqn:displaced-thermal-dOmega}
d\Omega
=
-2\tanh\!\left(\frac\beta2\right)
\operatorname{sech}^2\!\left(\frac\beta2\right)
\,d\beta\wedge dx\wedge dy,
\end{equation}
which is nonzero for every $\beta>0$.
\end{proposition}

\begin{proof}
Normality and faithfulness have already been observed above. Lemma~\ref{lem:displaced-thermal-smooth-expectations}
proves condition~\ref{cond:smoothness-of-expectations} in
Definition~\ref{def:gns-smooth-parametric-model}.
%%%%%%%%%%%%%%%%%%

Concerning condition \ref{cond:GNS-boundedness} in Definition \ref{def:gns-smooth-parametric-model}, let
\begin{equation}
v_m=u\,\partial_x+w\,\partial_y+\dot\beta\,\partial_\beta.
\end{equation}
%%%%%%%%%%%%%%%%%%%%%%%%%
Define $X_{v_m}$ by equation~\eqref{eqn:displaced-thermal-Xv}. By Lemma~\ref{lem:displaced-thermal-HS-representatives},
$X_{v_m}$ belongs to $U_{\varrho_{x,y,\beta}}(V_{\rho_{x,y,\beta}})$ and satisfies
\begin{equation}
\langle v_m,d\ell_b(m)\rangle
=
\Re\operatorname{Tr}\!\left(
X_{v_m}^*bS_{x,y,\beta}
\right)
\end{equation}
for all $b=b^*\in\mathcal B(\mathcal H)$. 
%%%%%%%%%%%%
Hence, by Cauchy--Schwarz, it follows that
\begin{align}
|\langle v_m,d\ell_b(m)\rangle|
&\leq
\|X_{v_m}\|_{HS}\,
\|bS_{x,y,\beta}\|_{HS}
\nonumber\\
&=
\|X_{v_m}\|_{HS}
\sqrt{\operatorname{Tr}(\varrho_{x,y,\beta}b^2)},
\end{align}
which proves condition~\ref{cond:GNS-boundedness} in Definition \ref{def:gns-smooth-parametric-model}. 
%%%%%%%%%%%%%%%%%%
Moreover, by
Proposition~\ref{prop:boundedness-canonical-representative}, $X_{v_m}$ is the Hilbert--Schmidt representative of the canonical real GNS lift, and equation \eqref{eqn:displaced-thermal-Xhat-z} follows from equations \eqref{eqn:complex-gns-vector-lift} and \eqref{eq:def_X_vm_pre}.
%%%%%%%%%%%%%%%%%%%%%%%%%%%%%

We now compute the induced Hermitian tensor and then verify condition \ref{cond:injectivity} in Definition \ref{def:gns-smooth-parametric-model}.
%%%%%%%%%%%%%%
By unitary invariance of the Hilbert--Schmidt inner product, it is enough to compute $K$ at $(x,y,\beta)=(0,0,\beta)$. 
%%%%%%%%%%%%%%%%%
At the reference point, it is 
\begin{equation}
X_{v}^{(0)}
=
\left(
2t_\beta(uQ+wP)+\dot\beta(\bar n_\beta I-N)
\right)\tau_\beta^{1/2}.
\end{equation}
The required thermal moments are
\begin{equation}
\operatorname{Tr}(\tau_\beta Q^2)
=
\operatorname{Tr}(\tau_\beta P^2)
=
\bar n_\beta+\frac12
=
\frac12\coth\!\left(\frac\beta2\right),
\end{equation}
\begin{equation}
\operatorname{Tr}\!\left(\tau_\beta(\bar n_\beta I-N)^2\right)
=
\bar n_\beta(\bar n_\beta+1)
=
\frac{1}{4\sinh^2(\beta/2)},
\end{equation}
and
\begin{equation}
\operatorname{Tr}(\tau_\beta QP)=\frac{i}{2},
\qquad
\operatorname{Tr}(\tau_\beta PQ)=-\frac{i}{2}.
\end{equation}
All mixed terms involving one of $Q,P$ and one factor of $\bar n_\beta I-N$
vanish by parity, or equivalently by diagonality of $\tau_\beta$ and $N$ in the
number basis.

For
\begin{equation}
v=u\,\partial_x+w\,\partial_y+\dot\beta\,\partial_\beta,
\qquad
v'=u'\,\partial_x+w'\,\partial_y+\dot\beta'\,\partial_\beta,
\end{equation}
we therefore obtain
\begin{align}
K_m(v,v')
&=
\operatorname{Tr}\!\left((X_{v'}^{(0)})^*X_v^{(0)}\right)
\nonumber\\
&=
2t_\beta(uu'+ww')
+
\frac{\dot\beta\,\dot\beta'}{4\sinh^2(\beta/2)}
-
2i\,t_\beta^2(uw'-wu').
\end{align}
By sesquilinearity and equation \eqref{eqn:displaced-thermal-Xhat-z}, this is precisely equation~\eqref{eqn:displaced-thermal-K}.
%%%%%%%%%%%%%%%%%%%%%%%%%%%%%%%%

Taking $v=v'$ gives
\begin{equation}
\|X_v\|_{HS}^2
=
2t_\beta(u^2+w^2)
+
\frac{\dot\beta^2}{4\sinh^2(\beta/2)}.
\end{equation}
Suppose now that
\begin{equation}
\langle v_m,d\ell_b(m)\rangle=0,
\qquad
\forall b=b^*\in\mathcal B(\mathcal H).
\end{equation}
Since $X_{v_m}$ belongs to the distinguished real GNS subspace and satisfies
the compatibility identity, uniqueness in
Proposition~\ref{prop:boundedness-canonical-representative} gives
$X_{v_m}=0$. Hence the preceding norm identity implies
\begin{equation}
u=w=\dot\beta=0,
\end{equation}
because $t_\beta>0$ and $\sinh(\beta/2)>0$ for $\beta>0$. 
%%%%%%%%%%
Thus $v_m=0$, and condition~\ref{cond:injectivity} in
Definition~\ref{def:gns-smooth-parametric-model} holds, showing that the model is GNS-smooth as in Definition \ref{def:gns-smooth-parametric-model}.
%%%%%%%%%%%%%%%%

To show that the model is Hermitian regular as in Definition \ref{def:regular model}, we first note that the real and imaginary parts of equation~\eqref{eqn:displaced-thermal-K} on
real tangent vectors give equations~\eqref{eqn:displaced-thermal-G} and \eqref{eqn:displaced-thermal-Omega}. 
%%%%%%%%%
Since the coefficients of $G$  and $\Omega$ are smooth functions of $\beta$, and the parameter
manifold is finite-dimensional, the model is symmetric, skew-symmetric, and Hermitian regular in the sense of Definition~\ref{def:regular model}.
%%%%%%%%%%%%%%%%%%%%%%%%%

Finally, differentiating
\begin{equation}
\Omega
=
-2\tanh^2\!\left(\frac\beta2\right)\,dx\wedge dy
\end{equation}
gives
\begin{equation}
d\Omega
=
-2\tanh\!\left(\frac\beta2\right)
\operatorname{sech}^2\!\left(\frac\beta2\right)
\,d\beta\wedge dx\wedge dy.
\end{equation}
This is nonzero for $\beta>0$.
\end{proof}

\begin{corollary}\label{cor:fixed-temperature-displaced-thermal-model}
Fix $\beta_0>0$ and consider the smooth injective immersion
\begin{equation}
k_{\beta_0}:\mathbb R^2\longrightarrow M,
\qquad
k_{\beta_0}(x,y):=(x,y,\beta_0).
\end{equation}
The restricted model
$(\mathbb R^2,\mathrm i\circ k_{\beta_0},\mathscr A)$ is a Hermitian regular
GNS-smooth normal parametric model. Its induced tensors are
\begin{equation}
k_{\beta_0}^*G
=
2\tanh\!\left(\frac{\beta_0}{2}\right)(dx^2+dy^2),
\end{equation}
and
\begin{equation}
k_{\beta_0}^*\Omega
=
-2\tanh^2\!\left(\frac{\beta_0}{2}\right)\,dx\wedge dy.
\end{equation}
Equivalently, writing $q:=e^{-\beta_0}\in(0,1)$, the fixed-temperature model is
\begin{equation}
\varrho_{x,y}
=
W(x,y)\varrho_0W(x,y)^*,
\qquad
\varrho_0
=
(1-q)\sum_{n=0}^\infty q^n|e_n\rangle\langle e_n|,
\end{equation}
and
\begin{equation}
G
=
2\,\frac{1-q}{1+q}(dx^2+dy^2),
\end{equation}
whereas
\begin{equation}
\Omega
=
-2\left(\frac{1-q}{1+q}\right)^2dx\wedge dy.
\end{equation}
In particular, $\Omega$ is a nonzero constant multiple of the standard
symplectic form on $\mathbb R^2$, and hence
\begin{equation}
d\Omega=0
\end{equation}
on every fixed-temperature submodel.
\end{corollary}

\begin{proof}
This follows from Proposition~\ref{prop:submanifold} applied to the immersion
$k_{\beta_0}$ and from the formulas in Proposition~\ref{prop:displaced-thermal-model}.
\end{proof}

%%%%%%%%%%%%%%%%%%%%%

\section{The real dual GNS bundle and the closedness problem for \texorpdfstring{$\Omega$}{Omega}}
\label{sec:closed-extensions}

The examples of the preceding sections show that the two-form $\Omega$ induced by the imaginary part of the GNS Hermitian tensor has no uniform closedness property.
It vanishes in the commutative case, it is closed for pure states because it is, up to normalization and sign/order convention, the Fubini--Study symplectic form, but it is not closed for faithful qubits and for the full displaced thermal model.
The purpose of this section is to explain this behaviour from a structural point of view.

The main point is that there are two related but distinct geometric questions.
The first one is a total-space question: after imposing enough regularity so that the real dual of the realified pullback GNS fibration becomes a smooth real Hilbert bundle over $M$, the fiberwise symplectic forms induced by the dual GNS Hermitian structure define a vertical two-form, and one may ask whether this vertical form admits a closed extension to the total space.
The second one is a base-space question: whether the two-form $\Omega$ induced on the parameter manifold is closed.
These two questions are not equivalent.
Indeed, $\Omega$ is not obtained as the pullback of a closed two-form on the total space along a section of the bundle.
Rather, it is obtained by pairing the canonical real dual GNS lift $L^{\mathbb R}$ with itself through the fiberwise symplectic form.
Consequently, the closedness of $\Omega$ is controlled by the covariant exterior derivative of $L^{\mathbb R}$, not merely by the existence of a closed extension on the total space.

Throughout this section, differential forms on smooth Hilbert manifolds are
understood as smooth sections of the bundles of continuous alternating
multilinear forms. Thus, for a smooth Hilbert manifold $N$, we write
\begin{equation}
\Omega^k(N)
=
\Gamma\!\left(
\operatorname{Alt}^k_{\mathrm{cont}}(TN;\mathbb R)
\right).
\end{equation}
All exterior derivatives, pullbacks, and contractions below are understood
in this category.  When we use de Rham cohomological language for Hilbert
manifolds, it is only in settings where the standard smooth homotopy
invariance of the chosen de Rham complex is available.

%%%%%%%%%%%%%%

\subsection{Bundle-regular models and the vertical symplectic form}

In full generality, the GNS fibration $\mathcal H(\mathscr A)\to\mathcal S(\mathscr A)$ is not locally trivial.
However, for the main quantum examples considered above, local triviality becomes available after restriction to the relevant class of states.
For faithful normal states on $\mathscr A=\mathcal B(\mathcal H)$, the Hilbert--Schmidt realization identifies the GNS fibers with the fixed Hilbert space $\mathcal B_2(\mathcal H)$.
For pure states, the GNS fiber over $\rho_\psi$ is identified with $\mathcal H$, and the change of representative $\psi\mapsto e^{i\theta}\psi$ acts by the standard scalar unitary action.
These examples motivate the following regularity assumption.

\begin{definition}
\label{def:bundle-regular-model}
Let $(M,\mathrm{i},\mathscr A)$ be a Hermitian regular GNS-smooth model as in Definitions~\ref{def:gns-smooth-parametric-model} and~\ref{def:regular model}.
We call it \emph{bundle-regular} if the following additional conditions hold.
\begin{enumerate}
\item The real dual of the realified pullback GNS fibration admits the structure of a smooth real Hilbert bundle
\begin{equation}\label{eqn:bundle-regular-real-dual-bundle}
\pi_E:E\longrightarrow M,
\qquad
E_m:=(\mathcal H_{\mathrm i(m)}^{\mathbb R})^*,
\end{equation}
whose fiberwise Hilbert product is transported from the dual GNS Hermitian product, namely
\begin{equation}\label{eqn:transported-real-metric-on-E}
g_m(\alpha,\beta)
:=
\Re\left\langle
\mathfrak C_{\mathrm i(m)}(\alpha),
\mathfrak C_{\mathrm i(m)}(\beta)
\right\rangle_{\mathrm i(m)}^* .
\end{equation}

\item Let
\begin{equation}
J_{\mathrm i(m)}:\mathcal H_{\mathrm i(m)}^{\mathbb R}
\longrightarrow
\mathcal H_{\mathrm i(m)}^{\mathbb R}
\end{equation}
denote the complex structure on the realified GNS Hilbert space, given by multiplication by $i$.
The fiberwise endomorphisms
\begin{equation}\label{eqn:dual-real-complex-structure}
I_m:E_m\longrightarrow E_m,
\qquad
I_m\alpha:=\alpha\circ J_{\mathrm i(m)},
\end{equation}
define a smooth bundle endomorphism $I:E\to E$ satisfying $I^2=-\mathrm{id}_E$.
Moreover, the definition of $I_m$ is compatible with $\mathfrak C_{\mathrm i(m)}$ in the sense that
\begin{equation}
\mathfrak C_{\mathrm i(m)}(I_m\alpha)
=
i\,\mathfrak C_{\mathrm i(m)}(\alpha).
\end{equation}
Consequently $I_m$ is $g_m$-orthogonal and
\begin{equation}
g_m(I_m\alpha,\beta)=-g_m(\alpha,I_m\beta).
\end{equation}

\item The canonical real dual GNS lift defined in \eqref{eqn:real-gns-lift},
\begin{equation}
L_m^{\mathbb R}:T_mM\longrightarrow E_m,
\end{equation}
is smooth as an $E$-valued one-form
\begin{equation}
L^{\mathbb R}\in\Omega^1(M;E).
\end{equation}

\item The bundle $E\to M$ admits a smooth connection $\nabla$ preserving both $g$ and $I$, namely
\begin{equation}
\nabla g=0,
\qquad
\nabla I=0.
\end{equation}
\end{enumerate}
\end{definition}

\medskip

For a bundle-regular model, define the fiberwise two-form
\begin{equation}\label{eqn:bundle-regular-fiberwise-omega}
\omega_m(\alpha,\beta)
:=
g_m(I_m\alpha,\beta),
\qquad
\alpha,\beta\in E_m.
\end{equation}
Equivalently,
\begin{equation}
\omega_m(\alpha,\beta)
=
\Im\left\langle
\mathfrak C_{\mathrm i(m)}(\alpha),
\mathfrak C_{\mathrm i(m)}(\beta)
\right\rangle_{\mathrm i(m)}^* .
\end{equation}
Since $I_m^2=-\mathrm{id}$ and $g_m$ is a Hilbert product, $\omega_m$ is a strong symplectic form on the real Hilbert space $E_m$.
The vertical tangent bundle $T_{\mathrm{vert}}E:=\ker(T\pi_E)$ is canonically identified with $\pi_E^*E$.
Hence the family $(\omega_m)_{m\in M}$ defines a smooth vertical two-form
\begin{equation}\label{eqn:bundle-regular-vertical-omega}
\omega_{\mathrm{vert}}
\in
\Gamma\!\left(
\operatorname{Alt}^2_{\mathrm{cont}}(T_{\mathrm{vert}}E;\mathbb R)
\right)
\end{equation}
by
\begin{equation}\label{eqn:bundle-regular-vertical-omega-definition}
\omega_{\mathrm{vert},\alpha}(V,W)
=
\omega_{\pi_E(\alpha)}(V,W),
\qquad
V,W\in T_\alpha^{\mathrm{vert}}E\cong E_{\pi_E(\alpha)}.
\end{equation}
This is only a vertical two-form: it is defined on pairs of vertical tangent vectors, not on arbitrary tangent vectors to $E$.

A closed extension of $\omega_{\mathrm{vert}}$ is a smooth differential two-form
\begin{equation}
\widetilde\Omega\in\Omega^2(E)
=
\Gamma\!\left(
\operatorname{Alt}^2_{\mathrm{cont}}(TE;\mathbb R)
\right)
\end{equation}
such that, for every fiber inclusion $\iota_m:E_m\hookrightarrow E$,
\begin{equation}\label{eqn:closed-extension-restriction}
\iota_m^*\widetilde\Omega=\omega_m,
\end{equation}
and
\begin{equation}\label{eqn:closed-extension-closed}
d_E\widetilde\Omega=0.
\end{equation}

With this notation, the two-form induced on the parameter manifold is
\begin{equation}\label{eqn:bundle-regular-Omega-from-LR}
\Omega_m(v,w)
=
\omega_m\bigl(L_m^{\mathbb R}(v),L_m^{\mathbb R}(w)\bigr).
\end{equation}
Indeed, by \eqref{eqn:complex-gns-lift-on-real-vectors}, the complex dual lift is given by $L_m(v)=\mathfrak C_{\mathrm i(m)}(L_m^{\mathbb R}(v))$, and therefore
\begin{equation}
\omega_m\bigl(L_m^{\mathbb R}(v),L_m^{\mathbb R}(w)\bigr)
=
\Im\left\langle L_m(v),L_m(w)\right\rangle_{\mathrm i(m)}^*
=
\Omega_m(v,w).
\end{equation}
Thus $\Omega$ is obtained by pairing the $E$-valued one-form $L^{\mathbb R}$ with itself through the transported fiberwise symplectic form.
It is not the pullback of a two-form on $E$ along a section of $E$.

\subsection{Closed extensions on the total space}

\begin{theorem}
\label{thm:closed-extension-bundle-regular}
Let $(M,\mathrm{i},\mathscr A)$ be a bundle-regular model as in Definition~\ref{def:bundle-regular-model}, and let $\omega_{\mathrm{vert}}$ be the smooth vertical two-form defined in \eqref{eqn:bundle-regular-vertical-omega}.
For every smooth connection $\nabla$ on $E$ preserving $g$ and $I$, the vertical two-form $\omega_{\mathrm{vert}}$ admits a closed extension to the total space $E$.
More precisely, let
\begin{equation}
\operatorname{ver}_\nabla:TE\longrightarrow \pi_E^*E
\end{equation}
denote the vertical projection associated with $\nabla$.
Define $\lambda_\nabla\in\Omega^1(E)$ by
\begin{equation}\label{eqn:bundle-regular-liouville-form}
(\lambda_\nabla)_\alpha(U)
=
\frac12\,
\omega_{\pi_E(\alpha)}
\bigl(\alpha,\operatorname{ver}_\nabla(U)\bigr),
\qquad
\alpha\in E,
\quad
U\in T_\alpha E.
\end{equation}
Then
\begin{equation}\label{eqn:bundle-regular-closed-extension}
\widetilde\Omega_\nabla:=d_E\lambda_\nabla
\end{equation}
is a closed differential two-form on $E$ satisfying
\begin{equation}\label{eqn:bundle-regular-extension-restriction}
\iota_m^*\widetilde\Omega_\nabla=\omega_m,
\qquad
\forall m\in M.
\end{equation}
\end{theorem}

\begin{proof}
Since $\nabla$ preserves $g$ and $I$, it preserves $\omega$.
Hence $\lambda_\nabla$ is a smooth one-form on the smooth Hilbert manifold $E$.
Therefore $\widetilde\Omega_\nabla=d_E\lambda_\nabla$ is closed by construction.
It remains to check the restriction to the fibers.
Fix $m\in M$.
Along the fiber $E_m$, the vertical projection $\operatorname{ver}_\nabla$ is the identity on tangent vectors to $E_m$.
Therefore
\begin{equation}
(\iota_m^*\lambda_\nabla)_\alpha(V)
=
\frac12\,\omega_m(\alpha,V),
\qquad
\alpha\in E_m,
\quad
V\in T_\alpha E_m\cong E_m.
\end{equation}
This is the standard Liouville one-form associated with the constant symplectic form $\omega_m$ on the vector space $E_m$.
Consequently,
\begin{equation}
d(\iota_m^*\lambda_\nabla)=\omega_m.
\end{equation}
Since exterior differentiation commutes with pullback, we get
\begin{equation}
\iota_m^*\widetilde\Omega_\nabla
=
\iota_m^*(d_E\lambda_\nabla)
=
d(\iota_m^*\lambda_\nabla)
=
\omega_m.
\end{equation}
\end{proof}

\begin{remark}
The closed-extension problem considered above is analogous to the problem studied by Gotay--Lashof--Sniatycki--Weinstein~\cite{GLSW1983} for finite-dimensional symplectic fibrations.
In that setting, a fiberwise closed form admits a closed extension to the total space precisely when its fiberwise cohomology class extends to a de Rham cohomology class on the total space.
In the present situation, the fibers are real Hilbert spaces and hence contractible.
Thus the cohomological obstruction represented by the fiberwise class of $\omega_m$ is absent, and Theorem~\ref{thm:closed-extension-bundle-regular} constructs a closed extension directly from a compatible connection by setting $\widetilde\Omega_\nabla=d_E\lambda_\nabla$.
\end{remark}

\begin{remark}\label{rem:curvature-horizontal-block}
The form $\widetilde\Omega_\nabla=d_E\lambda_\nabla$ is the connection-dependent exact extension analogous to the minimal-coupling form in the finite-dimensional theory.
More explicitly, let
\begin{equation}
T_\alpha E
=
T_\alpha^{\mathrm{hor},\nabla}E
\oplus
T_\alpha^{\mathrm{vert}}E
\end{equation}
be the horizontal--vertical splitting induced by $\nabla$.
If $\alpha\in E_m$, if $V,W\in T_\alpha^{\mathrm{vert}}E\cong E_m$, and if $H_u,H_v\in T_\alpha^{\mathrm{hor},\nabla}E$ are the horizontal lifts of $u,v\in T_mM$, then
\begin{equation}
\widetilde\Omega_\nabla(V,W)=\omega_m(V,W),
\qquad
\widetilde\Omega_\nabla(H_u,V)=0.
\end{equation}
With the curvature convention
\begin{equation}
F_\nabla(u,v)
=
[\nabla_u,\nabla_v]-\nabla_{[u,v]},
\end{equation}
so that
\begin{equation}
\operatorname{ver}_\nabla([H_u,H_v])
=
-F_\nabla(u,v)\alpha,
\end{equation}
one also has
\begin{equation}\label{eqn:minimal-coupling-horizontal-block}
\widetilde\Omega_\nabla(H_u,H_v)
=
\frac12\,
\omega_m\bigl(\alpha,F_\nabla(u,v)\alpha\bigr).
\end{equation}
Thus the curvature enters the connection-dependent extension only through its horizontal--horizontal block.
\end{remark}

\begin{remark}\label{rem:non-uniqueness-closed-extensions}
Closed extensions of $\omega_{\mathrm{vert}}$ are not unique.
Different choices of connection generally give different forms
$\widetilde\Omega_\nabla=d_E\lambda_\nabla$.
More generally, if $\widetilde\Omega$ and $\widetilde\Omega'$ are two closed
extensions of $\omega_{\mathrm{vert}}$, then their difference
\[
\Delta:=\widetilde\Omega'-\widetilde\Omega
\]
is a closed two-form on $E$ whose restriction to each fiber vanishes.
Equivalently, $\Delta$ vanishes on vertical--vertical pairs.  This condition
does not imply that $\Delta$ is basic; in general one cannot conclude that
$\Delta=\pi_E^*\beta$ for some closed two-form $\beta$ on $M$.

At the cohomological level, and only in settings where the relevant de Rham
homotopy invariance theorem applies to the Hilbert bundle $E\to M$, the
deformation retraction of $E$ onto the zero section implies that
$[\Delta]$ comes from the base.  Thus one may write, after choosing a closed
representative $\beta$ on $M$,
\[
\Delta=\pi_E^*\beta+d_E\gamma
\]
for some one-form $\gamma$ on $E$.  The exact term $d_E\gamma$ may have
vertical--horizontal components, and this is precisely why cohomological
pullback does not imply that the differential form itself is basic.
\end{remark}

%%%%%%%%%%%%%

\subsection{Closedness of the induced two-form on the parameter manifold}

The closed extension constructed above lives on the total space $E$.
It does not imply that the base two-form $\Omega$ is closed.
The latter is controlled by the covariant variation of the $E$-valued one-form $L^{\mathbb R}$.

\begin{proposition}[Closedness of the induced base two-form]
\label{prop:closedness-base-form-covariant}
Let $(M,\mathrm{i},\mathscr A)$ be bundle-regular, and let $\nabla$ be a smooth connection on $E$ preserving $g$ and $I$.
Then, for vector fields $X,Y,Z$ on $M$,
\begin{equation}\label{eqn:dOmega-covariant-LR}
d\Omega(X,Y,Z)
=
\sum_{\mathrm{cycl}}
\omega\bigl((d^\nabla L^{\mathbb R})(X,Y),L^{\mathbb R}(Z)\bigr),
\end{equation}
where the cyclic sum is taken over $(X,Y,Z)$, $(Y,Z,X)$, and $(Z,X,Y)$, and
\begin{equation}
(d^\nabla L^{\mathbb R})(X,Y)
=
\nabla_X(L^{\mathbb R}(Y))
-
\nabla_Y(L^{\mathbb R}(X))
-
L^{\mathbb R}([X,Y]).
\end{equation}
In particular, if $d^\nabla L^{\mathbb R}=0$, then $\Omega$ is closed.
\end{proposition}

\begin{proof}
Since $\nabla g=0$ and $\nabla I=0$, the connection preserves $\omega(\alpha,\beta)=g(I\alpha,\beta)$.
Therefore, for vector fields $X,Y,Z$ on $M$, differentiating
\begin{equation}
\Omega(Y,Z)=\omega(L^{\mathbb R}(Y),L^{\mathbb R}(Z))
\end{equation}
gives
\begin{equation}
X\bigl(\Omega(Y,Z)\bigr)
=
\omega(\nabla_XL^{\mathbb R}(Y),L^{\mathbb R}(Z))
+
\omega(L^{\mathbb R}(Y),\nabla_XL^{\mathbb R}(Z)).
\end{equation}
Using the skew-symmetry of $\omega$, the second term may be rewritten as
\begin{equation}
\omega(L^{\mathbb R}(Y),\nabla_XL^{\mathbb R}(Z))
=
-\omega(\nabla_XL^{\mathbb R}(Z),L^{\mathbb R}(Y)).
\end{equation}
Substituting this expression into the standard formula for the exterior derivative of a two-form,
\begin{align}
d\Omega(X,Y,Z)
&=
X\Omega(Y,Z)-Y\Omega(X,Z)+Z\Omega(X,Y)
\nonumber\\
&\quad
-\Omega([X,Y],Z)+\Omega([X,Z],Y)-\Omega([Y,Z],X),
\end{align}
and regrouping terms gives
\begin{equation}
d\Omega(X,Y,Z)
=
\sum_{\mathrm{cycl}}
\omega\bigl(
\nabla_XL^{\mathbb R}(Y)
-
\nabla_YL^{\mathbb R}(X)
-
L^{\mathbb R}([X,Y]),
L^{\mathbb R}(Z)
\bigr),
\end{equation}
which is precisely \eqref{eqn:dOmega-covariant-LR}.
The final assertion follows immediately.
\end{proof}

\begin{remark}
Equation~\eqref{eqn:dOmega-covariant-LR} separates the two roles played by the connection.
The curvature of $\nabla$ appears in the horizontal--horizontal block of the closed extension $\widetilde\Omega_\nabla$, as in \eqref{eqn:minimal-coupling-horizontal-block}.
By contrast, the obstruction to the closedness of the induced two-form on the parameter manifold is the covariant exterior derivative $d^\nabla L^{\mathbb R}$.
Thus the non-closedness observed in faithful qubits and in the displaced thermal model is not a contradiction with the existence of closed extensions on $E$; it reflects the way the canonical real dual GNS lift varies along the model.
\end{remark}

\subsection{Examples and further remarks}

\begin{proposition}[Bundle-regularity of the main quantum examples]
\label{prop:bundle-regularity-main-quantum-examples}
The following Hermitian regular GNS-smooth models are bundle-regular in the sense of Definition~\ref{def:bundle-regular-model}:
\begin{enumerate}
\item the pure-state model on $\mathbb{CP}(\mathcal H)$;
\item the finite-dimensional faithful-state model on $\mathscr S_f(\mathcal H)$;
\item the displaced thermal model on $\mathbb R^2\times(0,\infty)$;
\item every fixed-temperature displaced thermal submodel.
\end{enumerate}
\end{proposition}

\begin{proof}
For pure states, the GNS fiber over $\rho_\psi$ is identified with $\mathcal H$ by the unitary map used in the pure-state computation.
Under the phase change $\psi\mapsto e^{i\theta}\psi$, this identification changes by the standard scalar action of $U(1)$ on $\mathcal H$.
Hence the pullback GNS bundle over $\mathbb{CP}(\mathcal H)$ is the Hilbert bundle associated with the Hopf principal $U(1)$-bundle and the standard unitary representation of $U(1)$ on $\mathcal H$.
Its real dual is therefore a smooth real Hilbert bundle, and the transported metric and complex structure are smooth.
The Berry connection induces a smooth connection preserving these structures.
Since the real dual GNS lift is obtained from the smooth pure-state representative through the fiberwise real Riesz map, $L^{\mathbb R}$ is a smooth $E$-valued one-form.

For the finite-dimensional faithful-state model, the Hilbert--Schmidt realization identifies all GNS fibers with a fixed finite-dimensional Hilbert space.
Under this identification, the real dual bundle $E$ is a smooth finite-rank real Hilbert bundle.
The transported metric and complex structure are smooth, and a compatible connection exists.
The canonical lift is smooth because, in finite dimensions, the maps
\begin{equation}
\varrho\longmapsto \varrho^{1/2},
\qquad
\varrho\longmapsto \mathcal J_\varrho^{-1}
\end{equation}
are smooth on the open cone of faithful density matrices.
Thus $L^{\mathbb R}$ is smooth.
This proves bundle-regularity for the finite-dimensional faithful-state model, and hence for faithful qubits.

For the displaced thermal model, the Hilbert--Schmidt realization again identifies all GNS fibers with the fixed Hilbert space $\mathcal B_2(\mathcal H)$.
Thus the real dual bundle $E$ is a trivial smooth real Hilbert bundle.
The transported metric and complex structure are the fixed Hilbert--Schmidt ones, and the flat connection is compatible with them.
The explicit GNS representatives computed in the displaced thermal example depend smoothly on $(x,y,\beta)$ in Hilbert--Schmidt norm; therefore the real dual lift $L^{\mathbb R}$ is smooth.

Finally, bundle-regularity is preserved under restriction to immersed submanifolds: one pulls back the bundle, the metric, the complex structure, the connection, and the $E$-valued one-form $L^{\mathbb R}$.
Hence the fixed-temperature displaced thermal submodels are bundle-regular.
\end{proof}

\begin{remark}
For the commutative models considered in Section~\ref{sec:classical-case}, the induced two-form vanishes identically, $\Omega=0$.
Hence the closedness question on the parameter manifold is trivial.
In general, bundle-regularity is stronger than the dominated $2$-integrability assumptions used for the classical examples.
Those assumptions are sufficient to define the Fisher--Rao tensor and to prove $\Omega=0$, but they do not by themselves provide a smooth Hilbert-bundle structure on the family $L^2(X,\rho_m)$, nor smoothness of the canonical real dual lift as a bundle-valued one-form.
\end{remark}

\begin{remark}[Nondegenerate extensions and fatness]
One may also ask whether a closed extension can be chosen to be symplectic
on the total space $E$.  This is a different question from the closedness of
$\Omega$ on the parameter manifold.  In the Hilbert-manifold setting, the
word symplectic may mean either weak nondegeneracy or strong nondegeneracy;
in this remark we explicitly distinguish the two.

For the extension $\widetilde\Omega_\nabla$ above, one may add a closed base
form $\beta\in\Omega^2(M)$ and consider
\begin{equation}
\widetilde\Omega_{\nabla,\beta}
:=
\widetilde\Omega_\nabla+\pi_E^*\beta.
\end{equation}
With respect to the horizontal--vertical splitting determined by $\nabla$,
the vertical--vertical block is the strong symplectic form $\omega_m$, the
mixed block vanishes, and the horizontal--horizontal block at
$\alpha\in E_m$ is
\begin{equation}
\beta_m(u,v)
+
\frac12\,
\omega_m\bigl(\alpha,F_\nabla(u,v)\alpha\bigr).
\end{equation}
Consequently, $\widetilde\Omega_{\nabla,\beta}$ is weakly, respectively
strongly, symplectic precisely when this horizontal block is weakly,
respectively strongly, nondegenerate for every $\alpha\in E_m$.
This is the analogue, in the present vector-bundle setting, of Weinstein's
fatness criterion.

In particular, if one seeks a nondegenerate extension of this block-diagonal
type, the base must itself support a nondegenerate two-form of the
corresponding kind.  This explains why the issue is natural for pure-state
models and fixed-temperature displaced thermal submodels, but not for the
full faithful qubit model, whose parameter space is three-dimensional.
\end{remark}

\begin{remark}\label{rem:connection-to-other-covariances}
The construction of closed extensions is not tied to the GNS inner product alone.
It applies to any smoothly varying fiberwise Hermitian structure on a locally trivial Hilbert bundle over $M$.
In particular, replacing the GNS inner product by another field of covariances, for instance one associated with another operator monotone function in the Morozova--\v{C}encov--Petz classification, leads to the same total-space closed-extension problem.
The behaviour of the resulting base two-form is again governed by the covariant variation of the corresponding canonical lift.
\end{remark}

%%%%%%%%%%%%
%%%%%%%%%%%%

\section{Conclusions}

We have developed a GNS-based construction of geometric tensors on smooth
parametric models over $C^*$-algebras.  The central point is that the real
and imaginary parts of the dual GNS Hermitian product, organized fiberwise
over the state space and pulled back along suitable parametric models,
produce a Hermitian tensor on the complexified tangent bundle of the model.
Under the regularity assumptions of Definition~\ref{def:regular model},
its real and imaginary parts assemble into a smooth weak Riemannian metric
tensor $G$ and a smooth two-form $\Omega$.

This construction gives a common operator-algebraic origin for several
standard geometries of classical and quantum information theory.  In the
commutative dominated case it recovers the Fisher--Rao metric tensor and
the two-form vanishes.  For pure quantum states it recovers the
Fubini--Study metric tensor together with the Fubini--Study symplectic form,
up to the normalization and sign/order conventions imposed by the dual GNS
pairing.  For faithful quantum states it gives the SLD metric tensor
(equivalently, four times the Bures--Helstrom metric in the convention used
here), while the corresponding two-form is proportional, in finite
dimensions, to the expected commutator of the SLD representatives, or
equivalently to the mean Uhlmann curvature.

A distinctive feature of the construction is that it does not require the
state space $\mathcal S(\mathscr A)$ itself to be a smooth manifold.  The
usual ambient-manifold pullback picture is replaced by the GNS fibration
and its dual.  Tangent vectors to the parameter manifold are represented
by continuous functionals on the realified GNS fibers, and the ambiguity of
such representatives is removed by passing to the distinguished real
subspace generated by self-adjoint observables.  This produces the canonical
real dual GNS lift from which the tensors are defined.  In the faithful
quantum case, this viewpoint gives the SLD a natural geometric meaning:
it is not introduced as an external ansatz, but appears as the canonical
representative of a tangent direction.

The skew-symmetric tensor $\Omega$ is not merely a secondary byproduct of
the metric construction.  It detects genuinely noncommutative features of
the model, although its vanishing is weaker than pairwise commutativity of
the SLDs and should not be identified outright with the usual
quasiclassical condition.  The examples of faithful qubits and displaced
thermal states show that $\Omega$ need not be closed.  Section~\ref{sec:closed-extensions}
explains this phenomenon structurally: for bundle-regular models the
fiberwise symplectic form on the real dual GNS bundle admits
connection-dependent closed extensions to the total space, but the
closedness of the induced two-form on the parameter manifold is governed by
the covariant exterior derivative of the canonical real dual GNS lift.

\medskip

The framework developed in this paper suggests several natural directions for further investigation.

\medskip

\noindent\textit{Other fields of covariances and the Morozova--\v{C}encov--Petz classification.}
%%%%%%%%%%%%%%%%%%%%%%%%%%%%%%%%%%%%%%%%%%%%%%
As noted in the introduction, the GNS inner product is a particular instance of a \emph{field of covariances} in the sense of \cite{CDG2025}.
%%%%%%%%%%%%%%%%%%%%%%%%%%%%%%%%%%%%%
The finite-dimensional classification of fields of covariances in \cite{CDG2025} provides a unification of the classical \v{C}encov's result on the uniqueness of the Fisher--Rao metric tensor with the quantum Morozova--\v{C}encov--Petz characterization of monotone metric tensors that extends also to non-faithful states, thus applying to all states.
%%%%%%%%%%%%%%%%%%%%%%%%%%
It is natural to ask whether the pullback-like construction developed here can be carried out for the other fields of covariances, at least in the finite-dimensional case.
%%%%%%%%%%%%%%%%%%%%%%%%%%%%%%%%%%%%%%%%
The idea would be that each operator monotone function $f$ appearing in the classification in \cite{CDG2025} determines a different fiberwise inner product on a Hilbert fibration over the state space that is a sort of deformation of the original GNS one.
%%%%%%%%%%%%%%%%%%%%%%%%%%%%%%%
The same Riesz-representation procedure should produce a symmetric tensor $G^f$ and a skew-symmetric tensor $\Omega^f$ on any regular parametric model.
%%%%%%%%%%%%%%%%%%%%%%%%%%%%%%%%%%%%%%%
In the faithful finite-dimensional case, the symmetric part $G^f$ would recover the corresponding monotone metric in the Morozova--\v{C}encov--Petz family, while the skew-symmetric part $\Omega^f$ would be a new invariant depending on $f$.
%%%%%%%%%%%%%%%%%%%%%%%%%%%%%%%%%%%%%%%%%%
It is worth noting that the operator monotone functions in Petz's work \cite{P1996} satisfy a symmetricity condition that makes the imaginary part of the associated Hermitian product vanish, while the operator monotone functions appearing in \cite{CDG2025} do not necessarily satisfy this symmetricity condition, thus allowing a non-zero $\Omega^{f}$.
%%%%%%%%%%%%%%%%%%%%%%%%%%%%%%%%%%%%%%%%%%
Moreover, again at least in finite dimensions, the closed-extension machinery of Section~\ref{sec:closed-extensions} applies to any such Hilbert bundle.
%%%%%%%%%%%%%%%%%%%%%%%%%%%%%%%%%%%%%%%%%%%%%%%%%

\medskip
\noindent\textit{Behavior under quantum channels.}
%%%%%%%%%%%%%%%%%%%%%%%%%%%%%
A central property of the monotone quantum metrics in the Morozova--\v{C}encov--Petz classification is their monotonicity under completely positive trace-preserving maps, or equivalently under unital completely positive maps in the Heisenberg picture. 
%%%%%%%%%%%%%%%%%%%%%%%%%%%%%%%%%%%%
The present paper does not address this behavior. 
%%%%%%%%%%%%%%%%%%%%%%%%%%%%%%%%%%%%%%%%%%%%%%
For the metric part $G$, one may expect the pullback construction is compatible with the usual monotonicity of the SLD metric tensor under quantum channels, because quantum channels leads to contraction maps between GNS Hilbert spaces. 
%%%%%%%%%%%%%%%%%%%%%%%%%%%%%%%%%%%%%%%%%%%
For the skew-symmetric tensor $\Omega$, the appropriate statement is less clear, since a two-form is not ordered in the way a positive symmetric tensor is. 
%%%%%%%%%%%%%%%%%%%%%%%%%%%%%%%%%%%%%%%%%
A more natural problem is to formulate a functoriality, covariance, or contraction property for $\Omega$ under suitable state-preserving completely positive maps, and to understand how such a property interacts with quantum Cram\'er--Rao attainability and optimal POVMs.
%%%%%%%%%%%%%%%%%%%%%%%%%%%%%%%%%%%%%%%%%%%%%%%%%%%

\medskip
\noindent\textit{Nonparametric extensions and the work of Jen\v{c}ov\'a.}
%%%%%%%%%%%%%%%%%%%%%%%%%%%%%%%%%%%%%
The set of faithful normal states $\mathcal{S}_{nf}(\mathcal{M})$ on a von Neumann algebra $\mathcal{M}$ can be endowed with the structure of Banach manifold using quantum Orlicz spaces \cite{J2006,J2024},  providing an infinite-dimensional analogue of the Pistone--Sempi exponential manifold \cite{PS1995}.
%%%%%%%%%%%%%%%%%%%%%%%%%%%%%%%%%
This construction is based on relative entropy and state perturbation, and admits exponential and mixture connections as a dual pair.
%%%%%%%%%%%%%%%%%%%%%%%%%%%%%%%%%%%%%%%%%%%%%%%%%%
If $\mathrm{i}$ is the canonical inclusion of $\mathcal{S}_{nf}(\mathcal{M})$ in $\mathcal{S}(\mathcal{M})$, it would be of considerable interest to understand if $(\mathcal{S}_{nf}(\mathcal{M}),\mathrm{i},\mathcal{M})$ is regular in the sense of Definition \ref{def:regular model}.
%%%%%%%%%%%%%%%%%%%%%%%%%%%%%%%%%%%%
Furthermore, the manifold $\mathcal{S}_{nf}(\mathcal{M})$  is endowed with a canonical divergence satisfying a Pythagorean relation and is invariant under sufficient channels \cite{J2024}, which suggests a possible bridge to the monotonicity questions mentioned above.
%%%%%%%%%%%%%%%%%%%%%%%%%%%%%%%%%%%%%%%%%%%%%%

\medskip
\noindent\textit{The $2$-form $\Omega$ and quantum estimation theory.}
The definition of the skew-symmetric tensor $\Omega$ from the dual GNS Hermitian product is a novel construction of the paper.
%%%%%%%%%%%%%%%%%%%%%%%%%%%%%%%%%%%%%%%%%%%%
Its vanishing on quasiclassical models and its relation to the expected commutator of SLD operators (mean Uhlmann curvature) for finite-dimensional faithful quantum states suggest that it is related to an intrinsic obstruction to simultaneous quantum parameter estimation.
%%%%%%%%%%%%%%%%%%%%%%%%%%%%%%%%%%%%%%%%%%%%%%%%%
A precise formulation of this connection would require relating $\Omega$ to the Holevo bound or to the incompatibility of optimal measurements in multiparameter quantum estimation \cite{H2011,S2019}.
%%%%%%%%%%%%%%%%%%%%%%%%%%%%%%%%%
Making this link explicit, and understanding how the closedness properties of $\Omega$ studied in Section~\ref{sec:closed-extensions} bear on the estimation-theoretic picture, is an interesting open problem at the interface of differential geometry and quantum statistics.
%%%%%%%%%

\section*{Acknowledgements}

This work was partially supported by the Spanish Ministry of Economy and Competitiveness through the research projects PID2024-160539NB-I00 and  PID2024-156578NB-I00, and by the Madrid Government through the project TEC-2024/COM-84 QUITEMAD-CM. 
%%%%%%%%%%%%%%%%%%%%%%%%
%%%%%%%%%%%%%%%%%%%%%%%%%%%%%%%%%
A.I. acknowledges financial support from the Spanish Ministry of Economy and Competitiveness through the Severo Ochoa Program for Centers of Excellence in RD (SEV-2015/0554).
%%%%%%%%%%%%%%%%%%%%%%%%%%%%%%%%%%%%%%%%%%%%%%%%%%
This article is based upon work from COST Action CaLISTA CA21109, supported by COST (European Cooperation in Science and Technology). 
%%%%%%%%%%%%%%%%%%%%%%%%
F.M.C., A.I., and L.G.-B. acknowledge their participation in COST Action CaLISTA CA21109.
%%%%%%%%%%

\addcontentsline{toc}{section}{References}

%\bibliographystyle{plainurl}
%{\footnotesize
%\bibliography{biblio.bib}
%}
 
{\footnotesize
\printbibliography[title=References]
}
\end{document}